\documentclass{article}
\usepackage{amsmath,amssymb,amsfonts}

\def\op#1{\mathop{{\it\fam0} #1}\limits}

\newcommand{\beq}{\begin{equation}}
\newcommand{\eeq}{\end{equation}}
\newcommand{\ben}{\begin{eqnarray}}
\newcommand{\een}{\end{eqnarray}}
\newcommand{\be}{\begin{eqnarray*}}
\newcommand{\ee}{\end{eqnarray*}}
\newcommand{\bea}{\begin{eqalph}}
\newcommand{\eea}{\end{eqalph}}

\newcommand{\di}{{\mathrm {dim}\,}}

\newcommand{\hm}{{\mathrm{Hom}\,}}

\newcommand{\im}{{\mathrm {Im}\, }}
\newcommand{\llr}{\op\longleftarrow}

\newcommand{\al}{\alpha}
\newcommand{\bt}{\beta}
\newcommand{\dl}{\delta}
\newcommand{\la}{\lambda}
\newcommand{\La}{\Lambda}
\newcommand{\f}{\phi}
\newcommand{\vf}{\varphi}

\newcommand{\p}{\pi}

\newcommand{\om}{\omega}

\newcommand{\m}{\mu}
\newcommand{\n}{\nu}
\newcommand{\g}{\gamma}
\newcommand{\G}{\Gamma}
\newcommand{\kp}{\kappa}

\newcommand{\vr}{\varrho}
\newcommand{\thh}{\theta}
\newcommand{\vt}{\vartheta}
\newcommand{\cG}{{\mathfrak g}}
\newcommand{\ve}{\varepsilon}
\newcommand{\up}{\upsilon}
\newcommand{\e}{\epsilon}
\newcommand{\ap}{\approx}
\newcommand{\rdr}{\stackrel{\leftarrow}{\dr}{}}

\newcommand{\bll}{\bullet}
\newcommand{\bbc}{{\bf b}}

\newcommand{\nw}[1]{[{#1}]}
\newcommand{\nm}[1]{|{#1}|}

\newcommand{\id}{{\mathrm{Id}\,}}

\newcommand{\si}{\sigma}
\newcommand{\Si}{\Sigma}

\newcommand{\lto}{{\leftarrow}}

\newcommand{\cO}{{\mathcal O}}
\newcommand{\cA}{{\mathcal A}}

\newcommand{\cJ}{{\mathcal J}}
\newcommand{\cR}{{\mathcal R}}

\newcommand{\gL}{{\mathfrak L}}

\newcommand{\gd}{{\mathfrak d}}

\newcommand{\gS}{{\mathfrak S}}

\newcommand{\gQ}{{\mathfrak Q}}

\newcommand{\gA}{{\mathfrak A}}

\newcommand{\cP}{{\mathcal P}}
\newcommand{\cL}{{\mathcal L}}
\newcommand{\cV}{{\mathcal V}}

\newcommand{\cQ}{{\mathcal Q}}
\newcommand{\cE}{{\mathcal E}}

\newcommand{\cF}{{\mathcal F}}
\newcommand{\cC}{{\mathcal C}}

\newcommand{\cK}{{\mathcal K}}

\newcommand{\ccG}{{\mathcal G}}

\newcommand{\bu}{{\mathbf u}}

\newcommand{\cS}{{\mathcal S}}

\newcommand{\bL}{{\mathbf L}}

\newcommand{\bE}{{\mathbf E}}

\newcommand{\bb}{{\mathbf 1}}

\newcommand{\w}{\wedge}

\newcommand{\wt}{\widetilde}
\newcommand{\wh}{\widehat}
\newcommand{\ol}{\overline}

\newcommand{\dr}{\partial}

\newcommand{\ar}{\op\longrightarrow}

\newcommand{\ot}{\otimes}

\let\ssection=\section
\renewcommand{\section}{\setcounter{equation}{0}\ssection}

\newenvironment{eqalph}{\stepcounter{equation}
\setcounter{equationa}{\value{equation}} \setcounter{equation}{0}
\begin{eqnarray}}{\end{eqnarray}\setcounter{equation}{\value{equationa}}}

\newcounter{equationa}[section]
\newcounter{remark}[section]
\newcounter{example}[section]
\newcounter{theorem}[section]
\newcounter{condition}[section]
\newcounter{lemma}[section]
\newcounter{corollary}[section]
\newcounter{definition}[section]
\def\theremark{\arabic{section}.\arabic{remark}}

\def\thedefinition{\arabic{section}.\arabic{theorem}}

\newenvironment{proof}{{\it Proof.}}{\hfill $\Box$
\medskip }
\newenvironment{remark}{\refstepcounter{remark} \medskip {\bf Remark
\theremark.} }{ \medskip }
\newenvironment{example}{\refstepcounter{remark} \medskip {\bf
Example \theremark.} }{ \medskip }
\newenvironment{theorem}{\refstepcounter{theorem} \medskip{\bf
Theorem \thedefinition.}\it }{ \medskip }
\newenvironment{condition}{\refstepcounter{theorem} \medskip{\bf
Condition \thedefinition.}\it }{\medskip }
\newenvironment{lemma}{\refstepcounter{theorem} \medskip{\bf  Lemma
\thedefinition.}\it}{\medskip }
\newenvironment{corollary}{\refstepcounter{theorem} \medskip{\bf
Corollary \thedefinition.} \it}{ \medskip }
\newenvironment{definition}{\refstepcounter{theorem} \medskip{\bf
Definition \thedefinition.} \it}{ \medskip }

\newcommand{\mar}[1]{}

\begin{document}

\hbox{}

\begin{center}

{\Large\bf Noether theorems in a general setting}

\bigskip

G. Sardanashvily

\medskip

Department of Theoretical Physics, Moscow State University,
Moscow, Russia

\bigskip

\end{center}

\begin{abstract}
The first and second Noether theorems are formulated in a general
case of reducible degenerate Grassmann-graded Lagrangian theory of
even and odd variables on graded bundles. Such Lagrangian theory
is characterized by a hierarchy of non-trivial higher-stage
Noether identities and the corresponding higher-stage gauge
symmetries which are described in the homology terms. In these
terms, the second Noether theorems associate to the Koszul -- Tate
chain complex of higher-stage Noether identities the gauge cochain
sequence whose ascent operator provides higher-order gauge
symmetries of Lagrangian theory. If gauge symmetries are
algebraically closed, this operator is extended to the nilpotent
BRST operator which brings the gauge cochain sequence into the
BRST complex. In this framework, the first Noether theorem is
formulated as a straightforward corollary of the first variational
formula. It associates to any variational Lagrangian symmetry the
conserved current whose total differential vanishes on-shell. We
prove in a general setting that a conserved current of a gauge
symmetry is reduced to a total differential on-shell. The
physically relevant examples of gauge theory on principal bundles,
gauge gravitational theory on natural bundles, topological Chern
-- Simons field theory and topological BF theory are present. The
last one exemplifies a reducible Lagrangian system.
\end{abstract}

\tableofcontents

\section{Introduction}

The Noether theorems are well known to treat symmetries of
Lagrangian systems \cite{KS}. The first Noether theorem associates
to a Lagrangian symmetry the conserved current whose total
differential vanishes on-shell. The second ones provide the
correspondence between Noether identities and gauge symmetries of
a Lagrangian system.

We aim to formulate Noether theorems in a general case of
reducible degenerate Lagrangian systems characterized by a
hierarchy of non-trivial higher-stage Noether identities (Section
5.1). To describe this hierarchy, one need to involve
Grassmann-graded variables. In a general setting, we therefore
consider Grassmann-graded Lagrangian systems of even and odd
variables on a smooth manifold $X$ (Section 3).

Lagrangian theory of even (commutative) variables on an
$n$-dimensional smooth manifold $X$ conventionally is formulated
in terms of smooth fibre bundles over $X$ and jet manifolds of
their sections \cite{bau,book,olv,book13,tak2} in the framework of
general technique of non-linear differential operators and
equations \cite{bry,book,kras}. This formulation is based on the
categorial equivalence of projective $C^\infty(X)$-modules of
finite ranks and vector bundles over $X$ in accordance with the
classical Serre -- Swan theorem, generalized to non-compact
manifolds \cite{book09,ren,sard01}.

At the same time, different geometric models of odd variables
either on graded manifolds or supermanifolds are discussed
\cite{cari03,cia95,franc,mont92,mont,sard13}. Both graded
manifolds and supermanifolds are phrased in terms of sheaves of
graded commutative algebras \cite{bart,book09,sard09}. However,
graded manifolds are characterized by sheaves on smooth manifolds,
while supermanifolds are constructed by gluing of sheaves on
supervector spaces. Since non-trivial higher-stage Noether
identities of a Lagrangian system on a smooth manifold $X$ form
graded $C^\infty(X)$-modules, we follow the above mentioned Serre
-- Swan theorem extended to graded manifolds (Theorem \ref{vv0})
\cite{sard13,SS}. It states that, if a graded commutative
$C^\infty(X)$-ring is generated by a projective
$C^\infty(X)$-module of finite rank, it is isomorphic to a ring of
graded functions on a graded manifold whose body is $X$.
Accordingly, we describe odd variables in terms of graded
manifolds too \cite{book09,book13,sard13}.

Let us recall that a graded manifold is a locally-ringed space,
characterized by a smooth body manifold $Z$ and some structure
sheaf $\gA$ of Grassmann algebras on $Z$
\cite{bart,book09,sard09}.  Its sections form a graded commutative
$C^\infty(Z)$-ring $\cA$ of graded functions on a graded manifold
$(Z,\gA)$. It is called the structure ring of $(Z,\gA)$. The
differential calculus on a graded manifold is defined as the
Chevalley -- Eilenberg differential calculus over its structure
ring (Section 2.3). By virtue of the well-known Batchelor theorem
(Theorem \ref{lmp1a}), there exists a vector bundle $E\to Z$ with
a typical fibre $V$ such that the structure sheaf $\gA$ of
$(Z,\gA)$ is isomorphic to a sheaf $\gA_E$ of germs of sections of
the exterior bundle $\w E^*$ of the dual $E^*$ of $E$ whose
typical fibre is the Grassmann algebra $\w V^*$
\cite{bart,batch1}. This Batchelor's isomorphism is not canonical.
In applications, it however is fixed from the beginning.
Therefore, we restrict our consideration to graded manifolds
$(Z,\gA_E)$, called the simple graded manifolds, modelled over
vector bundles $E\to Z$.

Let us note that a smooth manifold $Z$ itself can be treated as a
trivial simple graded manifold $(Z,C^\infty_Z)$ modelled over a
trivial bundle $Z\times\mathbb R\to Z$ whose structure ring of
graded functions is reduced to a ring $C^\infty(Z)$ of smooth real
functions on $Z$ (Example \ref{triv}). Accordingly, a fibre bundle
$Y\to X$ in a Lagrangian theory of even even variables can be
regarded as a graded bundle of trivial graded manifolds $(Y,
C^\infty_Y) \to (X, C^\infty_X)$ (Example \ref{su20}). It follows
that, in a general setting, one can define a configuration space
of Grassmann-graded Lagrangian theory of even and odd variables as
being a graded bundle
\mar{su11}\beq
(Y, \gA_F) \to (X, C^\infty_X) \label{su11}
\eeq
over a trivial graded manifold $(X, C^\infty_X)$ (Section 2.4)
where $(Y, \gA_F)$ is a simple graded manifold modelled over a
vector bundle $F\to Y$ whose body is a smooth bundle $Y\to X$
\cite{book09,sard13,sard14}. If $Y\to X$ is a vector bundle, this
is a particular case of graded vector bundles in \cite{hern,mont}
whose base is a trivial graded manifold.

Lagrangian theory on a fibre bundles $Y\to X$ can be adequately
formulated in algebraic terms of a variational bicomplex of
exterior forms on the infinite order jet manifold $J^\infty Y$ of
sections of $Y\to X$, without appealing to the calculus of
variations \cite{ander,bau,jmp,book09,olv,book13,tak2}. This
technique is extended to Lagrangian theory on graded bundles
\cite{barn,jmp05a,lmp08,cmp04,book09,sard13}. It is
comprehensively phrased in terms of the Grassmann-graded
variational bicomplex (\ref{7}) of graded exterior forms on a
graded infinite order jet manifold $(J^\infty Y, \cA_{J^\infty
F})$ (Section 3). Lagrangians and the Euler -- Lagrange operator
are defined as the elements (\ref{0709}) and the coboundary
operator (\ref{0709'}) of this bicomplex, respectively. The
cohomology of the variational bicomplex provides the global
variational decomposition (\ref{g99}) for Lagrangians and Euler --
Lagrange operators (Theorem \ref{g103}).

In these terms, the first Noether theorem is formulated as a
straightforward corollary of the variational decomposition
(\ref{g99}) (Section 4). It associates to any variational symmetry
of a Lagrangian $L$ the conserved current (\ref{09200}) whose
total differential vanishes on-shell (Theorem \ref{j45}). It is
important that, in the case of a gauge symmetry, the corresponding
conserved current is reduced to a superpotential, i.e., this is a
total differential on-shell (Theorem \ref{supp}). This fact was
proved in different particular variants
\cite{fat94,book09,got92,julia}. We do this in a very general
setting \cite{gauge09}.

Given a gauge symmetry of a graded Lagrangian, the direct second
Noether theorem (Theorem \ref{825}) states that the Euler --
Lagrange operator obeys the corresponding Noether identities. A
problem is that any Euler -- Lagrange operator satisfies Noether
identities, which therefore must be separated into the trivial and
non-trivial ones. These Noether identities can obey first-stage
Noether identities, which in turn are subject to the second-stage
ones, and so on. Thus, there is a hierarchy of Noether and
higher-stage Noether identities which also must be separated into
the trivial and non-trivial ones. In accordance with general
analysis of Noether identities of differential operators
\cite{book09,oper}, if certain homology regularity conditions hold
(Condition \ref{v155}), one can associate to a Grassmann-graded
Lagrangian system the exact Koszul -- Tate complex (\ref{v94})
possessing the boundary operator whose nilpotentness is equivalent
to all complete non-trivial Noether and higher-stage Noether
identities (\ref{v64}) and (\ref{v93})
\cite{jmp05a,lmp08,book09,sard13}.

It should be noted that the notion of higher-stage Noether
identities has come from that of reducible constraints. The Koszul
-- Tate complex of Noether identities has been invented similarly
to that of constraints under the condition that Noether identities
are locally separated into independent and dependent ones
\cite{barn,fisch}. This condition is relevant for constraints,
defined by a finite set of functions which the inverse mapping
theorem is applied to. However, Noether identities unlike
constraints are differential equations. They are given by an
infinite set of functions on a Fr\'echet manifold of infinite
order jets where the inverse mapping theorem fails to be valid.
Therefore, the regularity condition for the Koszul -- Tate complex
of constraints is replaced with the above mentioned homology
regularity condition.

The inverse second Noether theorem formulated in homology terms
(Theorem \ref{w35}) associates to this Koszul -- Tate complex the
cochain sequence (\ref{w108}) with the ascent operator
(\ref{w108'}), called the gauge operator, whose components are
complete non-trivial gauge and higher-stage gauge symmetries of
Lagrangian field theory \cite{lmp08,book09,sard13}.

The gauge operator unlike the Koszul -- Tate one is not nilpotent,
unless gauge symmetries are abelian. This is the cause why an
intrinsic definition of non-trivial gauge and higher-stage gauge
symmetries meets difficulties. Another problem is that gauge
symmetries need not form an algebra \cite{fulp02,jmp09,gom}.
Therefore, we replace the notion of the algebra of gauge
symmetries with some conditions on the gauge operator. Gauge
symmetries are said to be algebraically closed if the gauge
operator admits a nilpotent extension, called the BRST (Becchi --
Rouet -- Stora -- Tyitin) operator (Section 5.4). If the BRST
operator exists, the above mentioned cochain sequence is brought
into the BRST complex. The Koszul -- Tate and BRST complexes
provide a BRST extension of original Lagrangian theory by
Grassmann-graded ghosts and Noether antifields
\cite{jmp09,book09,sard13}.

Classical field theory is formulated adequately as a Lagrangian
theory on graded bundles \cite{book09,sard08,sard13}. Section 6
contains some examples of relevant field models: gauge theory on
principal bundles, gravitation theory on natural bundles, Chern --
Simons topological theory, topological BF theory. The last one
exemplifies a reducible Lagrangian system.

\section{Graded bundles}

Throughout this work, by the Grassmann gradation is meant the
$\mathbb Z_2$-one, and a Grassmann graded structure simply is
called the graded one if there is no danger of confusion.
Hereafter, the symbol $\nw .$ stands for the Grassmann parity.

Smooth manifolds throughout are assumed to be Hausdorff,
second-countable and, consequently, paracompact and locally
compact, countable at infinity. Given a smooth manifold $X$, its
tangent and cotangent bundles $TX$ and $T^*X$ are endowed with
bundle coordinates $(x^\la,\dot x^\la)$ and $(x^\la,\dot x_\la)$
with respect to holonomic frames $\{\dr_\la\}$ and $\{dx^\la\}$,
respectively.

Given a manifold $X$ and its coordinate chart $(U; x^\la)$, a
multi-index $\La$ of the length $|\La|=k$ throughout denotes a
collection of indices $(\la_1...\la_k)$ modulo permutations. By
$\la+\La$ is meant a multi-index $(\la\la_1\ldots\la_k)$.
Summation over a multi-index $\La$ means separate summation over
each its index $\la_i$. We use the compact notation
$\dr_\La=\dr_{\la_k}\circ\cdots\circ\dr_{\la_1}$ and
$\La=(\la_1...\la_k)$.

\subsection{Grassmann-graded algebraic calculus}

Let us summarize the relevant basics of the Grassmann-graded
algebraic calculus \cite{bart,book09,sard09}.

Let $\cK$ be a commutative ring. A $\cK$-module $Q$ is called
graded if it is endowed with a grading automorphism $\g$,
$\g^2=\id$. A graded module falls into a direct sum of modules
$Q=Q_0 \oplus Q_1$ such that $\g(q)=(-1)^{[q]}q$, $q\in Q_{[q]}$.
One calls $Q_0$ and $Q_1$ the even and odd parts of $Q$,
respectively.

In particular, by a real graded vector space $B=B_0\oplus B_1$ is
meant a graded $\mathbb R$-module. It is said to be
$(n,m)$-dimensional if $B_0=\mathbb R^n$ and $B_1=\mathbb R^m$.

A $\cK$-algebra $\cA$ is called graded if it is a graded
$\cK$-module such that $[aa']=([a]+[a']){\rm mod}\,2$, where $a$
and $a'$ are graded-homogeneous elements of $\cA$. Its even part
$\cA_0$ is a subalgebra of $\cA$, and the odd one $\cA_1$ is an
$\cA_0$-module. If $\cA$ is a graded ring with the unit $\bb$,
then $[\bb]=0$.

A graded algebra $\cA$ is called graded commutative if
$aa'=(-1)^{[a][a']}a'a$.

\begin{example} \label{grass} \mar{grass}
Let $V$ be a real vector space, and let $\La=\w V$ be its exterior
algebra endowed with the Grassmann gradation
\mar{+66}\beq
\La=\La_0\oplus \La_1, \qquad \La_0=\mathbb R\op\bigoplus_{k=1}
\op\w^{2k} V, \qquad \La_1=\op\bigoplus_{k=1} \op\w^{2k-1} V.
\label{+66}
\eeq
It is a real graded commutative ring, called the Grassmann
algebra. A Grassmann algebra, seen as an additive group, admits
the decomposition
\mar{+11}\beq
\La=\mathbb R\oplus R =\mathbb R\oplus R_0\oplus R_1=\mathbb R
\oplus (\La_1)^2 \oplus \La_1, \label{+11}
\eeq
where $R$ is the ideal of nilpotents of $\La$. The corresponding
epimorphism $\si:\La\to\mathbb R$ is called the body map. Note
that there is a different definition of a Grassmann algebra
\cite{jad}. Hereafter, we restrict our consideration to Grassmann
algebras of finite rank when $V=\mathbb R^N$. Given a basis
$\{c^i\}$ for $V$, elements of the Grassmann algebra $\La$
(\ref{+66}) take a form
\mar{z784}\beq
a=\op\sum_{k=0,1,\ldots} \op\sum_{(i_1\cdots i_k)}a_{i_1\cdots
i_k}c^{i_1}\cdots c^{i_k}, \label{z784}
\eeq
where the second sum runs through all the tuples $(i_1\cdots i_k)$
such that no two of them are permutations of each other.
\end{example}

Given a graded algebra $\cA$, a left graded $\cA$-module $Q$ is
defined as a left $\cA$-module where $[aq]=([a]+[q]){\rm mod}\,2$.
Similarly, right graded $\cA$-modules are treated.

\begin{example} \label{lie} \mar{lie}
A graded algebra $\cG$ is called a Lie superalgebra if its product
$[.,.]$, called the Lie superbracket, obeys the relations
\be
&& [\ve,\ve']=-(-1)^{[\ve][\ve']}[\ve',\ve],\\
&& (-1)^{[\ve][\ve'']}[\ve,[\ve',\ve'']]
+(-1)^{[\ve'][\ve]}[\ve',[\ve'',\ve]] +
(-1)^{[\ve''][\ve']}[\ve'',[\ve,\ve']] =0.
\ee
Being decomposed in even and odd parts $\cG=\cG_0\oplus \cG_1$, a
Lie superalgebra $\cG$ obeys the relations
\be
[{\cG_0},{\cG_0}]\subset \cG_0, \qquad [{\cG_0},{\cG_1}]\subset
\cG_1, \qquad [{\cG_1},{\cG_1}]\subset \cG_1.
\ee
In particular, an even part $\cG_0$ of a Lie superalgebra $\cG$ is
a Lie algebra. A graded vector space $P$ is a $\cG$-module if it
is provided with an $\mathbb R$-bilinear map
\be
&& \cG\times P\ni (\ve,p)\to \ve p\in P, \qquad [\ve
p]=([\ve]+[p]){\rm mod}\,2,\\
&& [\ve,\ve']p=(\ve\circ\ve'-(-1)^{[\ve][\ve']}\ve'\circ\ve)p.
\ee
\end{example}

If $\cA$ is graded commutative, a graded $\cA$-module $Q$ is
provided with a graded $\cA$-bimodule structure by letting $qa =
(-1)^{[a][q]}aq$, $a\in\cA$, $q\in Q$.

Given a graded commutative ring $\cA$, the following are standard
constructions of new graded modules from the old ones.

$\bullet$ The direct sum of graded modules and a graded factor
module are defined just as those of modules over a commutative
ring.

$\bullet$ The tensor product $P\ot Q$ of graded $\cA$-modules $P$
and $Q$ is their tensor product as $\cA$-modules such that
\be
&& [p\ot q]=[p]+ [q], \qquad p\in P, \qquad q\in Q, \\
&&  ap\ot q=(-1)^{[p][a]}pa\ot q= (-1)^{[p][a]}p\ot aq,  \qquad a\in\cA.
\ee
In particular, the tensor algebra $\ot P$ of a graded $\cA$-module
$P$ is defined just as that of a module over a commutative ring.
Its quotient $\w P$ with respect to the ideal generated by
elements
\be
p\ot p' + (-1)^{[p][p']}p'\ot p, \qquad p,p'\in P,
\ee
is a bigraded exterior algebra of a graded module $P$ provided
with the graded exterior product
\be
p\w p' =- (-1)^{[p][p']}p'\w p.
\ee

$\bullet$ A morphism $\Phi:P\to Q$ of graded $\cA$-modules seen as
additive groups is said to be an even (resp. odd) graded morphism
if $\Phi$ preserves (resp. changes) the Grassmann parity of all
graded-homogeneous elements of $P$, and if the relations
\be
\Phi(ap)=(-1)^{[\Phi][a]}a\Phi(p), \qquad p\in P, \qquad a\in\cA,
\ee
hold. A morphism $\Phi:P\to Q$ of graded $\cA$-modules as additive
groups is called a graded $\cA$-module morphism if it is
represented by a sum of even and odd graded morphisms. A set
$\hm_\cA(P,Q)$ of graded morphisms of a graded $\cA$-module $P$ to
a graded $\cA$-module $Q$ is naturally a graded $\cA$-module. A
graded $\cA$-module $P^*=\hm_\cA(P,\cA)$ is called the dual of a
graded $\cA$-module $P$.

\subsection{Grassmann-graded differential calculus}

Linear differential operators and the differential calculus over a
graded commutative ring are defined similarly to those in
commutative geometry \cite{book09,sard09,book12}.

Let $\cK$ be a commutative ring and $\cA$ a graded commutative
$\cK$-ring. Let $P$ and $Q$ be graded $\cA$-modules. A
$\cK$-module $\hm_\cK (P,Q)$ of graded $\cK$-module homomorphisms
$\Phi:P\to Q$ can be endowed with the two graded $\cA$-module
structures
\be
(a\Phi)(p)= a\Phi(p),  \qquad  (\Phi\bll a)(p) = \Phi (a p),\qquad
a\in \cA, \quad p\in P.
\ee
Let us put
\mar{ws12}\beq
\dl_a\Phi= a\Phi -(-1)^{[a][\Phi]}\Phi\bll a, \qquad a\in\cA.
\label{ws12}
\eeq
An element $\Delta\in\hm_\cK(P,Q)$ is said to be a $Q$-valued
graded differential operator of order $s$ on $P$ if
$\dl_{a_0}\circ\cdots\circ\dl_{a_s}\Delta=0$ for any tuple of
$s+1$ elements $a_0,\ldots,a_s$ of $\cA$.

In particular, zero order graded differential operators coincide
with graded $\cA$-module morphisms $P\to Q$. A first order graded
differential operator $\Delta$ satisfies a relation
\be
&& \dl_a\circ\dl_b\,\Delta(p)=
ab\Delta(p)- (-1)^{([b]+[\Delta])[a]}b\Delta(ap)-
(-1)^{[b][\Delta]}a\Delta(bp)+\\
&& \qquad (-1)^{[b][\Delta]+([\Delta]+[b])[a]}
=0, \qquad a,b\in\cA, \quad p\in P.
\ee

For instance, let $P=\cA$. Any zero order $Q$-valued graded
differential operator $\Delta$ on $\cA$ is defined by its value
$\Delta(\bb)$.  A first order $Q$-valued graded differential
operator $\Delta$ on $\cA$ fulfils a condition
\be
\Delta(ab)= \Delta(a)b+ (-1)^{[a][\Delta]}a\Delta(b)
-(-1)^{([b]+[a])[\Delta]} ab \Delta(\bb), \qquad  a,b\in\cA.
\ee
It is called the $Q$-valued graded derivation of $\cA$ if
$\Delta(\bb)=0$, i.e., the graded Leibniz rule
\mar{ws10}\beq
\Delta(ab) = \Delta(a)b + (-1)^{[a][\Delta]}a\Delta(b), \quad
a,b\in \cA, \label{ws10}
\eeq
holds. One then observes that any first order graded differential
operator on $\cA$ falls into a sum
\be
\Delta(a)= \Delta(\bb)a +[\Delta(a)-\Delta(\bb)a]
\ee
of a zero order graded differential operator $\Delta(\bb)a$ and a
graded derivation $\Delta(a)-\Delta(\bb)a$. If $\dr$ is a graded
derivation of $\cA$, then $a\dr$ is so for any $a\in \cA$. Hence,
graded derivations of $\cA$ constitute a graded $\cA$-module
$\gd(\cA,Q)$, called the graded derivation module. If $Q=\cA$, a
graded derivation module $\gd\cA$ also is a Lie superalgebra over
a commutative ring $\cK$ with respect to a superbracket
\mar{ws14}\beq
[u,u']=u\circ u' - (-1)^{[u][u']}u'\circ u, \qquad u,u'\in \cA.
\label{ws14}
\eeq

Since $\gd\cA$ is a Lie $\cK$-superalgebra, let us consider the
Chevalley -- Eilenberg complex $C^*[\gd\cA;\cA]$ where a graded
commutative ring $\cA$ is a regarded as a $\gd\cA$-module
\cite{fuks,book09,book12}. It is the complex
\mar{ws85}\beq
0\to \cK\to \cA\ar^d C^1[\gd\cA;\cA]\ar^d \cdots
C^k[\gd\cA;\cA]\ar^d\cdots \label{ws85}
\eeq
where
\be
C^k[\gd\cA;\cA]=\hm_\cK(\op\w^k \gd\cA,\cA)
\ee
are $\gd\cA$-modules of $\cK$-linear graded morphisms of graded
exterior products $\op\w^k \gd\cA$ of a graded $\cK$-module
$\gd\cA$ to $\cA$. Let us bring homogeneous elements of $\op\w^k
\gd\cA$ into the form
\be
\ve_1\w\cdots\ve_r\w\e_{r+1}\w\cdots\w \e_k, \qquad
\ve_i\in\gd\cA_0, \quad \e_j\in\gd\cA_1.
\ee
Then the Chevalley -- Eilenberg coboundary operator $d$ of the
complex (\ref{ws85}) is given by the expression
\mar{ws86}\ben
&& dc(\ve_1\w\cdots\w\ve_r\w\e_1\w\cdots\w\e_s)=
\label{ws86}\\
&&\op\sum_{i=1}^r (-1)^{i-1}\ve_i
c(\ve_1\w\cdots\wh\ve_i\cdots\w\ve_r\w\e_1\w\cdots\e_s)+
\nonumber \\
&& \op\sum_{j=1}^s (-1)^r\ve_i
c(\ve_1\w\cdots\w\ve_r\w\e_1\w\cdots\wh\e_j\cdots\w\e_s)
+\nonumber\\
&& \op\sum_{1\leq i<j\leq r} (-1)^{i+j}
c([\ve_i,\ve_j]\w\ve_1\w\cdots\wh\ve_i\cdots\wh\ve_j
\cdots\w\ve_r\w\e_1\w\cdots\w\e_s)+\nonumber\\
&&\op\sum_{1\leq i<j\leq s} c([\e_i,\e_j]\w\ve_1\w\cdots\w
\ve_r\w\e_1\w\cdots
\wh\e_i\cdots\wh\e_j\cdots\w\e_s)+\nonumber\\
&& \op\sum_{1\leq i<r,1\leq j\leq s} (-1)^{i+r+1}
c([\ve_i,\e_j]\w\ve_1\w\cdots\wh\ve_i\cdots\w\ve_r\w
\e_1\w\cdots\wh\e_j\cdots\w\e_s),\nonumber
\een
where the caret $\,\wh{}\,$ denotes omission.

It is easily justified that the complex (\ref{ws85}) contains a
subcomplex $\cO^*[\gd\cA]$ of $\cA$-linear graded morphisms. The
$\mathbb N$-graded module $\cO^*[\gd\cA]$ is provided with the
structure of a bigraded $\cA$-algebra with respect to the graded
exterior product
\mar{ws103'}\ben
&& \f\w\f'(u_1,...,u_{r+s})= \op\sum_{i_1<\cdots<i_r;j_1<\cdots<j_s} {\rm
Sgn}^{i_1\cdots i_rj_1\cdots j_s}_{1\cdots r+s} \f(u_{i_1},\ldots,
u_{i_r}) \f'(u_{j_1},\ldots,u_{j_s}), \label{ws103'} \\
&& \f\in \cO^r[\gd\cA], \qquad \f'\in \cO^s[\gd\cA], \qquad u_k\in \gd\cA,
\nonumber
\een
where $u_1,\ldots, u_{r+s}$ are graded-homogeneous elements of
$\gd\cA$ and
\be
u_1\w\cdots \w u_{r+s}= {\rm Sgn}^{i_1\cdots i_rj_1\cdots
j_s}_{1\cdots r+s} u_{i_1}\w\cdots\w u_{i_r}\w u_{j_1}\w\cdots\w
u_{j_s}.
\ee
The graded Chevalley -- Eilenberg coboundary operator $d$
(\ref{ws86}) and the graded exterior product $\w$ (\ref{ws103'})
bring $\cO^*[\gd\cA]$ into a differential bigraded algebra
(henceforth DBGA) whose elements obey relations
\mar{ws45}\beq
\f\w \f'=(-1)^{|\f||\f'|+[\f][\f']}\f'\w\f, \qquad  d(\f\w\f')=
d\f\w\f' +(-1)^{|\f|}\f\w d\f'. \label{ws45}
\eeq
It is called the graded differential calculus over a graded
commutative $\cK$-ring $\cA$. In particular, we have
\mar{ws47}\beq
\cO^1[\gd\cA]=\hm_\cA(\gd\cA,\cA)=\gd\cA^*. \label{ws47}
\eeq
One can extend this duality relation to the graded interior
product of $u\in\gd\cA$ with any element $\f\in \cO^*[\gd\cA]$ by
the rules
\mar{ws46}\ben
&& u\rfloor(bda) =(-1)^{[u][b]}bu(a),\qquad a,b \in\cA, \nonumber\\
&& u\rfloor(\f\w\f')=
(u\rfloor\f)\w\f'+(-1)^{|\f|+[\f][u]}\f\w(u\rfloor\f').
\label{ws46}
\een
As a consequence, any graded derivation $u\in\gd\cA$ of $\cA$
yields a derivation
\mar{+117}\ben
&& \bL_u\f= u\rfloor d\f + d(u\rfloor\f), \qquad \f\in\cO^*[\gd\cA], \qquad
u\in\gd\cA, \label{+117} \\
&& \bL_u(\f\w\f')=\bL_u(\f)\w\f' + (-1)^{[u][\f]}\f\w\bL_u(\f'), \nonumber
\een
called the graded Lie derivative of the DBGA $\cO^*[\gd\cA]$.

Note that, if $\cA$ is a commutative ring, the graded Chevalley --
Eilenberg differential calculus comes to the familiar one.

The minimal graded differential calculus $\cO^*\cA\subset
\cO^*[\gd\cA]$  over a graded commutative ring $\cA$ consists of
the monomials $a_0da_1\w\cdots\w da_k$, $a_i\in\cA$. The
corresponding complex
\mar{t100}\beq
0\to\cK\ar \cA\ar^d\cO^1\cA\ar^d \cdots  \cO^k\cA\ar^d \cdots
\label{t100}
\eeq
is called the bigraded de Rham complex of a graded commutative
$\cK$-ring $\cA$.

\subsection{Graded manifolds}

A graded manifold of dimension $(n,m)$ is defined as a
local-ringed space $(Z,\gA)$ where $Z$ is an $n$-dimensional
smooth manifold $Z$ and $\gA=\gA_0\oplus\gA_1$ is a sheaf of
Grassmann algebras $\Lambda$ of rank $m$ (see Example \ref{grass})
such that \cite{bart,book09,sard09,book12,book13}:

$\bullet$ there is the exact sequence of sheaves
\mar{cmp140}\beq
0\to \cR \to\gA \op\to^\si C^\infty_Z\to 0, \qquad
\cR=\gA_1+(\gA_1)^2,\label{cmp140}
\eeq
where $\si$ is a body epimorphism onto a sheaf $C^\infty_Z$ of
smooth real functions on $Z$;

$\bullet$ $\cR/\cR^2$ is a locally free sheaf of
$C^\infty_Z$-modules of finite rank (with respect to pointwise
operations), and the sheaf $\gA$ is locally isomorphic to the
exterior product $\w_{C^\infty_Z}(\cR/\cR^2)$.

The sheaf $\gA$ is called a structure sheaf of a graded manifold
$(Z,\gA)$, and a manifold $Z$ is said to be the body of $(Z,\gA)$.
Sections of the sheaf $\gA$ are called graded functions on a
graded manifold $(Z,\gA)$. They make up a graded commutative
$C^\infty(Z)$-ring $\gA(Z)$ called the structure ring of
$(Z,\gA)$.

By virtue of the well-known Batchelor theorem \cite{bart,batch1},
graded manifolds possess the following structure.

\begin{theorem} \label{lmp1a} \mar{lmp1a}
Let $(Z,\gA)$ be a graded manifold. There exists a vector bundle
$E\to Z$ with an $m$-dimensional typical fibre $V$ such that the
structure sheaf $\gA$ of $(Z,\gA)$ is isomorphic to the structure
sheaf $\gA_E$ of germs of sections of the exterior bundle
\mar{ss12f11}\beq
\w E=Z\times \mathbb R \op\oplus_Z E \op\oplus_Z \op\w^2
E\op\oplus_Z\cdots \op\w^k E, \qquad k=\di E-\di Z,
\label{ss12f11}
\eeq
whose typical fibre is a Grassmann algebra $\w V^*$.
\end{theorem}

Combining Theorem \ref{lmp1a} and the above mentioned classical
Serre -- Swan theorem leads to the following Serre -- Swan theorem
for graded manifolds \cite{jmp05a,SS}.

\begin{theorem} \label{vv0} \mar{vv0}
Let $Z$ be a smooth manifold. A graded commutative
$C^\infty(Z)$-algebra $\cA$ is isomorphic to the structure ring of
a graded manifold with a body $Z$ iff it is the exterior algebra
of some projective $C^\infty(Z)$-module of finite rank.
\end{theorem}

\begin{proof}  By virtue of the Batchelor theorem,
any graded manifold is isomorphic to a simple graded manifold
$(Z,\gA_E)$ modelled over some vector bundle $E\to Z$. Its
structure ring $\cA_E$ (\ref{33f1}) of graded functions consists
of sections of the exterior bundle $\w E^*$ (\ref{ss12f11}). The
classical Serre -- Swan theorem states that a $C^\infty(Z)$-module
is isomorphic to the module of sections of a smooth vector bundle
over $Z$ iff it is a projective module of finite rank.
\end{proof}

It should be emphasized that Batchelor's isomorphism in Theorem
\ref{lmp1a} fails to be canonical. We agree to call $(Z,\gA_E)$ in
Theorem \ref{lmp1a} the simple graded manifold modelled over a
characteristic vector bundle $E\to Z$. Accordingly, a structure
ring $\gA_E(Z)$ of a simple graded manifold $(Z,\gA_E)$ is a
structure module
\mar{33f1}\beq
\cA_E=\gA_E(Z)=\w E^*(Z) \label{33f1}
\eeq
of sections of the exterior bundle $\w E^*$.

\begin{example} \label{triv} \mar{triv}
One can treat a local-ringed space $(Z,\gA_0=C_Z^\infty)$ as a
trivial graded manifold. It is a simple graded manifold whose
characteristic bundle is $E=Z\times\{0\}$. Its structure module is
a ring $C^\infty(Z)$ of smooth real functions on $Z$.
\end{example}

Given a simple graded manifold $(Z,\gA_E)$, every trivialization
chart $(U; z^A,q^a)$ of a vector bundle $E\to Z$ yields a
splitting domain $(U; z^A,c^a)$ of $(Z,\gA_E)$ where $\{c^a\}$ is
the corresponding local fibre basis for $E^*\to X$, i.e., $c^a$
are locally constant sections of $E^*\to X$ such that $q_b\circ
c^a=\dl^a_b$. Graded functions on such a chart are $\La$-valued
functions
\mar{z785}\beq
f=\op\sum_{k=0}^m \frac1{k!}f_{a_1\ldots a_k}(z)c^{a_1}\cdots
c^{a_k}, \label{z785}
\eeq
where $f_{a_1\cdots a_k}(z)$ are smooth functions on $U$. One
calls $\{z^A,c^a\}$ the local generating basis for a graded
manifold $(Z,\gA_E)$. Transition functions $q'^a=\rho^a_b(z^A)q^b$
of bundle coordinates on $E\to Z$ induce the corresponding
transformation $c'^a=\rho^a_b(z^A)c^b$ of the associated local
generating basis for a graded manifold $(Z,\gA_E)$ and the
according coordinate transformation law of graded functions
(\ref{z785}).

Let us consider the graded derivation module $\gd\gA(Z)$ of a real
graded commutative ring $\gA(Z)$. It is a real Lie superalgebra
relative to the superbracket (\ref{ws14}). Its elements are called
the graded vector fields on a graded manifold $(Z,\gA)$. A key
point is the following.

\begin{lemma} \mar{mm} \label{mm}
Graded vector fields $u\in\gd\cA_E$ on a simple graded manifold
$(Z,\gA_E)$ are represented by sections of some vector bundle as
follows \cite{book09,sard09,book13}.
\end{lemma}

\begin{proof}
Due to the canonical splitting $VE= E\times E$, the vertical
tangent bundle $VE$ of $E\to Z$ can be provided with the fibre
bases $\{\dr/\dr c^a\}$, which are the duals of the bases
$\{c^a\}$. Then graded vector fields on a splitting domain
$(U;z^A,c^a)$ of $(Z,\gA_E)$ read
\mar{hn14}\ben
&& u= u^A\dr_A + u^a\frac{\dr}{\dr c^a}, \label{hn14}\\
&& u'^A =u^A, \qquad u'^a=\rho^a_ju^j +u^A\dr_A(\rho^a_j)c^j, \nonumber\\
&& \frac{\dr}{\dr c^a}\circ\frac{\dr}{\dr c^b} =-\frac{\dr}{\dr
c^b}\circ\frac{\dr}{\dr c^a}, \qquad \dr_A\circ\frac{\dr}{\dr
c^a}=\frac{\dr}{\dr c^a}\circ \dr_A. \nonumber
\een
where $u^A, u^a$ are local graded functions on $U$, and they act
on graded functions $f\in\gA_E(U)$ (\ref{z785}) by the rule
\mar{cmp50a}\beq
u(f_{a\ldots b}c^a\cdots c^b)=u^A\dr_A(f_{a\ldots b})c^a\cdots c^b
+u^k f_{a\ldots b}\frac{\dr}{\dr c^k}\rfloor (c^a\cdots c^b).
\label{cmp50a}
\eeq
This rule implies the corresponding coordinate transformation law
\be
u'^A =u^A, \qquad u'^a=\rho^a_ju^j +u^A\dr_A(\rho^a_j)c^j
\ee
of graded vector fields. It follows that graded vector fields
(\ref{hn14}) can be represented by sections of a vector bundle
$\cV_E$ which is locally isomorphic to a vector bundle $\w
E^*\op\ot_Z(E\op\oplus_Z TZ)$.
\end{proof}

Given a structure ring $\cA_E$ of graded functions on a simple
graded manifold $(Z,\gA_E)$ and the real Lie superalgebra
$\gd\cA_E$ of its graded derivations, let us consider the graded
differential calculus
\mar{33f21}\beq
\cS^*[E;Z]=\cO^*[\gd\cA_E] \label{33f21}
\eeq
over $\cA_E$ where $\cS^0[E;Z]=\cA_E$.

\begin{lemma} \mar{mm1} \label{mm1}
Since the graded derivation module $\gd\cA_E$ is isomorphic to the
structure module of sections of a vector bundle $\cV_E\to Z$ in
Lemma \ref{mm}, elements of $\cS^*[E;Z]$ are represented by
sections of the exterior bundle $\w\ol\cV_E$ of the $\cA_E$-dual
$\ol\cV_E\to Z$ of $\cV_E$.
\end{lemma}

With respect to the dual fibre bases $\{dz^A\}$ for $T^*Z$ and
$\{dc^b\}$ for $E^*$, sections of $\ol\cV_E$ take a coordinate
form
\be
&& \f=\f_A dz^A + \f_adc^a, \\
&& \f'_a=\rho^{-1}{}_a^b\f_b, \qquad \f'_A=\f_A
+\rho^{-1}{}_a^b\dr_A(\rho^a_j)\f_bc^j.
\ee
The duality isomorphism $\cS^1[E;Z]=\gd\cA_E^*$ (\ref{ws47}) is
given by the graded interior product
\be
u\rfloor \f=u^A\f_A + (-1)^{\nw{\f_a}}u^a\f_a.
\ee
Elements of $\cS^*[E;Z]$ are called graded exterior forms on a
graded manifold $(Z,\gA_E)$. In particular, elements of
$\cS^*[E;Z]$ are graded functions on $(Z,\gA_E)$.

Seen as an $\cA_E$-algebra, the DBGA $\cS^*[E;Z]$ (\ref{33f21}) on
a splitting domain $(U;z^A,c^a)$ is locally generated by graded
one-forms $dz^A$, $dc^i$ such that
\be
dz^A\w dc^i=-dc^i\w dz^A, \qquad dc^i\w dc^j= dc^j\w dc^i.
\ee
Accordingly, the graded Chevalley -- Eilenberg coboundary operator
$d$ (\ref{ws86}), called the graded exterior differential, reads
\be
d\f= dz^A \w \dr_A\f +dc^a\w \frac{\dr}{\dr c^a}\f,
\ee
where derivatives $\dr_\la$, $\dr/\dr c^a$ act on coefficients of
graded exterior forms by the formula (\ref{cmp50a}), and they are
graded commutative with graded forms $dz^A$ and $dc^a$. The
formulas (\ref{ws45}) -- (\ref{+117}) hold.

\begin{lemma} \label{v62} \mar{v62}
The DBGA $\cS^*[E;Z]$ (\ref{33f21}) is a minimal differential
calculus over $\cA_E$, i.e., it is generated by elements $df$,
$f\in \cA_E$ \cite{book09}.
\end{lemma}

\begin{proof}
Since $\gd\cA_E=\cV_E(Z)$, this is a projective $C^\infty(Z)$- and
$\cA_E$-module of finite rank, and so is its $\cA_E$-dual
$\cS^1[E;Z]$. Hence, $\gd\cA_E$ is the $\cA_E$-dual of
$\cS^1[E;Z]$ and, consequently, $\cS^1[E;Z]$ is generated by
elements $df$, $f\in \cA_E$.
\end{proof}

The bigraded de Rham complex (\ref{t100}) of the minimal graded
differential calculus $\cS^*[E;Z]$ reads
\mar{+137}\beq
0\to \mathbb R\to \cA_E \ar^d \cS^1[E;Z]\ar^d\cdots
\cS^k[E;Z]\ar^d\cdots. \label{+137}
\eeq
Its cohomology $H^*(\cA_E)$  is called the de Rham cohomology of a
simple graded manifold $(Z,\gA_E)$.

In particular, given the differential graded algebra $\cO^*(Z)$ of
exterior forms on $Z$, there exists a canonical monomorphism
\mar{uut}\beq
\cO^*(Z)\to \cS^*[E;Z] \label{uut}
\eeq
and a body epimorphism $\cS^*[E;Z]\to \cO^*(Z)$ which are cochain
morphisms of the de Rham complex (\ref{+137}) and that of
$\cO^*(Z)$. Then one can show the following
\cite{book09,ijgmmp07}.

\begin{theorem} \label{33t3} \mar{33t3}
The de Rham cohomology of a graded manifold $(Z,\gA_E)$ equals the
de Rham cohomology of its body $Z$.
\end{theorem}

\begin{proof}
Let $\gA_E^k$ denote the sheaf of germs of graded $k$-forms on
$(Z,\gA_E)$. Its structure module is $\cS^k[E;Z]$. These sheaves
constitute the complex
\mar{1033}\beq
0\to\mathbb R\ar \gA_E \ar^d \gA_E^1\ar^d\cdots
\gA_E^k\ar^d\cdots. \label{1033}
\eeq
Its members $\gA_E^k$ are sheaves of $C^\infty_Z$-modules on $Z$
and, consequently, are fine and acyclic. Furthermore, the
Poincar\'e lemma for graded exterior forms holds \cite{bart}. It
follows that the complex (\ref{1033}) is a fine resolution of the
constant sheaf $\mathbb R$ on a manifold $Z$.  Then, by virtue of
the abstract de Rham theorem \cite{book09,hir}, there is an
isomorphism
\mar{+136}\beq
H^*(\cA_E)=H^*(Z;\mathbb R)=H^*_{\rm DR}(Z) \label{+136}
\eeq
of the cohomology $H^*(\cA_E)$ to the de Rham cohomology $H^*_{\rm
DR}(Z)$ of a smooth manifold $Z$.
\end{proof}

\begin{corollary} \label{33c1} \mar{33c1}
The cohomology isomorphism (\ref{+136}) accompanies the cochain
monomorphism (\ref{uut}). Hence, any closed graded exterior form
is decomposed into a sum $\f=\si +d\xi$ where $\si$ is a closed
exterior form on $Z$.
\end{corollary}

\subsection{Graded bundles over smooth manifolds}

A morphism of graded manifolds $(Z,\gA) \to (Z',\gA')$ is defined
as that of local-ringed spaces
\mar{su1}\beq
\phi:Z\to Z', \qquad \wh\Phi: \gA'\to \phi_*\gA, \label{su1}
\eeq
where $\phi$ is a manifold morphism and $\wh\Phi$ is a sheaf
morphism of $\gA'$ to the direct image $\phi_*\gA$ of $\gA$ onto
$Z'$ \cite{book05,ten}. The morphism (\ref{su1}) of graded
manifolds is said to be:

$\bullet$ a monomorphism if $\phi$ is an injection and $\wh\Phi$
is an epimorphism;

$\bullet$ an epimorphism if $\phi$ is a surjection and $\wh\Phi$
is a monomorphism.

An epimorphism of graded manifolds $(Z,\gA) \to (Z',\gA')$ where
$Z\to Z'$ is a fibre bundle is called the graded bundle
\cite{hern,stavr}. In this case, a sheaf monomorphism $\wh\Phi$
induces a monomorphism of canonical presheaves $\ol \gA'\to \ol
\gA$, which associates to each open subset $U\subset Z$ the ring
of sections of $\gA'$ over $\phi(U)$. Accordingly, there is a
pull-back monomorphism of the structure rings $\gA'(Z')\to\gA(Z)$
of graded functions on graded manifolds $(Z',\gA')$ and $(Z,\gA)$.

In particular, let $(Y,\gA)$ be a graded manifold whose body $Z=Y$
is a fibre bundle $\pi:Y\to X$. Let us consider the trivial graded
manifold $(X,\gA_0=C^\infty_X)$ (see Example \ref{triv}). Then we
have a graded bundle
\mar{su3}\beq
(Y,\gA) \to (X,C^\infty_X). \label{su3}
\eeq
We agree to call the graded bundle (\ref{su3}) over a trivial
graded manifold $(X, C^\infty_X)$ the graded bundle over a smooth
manifold \cite{sard14}. Let us denote it by $(X,Y,\gA)$. Given a
graded bundle $(X,Y,\gA)$, the local generating basis for a graded
manifold $(Y,\gA)$ can be brought into a form $(x^\la, y^i, c^a)$
where $(x^\la, y^i)$ are bundle coordinates of $Y\to X$.

If $Y\to X$ is a vector bundle, the graded bundle (\ref{su3}) is a
particular case of graded fibre bundles in \cite{hern,mont} when
their base is a trivial graded manifold.

\begin{example} \label{su20} \mar{su20}
Let $Y\to X$ be a fibre bundle. Then a  trivial graded manifold
$(Y,C^\infty_Y)$ together with a real ring monomorphism
$C^\infty(X)\to C^\infty(Y)$ is the graded bundle
$(X,Y,C^\infty_Y)$ (\ref{su3}).
\end{example}

\begin{example} \label{su21} \mar{su21} A graded manifold $(X,\gA)$ itself can
be treated as the graded bundle $(X,X, \gA)$ (\ref{su3})
associated to the identity smooth bundle $X\to X$.
\end{example}

Let $E\to Z$ and $E'\to Z'$ be vector bundles and $\Phi: E\to E'$
their bundle morphism over a morphism $\phi: Z\to Z'$. Then every
section $s^*$ of the dual bundle $E'^*\to Z'$ defines the
pull-back section $\Phi^*s^*$ of the dual bundle $E^*\to Z$ by the
law
\be
v_z\rfloor \Phi^*s^*(z)=\Phi(v_z)\rfloor s^*(\vf(z)), \qquad
v_z\in E_z.
\ee
It follows that a bundle morphism $(\Phi,\phi)$ yields a morphism
of simple graded manifolds
\mar{w901}\beq
(Z,\gA_E) \to (Z',\gA_{E'}). \label{w901}
\eeq
This is a pair $(\phi,\wh\Phi=\phi_*\circ\Phi^*)$ of a morphism
$\phi$ of  body manifolds and the composition $\phi_*\circ\Phi^*$
of the pull-back $\cA_{E'}\ni f\to \Phi^*f\in\cA_E$ of graded
functions and the direct image $\phi_*$ of a sheaf $\gA_E$ onto
$Z'$. Relative to local bases $(z^A,c^a)$ and $(z'^A,c'^a)$ for
$(Z,\gA_E)$ and $(Z',\gA_{E'})$, the morphism (\ref{w901}) of
simple graded manifolds reads $z'=\phi(z)$,
$\wh\Phi(c'^a)=\Phi^a_b(z)c^b$.

The graded manifold morphism (\ref{w901}) is a monomorphism (resp.
epimorphism) if $\Phi$ is a bundle injection (resp. surjection).

In particular, the graded manifold morphism (\ref{w901}) is a
graded bundle if $\Phi$ is a fibre bundle. Let $\cA_{E'} \to
\cA_E$ be the corresponding pull-back monomorphism of the
structure rings. By virtue of Lemma \ref{v62} it yields a
monomorphism of the DBGAs
\mar{xxx}\beq
\cS^*[E';Z']\to \cS^*[E;Z]. \label{xxx}
\eeq

Let $(Y,\gA_F)$ be a simple graded manifold modelled over a vector
bundle $F\to Y$. This is a graded bundle $(X,Y,\gA_F)$ modelled
over a composite bundle
\mar{su5}\beq
F\to Y\to X.  \label{su5}
\eeq
The structure ring of graded functions on a simple graded manifold
$(Y,\gA_F)$ is the graded commutative $C^\infty(X)$-ring $\cA_F=\w
F^*(Y)$ (\ref{33f1}). Let the composite bundle (\ref{su5}) be
provided with adapted bundle coordinates $(x^\la,y^i,q^a)$
possessing transition functions
\be
x'^\la(x^\mu), \qquad y'^i(x^\m,y^j), \qquad
q'^a=\rho^a_b(x^\mu,y^j)q^b.
\ee
Then the corresponding local generating basis for a simple graded
manifold $(Y,\gA_F)$ is $(x^\la,y^i,c^a)$ together with transition
functions
\be
x'^\la(x^\mu), \qquad y'^i(x^\m,y^j), \qquad
c'^a=\rho^a_b(x^\mu,j^j)c^b.
\ee
We call it the local generating basis for a graded bundle
$(X,Y,\gA_F)$.

\subsection{Graded jet manifolds}

As was mentioned above, Lagrangian theory on a smooth fibre bundle
$Y\to X$ is formulated in terms of the variational bicomplex on
jet manifolds $J^*Y$ of $Y$. These are fibre bundles over $X$ and,
therefore, they can be regarded as trivial graded bundles $(X,
J^kY, C^\infty_{J^kY})$. Then let us describe their partners in
the case of graded bundles (\ref{su11}) as follows.

Note that, given a graded manifold $(X,\gA)$ and its structure
ring $\cA$, one can define the jet module $J^1\cA$ of a
$C^\infty(X)$-ring $\cA$ \cite{book05,book12}. If $(X,\gA_E)$ is a
simple graded manifold modelled over a vector bundle $E\to X$, the
jet module $J^1\cA_E$ is a module of global sections of the jet
bundle $J^1(\w E^*)$. A problem is that $J^1\cA_E$ fails to be a
structure ring of some graded manifold. By this reason, we have
suggested a different construction of jets of graded manifolds,
though it is applied only to simple graded manifolds
\cite{book09,book13,sard13}.

Let $(X,\cA_E)$ be a simple graded manifold modelled over a vector
bundle $E\to X$. Let us consider a $k$-order jet manifold $J^kE$
of $E$ It is a vector bundle over $X$. Then let $(X,\cA_{J^kE})$
be a simple graded manifold modelled over $J^kE\to X$. We agree to
call $(X,\cA_{J^kE})$ the graded $k$-order jet manifold of a
simple graded manifold $(X,\cA_E)$. Given a splitting domain $(U;
x^\la,c^a)$ of a graded manifold $(Z,\cA_E)$, we have a splitting
domain
\be
(U; x^\la,c^a, c^a_\la,c^a_{\la_1\la_2}, \ldots
c^a_{\la_1\ldots\la_k}), \qquad
c'{}^a_{\la\la_1\ldots\la_r}=\rho^a_b(x)c^a_{\la\la_1\ldots\la_r}
+ \dr_\la\rho^a_b(x) c^a_{\la_1\ldots\la_r},
\ee
of a graded jet manifold $(X,\cA_{J^kE})$.

As was mentioned above, a graded manifold is a particular graded
bundle over its body (Example \ref{su21}). Then the definition of
graded jet manifolds is generalized to graded bundles over smooth
manifolds as follows \cite{sard14}.

Let $(X,Y,\gA_F)$ be a graded bundle modelled over the composite
bundle (\ref{su5}). It is readily observed that the jet manifold
$J^rF$ of $F\to X$ is a vector bundle $J^rF\to J^rY$ coordinated
by $(x^\la, y^i_\La, q^a_\La)$, $0\leq |\La|\leq r$. Let
$(J^rY,\gA_r=\gA_{J^rF})$ be a simple graded manifold modelled
over this vector bundle. Its local generating basis is $(x^\la,
y^i_\La, c^a_\La)$, $0\leq|\La|\leq r$. We call $(J^rY,\gA_r)$ the
graded $r$-order jet manifold of a graded bundle $(X,Y,\gA_F)$.

In particular, let $Y\to X$ be a smooth bundle seen as a trivial
graded bundle $(X, Y, C^\infty_Y)$ modelled over a composite
bundle $Y\times\{0\}\to Y\to X$. Then its graded jet manifold is a
trivial graded bundle $(X, J^rY, C^\infty_{J^rY)})$, i.e., a jet
manifold $J^rY$ of $Y$.

Thus, the above definition of graded jet manifolds of graded
bundles is compatible with the conventional definition of jets of
fibre bundles. It differs from that of jet graded bundles in
\cite{hern,mont},  but reproduces the heuristic notion of jets of
odd ghosts in BRST field theory \cite{barn,bran01}.

Jet manifolds $J^*Y$ of a fibre bundle $Y\to X$ form the inverse
sequence
\mar{j1}\beq
Y\op\longleftarrow^\pi J^1Y \longleftarrow \cdots J^{r-1}Y
\op\longleftarrow^{\pi^r_{r-1}} J^rY\longleftarrow\cdots,
\label{j1}
\eeq
of affine bundles $\pi^r_{r-1}$. One can think of elements of its
projective limit $J^\infty Y$ as being infinite order jets of
sections of $Y\to X$ identified by their Taylor series at points
of $X$. The set $J^\infty Y$ is endowed with the projective limit
topology which makes $J^\infty Y$ into a paracompact Fr\'echet
manifold \cite{book09,tak2}. It is called the infinite order jet
manifold. A bundle coordinate atlas $(x^\la, y^i)$ of $Y$ provides
$J^\infty Y$ with the adapted manifold coordinate atlas
\mar{j3}\ben
 && (x^\la, y^i_\La), \quad 0\leq|\La|, \qquad
{y'}^i_{\la+\La}=\frac{\dr x^\m}{\dr x'^\la}d_\m y'^i_\La,
\label{j3} \\
&& d_\la= \dr_\la + y^i_\la \dr_i +
\op\sum_{0<|\La|}y^i_{\la+\La}\dr^\La_i, \nonumber
\een
where $d_\la$ are total derivatives. A fibre bundle $Y$ is a
strong deformation retract of the infinite order jet manifold
$J^\infty Y$ \cite{ander,jmp}. Then by virtue of the Vietoris --
Begle theorem \cite{bred}, there is an isomorphism
\mar{j19'}\beq
H^*(J^\infty Y;\mathbb R)=H^*(Y;\mathbb R)=H^*_{\rm DR}(Y)
\label{j19'}
\eeq
between the cohomology of $J^\infty Y$ with coefficients in the
constant sheaf $\mathbb R$ and the de Rham cohomology of $Y$.

The inverse sequence (\ref{j1}) of jet manifolds yields the direct
sequence of graded differential algebras $\cO_r^*$ of exterior
forms on finite order jet manifolds
\mar{5.7}\beq
\cO^*(X)\op\longrightarrow^{\pi^*} \cO^*(Y)
\op\longrightarrow^{\pi^1_0{}^*} \cO_1^* \longrightarrow \cdots
\cO_{r-1}^*\op\longrightarrow^{\pi^r_{r-1}{}^*}
 \cO_r^* \longrightarrow\cdots, \label{5.7}
\eeq
where $\pi^r_{r-1}{}^*$ are the pull-back monomorphisms. Its
direct limit
\mar{ppp}\beq
\cO^*_\infty =\op\lim^\to \cO_r^* \label{ppp}
\eeq
consists of all exterior forms on finite order jet manifolds
modulo the pull-back identification. The $\cO^*_\infty$
(\ref{ppp}) is a differential graded algebra which inherits the
operations of the exterior differential $d$ and exterior product
$\w$ of exterior algebras $\cO^*_r$. One can show that the
cohomology $H^*(\cO_\infty^*)$ of the de Rham complex
\mar{5.13} \beq
0\longrightarrow \mathbb R\longrightarrow \cO^0_\infty
\op\longrightarrow^d\cO^1_\infty \op\longrightarrow^d \cdots
\label{5.13}
\eeq
of a differential graded algebra $\cO^*_\infty$ equals the de Rham
cohomology $H^*_{\rm DR}(Y)$ of a fibre bundle $Y$
\cite{and,book09,book13}. This follows from the fact that, by
virtue of the well-known theorem, the cohomology
$H^*(\cO_\infty^*)$ is isomorphic to the direct limit of the
cohomology groups $H^*(\cO_r^*)=H^*_{\rm DR}(j^rY)$, but all of
them equal the de Rham cohomology $H^*_{\rm DR}(Y)$ of $Y$ because
$J^rY\to J^{r-1}Y$ are affine bundles and, consequently, $Y$ is a
strong deformation retract of any finite order jet manifold
$J^rY$.

The fibre bundles $J^{r+1}Y\to J^rY$ (\ref{j1}) and the
corresponding bundles $J^{r+1}F\to J^rF$ also yield the graded
bundles
\be
(J^{r+1}Y,\gA_{r+1}) \to (J^rY,\gA_r),
\ee
including pull-back monomorphism of the structure rings
\mar{34f1}\beq
\cS^0_r[F;Y]\to \cS^0_{r+1}[F;Y] \label{34f1}
\eeq
of graded functions on graded manifolds $(J^rY,\gA_r)$ and
$(J^{r+1}Y,\gA_{r+1})$. As a consequence, we have the inverse
sequence of graded manifolds
\be
(Y,\cA_F)\op\longleftarrow (J^1Y,\gA_{J^1F}) \longleftarrow \cdots
(J^{r-1}Y, \gA_{J^{r-1}F}) \op\longleftarrow (J^rY,
\gA_{J^rF})\longleftarrow\cdots.
\ee
One can think on its inverse limit $(J^\infty Y, \cA_{J^\infty
F})$ as the graded infinite order jet manifold  whose body is an
infinite order jet manifold $J^\infty Y$ and whose structure sheaf
$\cA_{J^\infty F}$ is a sheaf of germs of graded functions on
graded manifolds $(J^*Y,\gA_{J^*F})$ \cite{book09,sard13}. However
$(J^\infty Y, \cA_{J^\infty F})$ fails to be a graded manifold in
a strict sense because the projective image $J^\infty Y$ of the
sequence (\ref{j1}) is a Fr\'eche manifold, but not the smooth
one.

By virtue of Lemma \ref{v62}, the differential calculus
$\cS^*_r[F;Y]$ are minimal. Therefore, the monomorphisms of
structure rings (\ref{34f1}) yields the pull-back monomorphisms
(\ref{xxx}) of DBGAs
\mar{v4'}\beq
\pi_r^{r+1*}:\cS^*_r[F;Y]\to \cS^*_{r+1}[F;Y]. \label{v4'}
\eeq
As a consequence, we have the direct system of DBGAs
\mar{j2}\beq
\cS^*[F;Y]\ar^{\pi^*} \cS^*_1[F;Y]\ar\cdots \cS^*_{r-1}[F;Y]
\op\ar^{\pi^{r*}_{r-1}}\cS^*_r[F;Y]\ar\cdots. \label{j2}
\eeq
The DBGA  $\cS^*_\infty[F;Y]$ that we associate to a graded bundle
$(Y,\gA_F)$ is defined as the direct limit
\mar{5.77a}\beq
\cS^*_\infty [F;Y]=\op\lim^\to \cS^*_r[F;Y] \label{5.77a}
\eeq
of the direct system (\ref{j2}). It consists of all graded
exterior forms $\f\in \cS^*[F_r;J^rY]$ on graded manifolds
$(J^rY,\gA_r)$ modulo the monomorphisms (\ref{v4'}). Its elements
obey the relations (\ref{ws45}).

The cochain monomorphisms $\cO^*_r\to \cS^*_r[F;Y]$ (\ref{uut})
provide a monomorphism of the direct system (\ref{5.7}) to the
direct system (\ref{j2}) and, consequently, a monomorphism
\mar{v7}\beq
\cO^*_\infty\to \cS^*_\infty[F;Y]  \label{v7}
\eeq
of their direct limits. In particular, $\cS^*_\infty[F;Y]$ is an
$\cO^0_\infty$-algebra. Accordingly, the body epimorphisms
$\cS^*_r[F;Y]\to \cO^*_r$ yield an epimorphism of
$\cO^0_\infty$-algebras
\mar{v7'}\beq
\cS^*_\infty[F;Y]\to \cO^*_\infty.  \label{v7'}
\eeq
It is readily observed that the morphisms (\ref{v7}) and
(\ref{v7'}) are cochain morphisms between the de Rham complex
(\ref{5.13}) of $\cO^*_\infty$ and the de Rham complex
\mar{g110}\beq
0\to\mathbb R\ar \cS^0_\infty[F;Y]\ar^d \cS^1_\infty[F;Y]\cdots
\ar^d\cS^k_\infty[F;Y] \ar\cdots \label{g110}
\eeq
of a DBGA $\cS^*_\infty[F;Y]$. Moreover, the corresponding
homomorphisms of cohomology groups of these complexes are
isomorphisms as follows \cite{book09}.

\begin{lemma} \label{v9} \mar{v9} There is an isomorphism
\mar{v10'}\beq
H^*(\cS^*_\infty[F;Y])= H^*_{DR}(Y) \label{v10'}
\eeq
of the cohomology $H^*(\cS^*_\infty[F;Y])$ of the de Rham complex
(\ref{g110}) to the de Rham cohomology  $H^*_{DR}(Y)$ of $Y$.
\end{lemma}

\begin{proof}
The complex (\ref{g110}) is the direct limit of the de Rham
complexes of the differential graded algebras $\cS^*_r[F;Y]$.
Therefore, the direct limit of cohomology groups of these
complexes is the cohomology of the de Rham complex (\ref{g110}) in
accordance with the above mentioned theorem. By virtue of the
cohomology isomorphism (\ref{+136}), cohomology of the de Rham
complex of $\cS^*_r[F;Y]$ equals the de Rham cohomology of $J^rY$
and, consequently, that of $Y$, which is the strong deformation
retract of any jet manifold $J^rY$. Hence, the isomorphism
(\ref{v10'}) holds.
\end{proof}

\begin{corollary} \mar{34c1} \label{34c1}
Any closed graded form $\f\in \cS^*_\infty[F;Y]$ is decomposed
into the sum $\f=\si +d\xi$ where $\si$ is a closed exterior form
on $Y$.
\end{corollary}

One can think of  elements of $\cS^*_\infty[F;Y]$ as being graded
differential forms on an infinite order jet manifold $J^\infty Y$
\cite{book09,sard13}. Indeed, let $\gS^*_r[F;Y]$ be the sheaf of
DBGAs on $J^rY$ and $\ol\gS^*_r[F;Y]$ its canonical presheaf. Then
the above mentioned presheaf monomorphisms $\ol \gA_r\to \ol
\gA_{r+1}$ yield the direct system of presheaves
\mar{v15}\beq
\ol\gS^*[F;Y]\ar \ol\gS^*_1[F;Y] \ar\cdots \ol\gS^*_r[F;Y]
\ar\cdots, \label{v15}
\eeq
whose direct limit $\ol\gS_\infty^*[F;Y]$ is a presheaf of DBGAs
on the infinite order jet manifold $J^\infty Y$. Let
$\gQ^*_\infty[F;Y]$ be the sheaf of DBGAs of germs of the presheaf
$\ol\gS_\infty^*[F;Y]$. One can think of the pair $(J^\infty Y,
\gQ^0_\infty[F;Y])$ as being a graded Fr\'echet manifold, whose
body is the infinite order jet manifold $J^\infty Y$ and the
structure sheaf $\gQ^0_\infty[F;Y]$ is the sheaf of germs of
graded functions on graded manifolds $(J^rY,\gA_r)$. We agree to
call $(J^\infty Y, \gQ^0_\infty[F;Y])$ the graded infinite order
jet manifold. The structure module $\cQ^*_\infty[F;Y]$ of sections
of $\gQ^*_\infty[F;Y]$ is a DBGA such that, given an element
$\f\in \cQ^*_\infty[F;Y]$ and a point $z\in J^\infty Y$, there
exist an open neighbourhood $U$ of $z$ and a graded exterior form
$\f^{(k)}$ on some finite order jet manifold $J^kY$ so that
$\f|_U= \pi^{\infty*}_k\f^{(k)}|_U$. In particular, there is the
monomorphism
\mar{34f5}\beq
\cS^*_\infty[F;Y] \to\cQ^*_\infty[F;Y]. \label{34f5}
\eeq

Due to this monomorphism, one can restrict $\cS^*_\infty[F;Y]$ to
the coordinate chart (\ref{j3}) of $J^\infty Y$ and say that
$\cS^*_\infty[F;Y]$ as an $\cO^0_\infty$-algebra is locally
generated by the elements
\be
(c^a_\La,
dx^\la,\thh^a_\La=dc^a_\La-c^a_{\la+\La}dx^\la,\thh^i_\La=
dy^i_\La-y^i_{\la+\La}dx^\la), \qquad 0\leq |\La|,
\ee
where $c^a_\La$, $\thh^a_\La$ are odd and $dx^\la$, $\thh^i_\La$
are even. We agree to call $(y^i,c^a)$ the local generating basis
for $\cS^*_\infty[F;Y]$. Let the collective symbol $s^A$ stand for
its elements. Accordingly, the notation $s^A_\La$ for their jets
and the notation
\mar{kk}\beq
\theta^A_\La=ds^A_\La- s^A_{\la+\La}dx^\la \label{kk}
\eeq
for the contact forms are introduced. For the sake of simplicity,
we further denote $[A]=[s^A]$.

\begin{remark} \mar{pkp} \label{pkp}
Let $(X,Y,\gA_F)$ and $(X,Y',\gA_{F'})$ be graded bundles modelled
over composite bundles $F\to Y\to X$ and $F'\to Y'\to X$,
respectively. Let $F\to F'$ be a fibre bundle over a fibre bundle
$Y\to Y'$ over $X$. Then we have a graded bundle
\be
(X,Y,\gA_F) \to (X,Y',\gA_{F'})
\ee
together with the pull-back monomorphism (\ref{xxx}) of DBGAs
\mar{lmk}\beq
\cS^* [F';Y'] \to \cS^* [F;Y]. \label{lmk}
\eeq
Let $(X,J^rY,\gA_{J^rF})$ and $(X,J^rY',\gA_{J^rF'})$ be graded
bundles modelled over composite bundles $J^rF\to J^rY\to X$ and
$J^rF'\to J^rY'\to X$, respectively. Since $J^rF\to J^rF'$ is a
fibre bundle over a a fibre bundle $J^rY\to J^rY'$ over $X$
\cite{book}, we also get a graded bundle
\be
(X,J^rY,\gA_{J^rF})\to (X,J^rY',\gA_{J^rF'})
\ee
together with the pull-back monomorphism of DBGAs
\mar{lmm}\beq
\cS^*_r [F';Y'] \to \cS^*_r [F;Y]. \label{lmm}
\eeq
The monomorphisms (\ref{lmk}) -- (\ref{lmm}), $r=1, 2,\ldots$,
provide a monomorphism of the direct limits
\mar{llm}\beq
\cS^*_\infty [F';Y'] \to \cS^*_\infty [F;Y]. \label{llm}
\eeq
of DBGAs $\cS^*_r [F';Y']$ and $\cS^*_r [F;Y]$, $r=0,1, 2,\ldots$.
\end{remark}

\begin{remark} \mar{ghg} \label{ghg}
Let $(X,Y,\gA_F)$ and $(X,Y',\gA_{F'})$ be graded bundles modelled
over composite bundles $F\to Y\to X$ and $F'\to Y'\to X$,
respectively. We define their product over $X$ as the graded
bundle
\mar{tyt}\beq
(X,Y,\gA_F)\op\times_X(X,Y',\gA_{F'})=(X,Y\op\times_X Y',
\gA_{F\op\times_X F'}) \label{tyt}
\eeq
modelled over a composite bundle
\mar{wqw}\beq
F\op\times_X F'=F\op\times_{Y\times Y'} F'\to Y\op\times_X Y'\to
X. \label{wqw}
\eeq
Let us consider the corresponding DBGA
\mar{bvb}\beq
\cS^*_\infty [F\op\times_XF';Y\op\times_X Y']. \label{bvb}
\eeq
Then in accordance with Remark \ref{pkp}, there are the
monomorphisms (\ref{llm}) of BGDAs
\mar{opo}\beq
\cS^*_\infty [F;Y]\to \cS^*_\infty [F\op\times_XF;Y\op\times_X
Y'], \qquad \cS^*_\infty [F';Y']\to \cS^*_\infty
[F\op\times_XF;Y\op\times_X Y']. \label{opo}
\eeq
\end{remark}

\section{Graded Lagrangian formalism}

Let $(X,Y,\gA_F)$ be a graded bundle modelled over the composite
bundle (\ref{su5}) over an $n$-dimensional smooth manifold $X$,
and let $\cS^*_\infty [F;Y]$ be the associated DBGA (\ref{5.77a})
of graded exterior forms on graded jet manifolds of $(X,Y,\gA_F)$.
As was mentioned above Grassmann-graded Lagrangian theory of even
and odd variables on a graded bundle is formulated in terms of the
variational bicomplex which the DBGA $\cS^*_\infty [F;Y]$ is split
in \cite{lmp08,cmp04,book09,sard13}.

A DBGA $\cS^*_\infty[F;Y]$ is decomposed into
$\cS^0_\infty[F;Y]$-modules $\cS^{k,r}_\infty[F;Y]$ of $k$-contact
and $r$-horizontal graded forms together with the corresponding
projections
\be
h_k:\cS^*_\infty[F;Y]\to \cS^{k,*}_\infty[F;Y], \qquad
h^m:\cS^*_\infty[F;Y]\to \cS^{*,m}_\infty[F;Y].
\ee
Accordingly, the graded exterior differential $d$ on
$\cS^*_\infty[F;Y]$ falls into a sum $d=d_V+d_H$ of the vertical
graded differential
\be
d_V \circ h^m=h^m\circ d\circ h^m, \qquad d_V(\f)=\thh^A_\La \w
\dr^\La_A\f, \qquad \f\in\cS^*_\infty[F;Y],
\ee
and the total graded differential
\be
d_H\circ h_k=h_k\circ d\circ h_k, \qquad d_H\circ h_0=h_0\circ d,
\qquad d_H(\f)=dx^\la\w d_\la(\f),
\ee
where
\be
d_\la = \dr_\la + \op\sum_{0\leq|\La|} s^A_{\la+\La}\dr_A^\La
\ee
are the graded total derivatives. These differentials obey the
nilpotent relations
\be
d_H\circ d_H=0, \qquad d_V\circ d_V=0, \qquad d_H\circ
d_V+d_V\circ d_H=0.
\ee

A DBGA $\cS^*_\infty[F;Y]$ also is provided with the graded
projection endomorphism
\be
&& \vr=\op\sum_{k>0} \frac1k\ol\vr\circ h_k\circ h^n:
\cS^{*>0,n}_\infty[F;Y]\to \cS^{*>0,n}_\infty[F;Y], \\
&& \ol\vr(\f)= \op\sum_{0\leq|\La|} (-1)^{\nm\La}\thh^A\w
[d_\La(\dr^\La_A\rfloor\f)], \qquad \f\in \cS^{>0,n}_\infty[F;Y],
\ee
such that $\vr\circ d_H=0$, and with the nilpotent graded
variational operator
\mar{34f10}\beq
\dl=\vr\circ d :\cS^{*,n}_\infty[F;Y]\to
\cS^{*+1,n}_\infty[F;Y].\label{34f10}
\eeq
With these operators a DBGA $\cS^{*,}_\infty[F;Y]$ is decomposed
into the Grassmann-graded variational bicomplex
\mar{7}\beq
\begin{array}{ccccrlcrlcccrlcrl}
 & &  &  & & \vdots & & & \vdots  & & & & &
\vdots & &   & \vdots \\
& & & & _{d_V} & \put(0,-7){\vector(0,1){14}} & & _{d_V} &
\put(0,-7){\vector(0,1){14}} & &  & & _{d_V} &
\put(0,-7){\vector(0,1){14}}& & _{-\dl} & \put(0,-7){\vector(0,1){14}} \\
 &  & 0 & \to & &\cS^{1,0}_\infty[F;Y] &\op\to^{d_H} & &
\cS^{1,1}_\infty[F;Y] & \op\to^{d_H} & \cdots & & &
\cS^{1,n}_\infty[F;Y] &\op\to^\vr &  & \vr(\cS^{1,n}_\infty[F;Y])\to  0\\
& & & & _{d_V} &\put(0,-7){\vector(0,1){14}} & & _{d_V} &
\put(0,-7){\vector(0,1){14}} & & &  & _{d_V} &
\put(0,-7){\vector(0,1){14}}
 & & _{-\dl} & \put(0,-7){\vector(0,1){14}} \\
0 & \to & \mathbb R & \to & & \cS^0_\infty[F;Y] &\op\to^{d_H} & &
\cS^{0,1}_\infty[F;Y] & \op\to^{d_H} & \cdots & & &
\cS^{0,n}_\infty[F;Y] & \equiv &  & \cS^{0,n}_\infty[F;Y] \\
& & & & & \put(0,-7){\vector(0,1){14}} & &  &
\put(0,-7){\vector(0,1){14}} & & & &  &
\put(0,-7){\vector(0,1){14}} & &  & \\
0 & \to & \mathbb R & \to & & \cO^0(X) &\op\to^d & & \cO^1(X) &
\op\to^d & \cdots & & &
\cO^n(X) & \op\to^d &  & 0 \\
& & & & &\put(0,-5){\vector(0,1){10}} & & &
\put(0,-5){\vector(0,1){10}} & &  &  & &
\put(0,-5){\vector(0,1){10}} & &  & \\
& & & & &0 & &  & 0 & & & & &  0 & &  &
\end{array}
\label{7}
\eeq
Its relevant cohomology has been found
\cite{book09,ijgmmp07,sard13}.

We restrict our consideration to the short variational subcomplex
\mar{g111}\ben
&& 0\to \mathbb R\to \cS^0_\infty[F;Y]\ar^{d_H}\cS^{0,1}_\infty[F;Y]
\cdots \ar^{d_H}\cS^{0,n}_\infty[F;Y]\ar^\dl \bE_1, \label{g111}\\
&& \bE_1=\vr(\cS^{1,n}_\infty[F;Y]), \qquad n=\di X, \nonumber
\een
of the bicomplex (\ref{7}) and the subcomplex of one-contact
graded forms
\mar{g112}\beq
 0\to \cS^{1,0}_\infty[F;Y]\ar^{d_H} \cS^{1,1}_\infty[F;Y]\cdots
\ar^{d_H}\cS^{1,n}_\infty[F;Y]\ar^\vr \bE_1\to 0. \label{g112}
\eeq

They possess the following cohomology
\cite{cmp04,ijgmmp07,sard13}.

\begin{theorem} \label{v11} \mar{v11}
Cohomology of the complex (\ref{g111}) equals the de Rham
cohomology $H^*_{DR}(Y)$ of $Y$.
\end{theorem}

\begin{theorem} \label{v11'} \mar{v11'}
The complex (\ref{g112}) is exact.
\end{theorem}

Decomposed into a variational bicomplex, the DBGA
$\cS^*_\infty[F;Y]$ describes Grassmann-graded Lagrangian theory
on a graded bundle $(X,Y,\gA_F)$. Its graded Lagrangian is defined
as an element
\mar{0709}\beq
L=\cL\om\in \cS^{0,n}_\infty[F;Y], \qquad \om=dx^1\w\cdots\w dx^n,
\label{0709}
\eeq
of the graded variational complex (\ref{g111}). Accordingly, a
graded exterior form
\mar{0709'}\beq
\dl L= \thh^A\w \cE_A\om=\op\sum_{0\leq|\La|}
 (-1)^{|\La|}\thh^A\w d_\La (\dr^\La_A L)\om\in \bE_1 \label{0709'}
\eeq
is said to be its graded Euler -- Lagrange operator. Its kernel
defines the Euler -- Lagrange equation
\mar{eq}\beq
\dl L=0, \qquad \cE_A=\op\sum_{0\leq|\La|}
 (-1)^{|\La|}\thh^A\w d_\La (\dr^\La_A L)=0. \label{eq}
\eeq

Therefore, we agree to call a pair $(\cS^{0,n}_\infty[F;Y],L)$ the
Grassmann-graded (or, simply, graded) Lagrangian system and
$\cS^*_\infty[F;Y]$ its structure algebra.

The following is a corollary of Theorems \ref{v11} and (\ref{v11'}
\cite{cmp04,book09,sard13}.

\begin{theorem} \label{cmp26} \mar{cmp26}
Every $d_H$-closed graded form $\f\in\cS^{0,m<n}_\infty[F;Y]$
falls into the sum
\mar{g214}\beq
\f=h_0\si + d_H\xi, \qquad \xi\in \cS^{0,m-1}_\infty[F;Y],
\label{g214}
\eeq
where $\si$ is a closed $m$-form on $Y$. Any $\dl$-closed (i.e.,
variationally trivial) graded Lagrangian $L\in
\cS^{0,n}_\infty[F;Y]$ is a sum
\mar{g215}\beq
L=h_0\si + d_H\xi, \qquad \xi\in \cS^{0,n-1}_\infty[F;Y],
\label{g215}
\eeq
where $\si$ is a closed $n$-form on $Y$.
\end{theorem}

\begin{proof}
The complex (\ref{g111}) possesses the same cohomology as the
short variational complex
\mar{b317'}\beq
0\to\mathbb R\to \cO^0_\infty  \ar^{d_H}\cO^{0,1}_\infty \cdots
\op\ar^{d_H} \cO^{0,n}_\infty  \op\ar^\dl \bE_1 \label{b317'}
\eeq
of the differential graded algebra  $\cO^*_\infty$. The
monomorphism (\ref{v7}) and the body epimorphism (\ref{v7'}) yield
the corresponding cochain morphisms of the complexes (\ref{g111})
and (\ref{b317'}). Therefore, cohomology of the complex
(\ref{g111}) is the image of the cohomology of $\cO^*_\infty$.
\end{proof}

\begin{corollary} \mar{34c5} \label{34c5}
Any variationally trivial odd Lagrangian is $d_H$-exact.
\end{corollary}

The exactness of the complex (\ref{g112}) at the term
$\cS^{1,n}_\infty[F;Y]$ results in the following
\cite{cmp04,book09,sard13}.

\begin{theorem} \label{g103} \mar{g103}
Given a graded Lagrangian $L$, there is the decomposition
\mar{g99,'}\ben
&& dL=\dl L - d_H\Xi_L,
\qquad \Xi\in \cS^{n-1}_\infty[F;Y], \label{g99}\\
&& \Xi_L=L+\op\sum_{s=0} \thh^A_{\nu_s\ldots\nu_1}\w
F^{\la\nu_s\ldots\nu_1}_A\om_\la, \label{g99'}\\
&& F_A^{\nu_k\ldots\nu_1}= \dr_A^{\nu_k\ldots\nu_1}\cL-d_\la
F_A^{\la\nu_k\ldots\nu_1} +\si_A^{\nu_k\ldots\nu_1},\qquad
k=1,2,\ldots,\nonumber
\een
where local graded functions $\si$ obey the relations
\be
\si^\nu_A=0,\qquad \si_A^{(\nu_k\nu_{k-1})\ldots\nu_1}=0.
\ee
\end{theorem}

The form $\Xi_L$ (\ref{g99'}) provides a global Lepage equivalent
of a graded Lagrangian $L$. In particular, one can locally choose
$\Xi_L$ (\ref{g99'}) where all functions $\si$ vanish.

The formula (\ref{g99}), called the variational decomposition,
play a prominent role in the formulation and the proof of Noether
theorems.

\begin{example} \mar{first} \label{first} Let us consider first order Lagrangian
theory on a fibre bundle $Y\to X$. Its structure algebra
(\ref{5.77a}) is the $\cO_\infty^*$ (\ref{ppp}). The first order
Lagrangian (\ref{0709}) reads
\mar{23f2}\beq
L=\cL\om: J^1Y\to \op\w^n T^*X. \label{23f2}
\eeq
The corresponding second-order Euler -- Lagrange operator
(\ref{0709'}) takes a form
\mar{305}\ben
&& \cE_L: J^2Y\to T^*Y\w(\op\w^nT^*X), \nonumber \\
&& \cE_L= (\dr_i\cL- d_\la\pi^\la_i) \thh^i\w\om, \qquad
\pi^\la_i=\dr^\la_i\cL. \label{305}
\een
Its kernel defines the second order Euler -- Lagrange equation
\mar{b327}\beq
(\dr_i- d_\la\dr^\la_i)\cL=0. \label{b327}
\eeq
Given the first order Lagrangian $L$ (\ref{23f2}), its Lepage
equivalents $\Xi_L$ (\ref{g99'}) in the variational decomposition
(\ref{g99}) read
\mar{22f44}\beq
\Xi_L=L+(\p^\la_i-d_\m \si^{\m\la}_i)\thh^i\w\om_\la
+\si^{\la\m}_i \thh^i_\m\w\om_\la, \label{22f44}
\eeq
where $\si^{\m\la}_i=-\si^{\la\m}_i$ are skew-symmetric local
functions on $Y$. One usually choose the Poincar\'e -- Cartan form
\mar{303}\beq
H_L=\cL\om +\p^\la_i\thh^i\w\om_\la.  \label{303}
\eeq
\end{example}

\section{First Noether theorem}

As was mentioned above, the first Noether theorem (Theorem
\ref{j45}) is a straightforward corollary of the variational
decomposition (\ref{g99}).

\subsection{Infinitesimal graded transformations of Lagrangian systems}

Given a graded Lagrangian system $(\cS^*_\infty[F;Y], L)$, by its
infinitesimal transformations are meant contact graded derivations
of the real graded commutative ring $\cS^0_\infty[F;Y]$. These
derivations constitute a $\cS^0_\infty[F;Y]$-module $\gd
\cS^0_\infty[F;Y]$ which is a real Lie superalgebra with respect
to the Lie superbracket (\ref{ws14}). The following holds
\cite{cmp04,book09}.

\begin{theorem} \label{35t1} \mar{35t1}
The derivation module $\gd\cS^0_\infty[F;Y]$ is isomorphic to the
$\cS^0_\infty[F;Y]$-dual $(\cS^1_\infty[F;Y])^*$ of the module of
graded one-forms $\cS^1_\infty[F;Y]$.
\end{theorem}

\begin{proof} At first, let us show that $\cS^*_\infty[F;Y]$ is generated by elements
$df$, $f\in \cS^0_\infty[F;Y]$. It suffices to justify that any
element of $\cS^1_\infty[F;Y]$ is a finite
$\cS^0_\infty[F;Y]$-linear combination of elements $df$, $f\in
\cS^0_\infty[F;Y]$. Indeed, every $\f\in\cS^1_\infty[F;Y]$ is a
graded exterior form on some finite order jet manifold $J^rY$,
i.e., a section of a vector bundle $\ol \cV_{J^rF}\to J^rY$ in
accordance with Lemma \ref{mm1}. Then by virtue of the classical
Serre -- Swan theorem, a $C^\infty(J^rY)$-module $\cS^1_r[F;Y]$ of
graded one-forms on $J^rY$ is a projective module of finite rank,
i.e., $\f$ is represented by a finite $C^\infty(J^rY)$-linear
combination of elements $df$, $f\in \cS^0_r[F;Y]\subset
\cS^0_\infty[F;Y]$. Any element $\Phi\in (\cS^1_\infty[F;Y])^*$
yields a derivation $\vt_\Phi(f)=\Phi(df)$ of the  real ring
$\cS^0_\infty[F;Y]$. Since the module $\cS^1_\infty[F;Y]$ is
generated by elements $df$, $f\in \cS^0_\infty[F;Y]$, different
elements of $(\cS^1_\infty[F;Y])^*$ provide different derivations
of $\cS^0_\infty[F;Y]$, i.e., there is a monomorphism
$(\cS^0_\infty[F;Y])^*\to \gd\cS^0_\infty[F;Y]$. By the same
formula, any derivation $\vt\in \gd\cS^0_\infty[F;Y]$ sends $df\to
\vt(f)$ and, since $\cS^0_\infty[F;Y]$ is generated by elements
$df$, it defines a morphism $\Phi_\vt:\cS^1_\infty[F;Y]\to
\cS^0_\infty[F;Y]$. Moreover, different derivations $\vt$ provide
different morphisms $\Phi_\vt$. Thus, we have a monomorphism
$\gd\cS^0_\infty[F;Y]\to (\cS^1_\infty[F;Y])^*$ and, consequently,
isomorphism $\gd\cS^0_\infty[F;Y]=(\cS^0_\infty[F;Y])^*$.
\end{proof}

The proof of Theorem \ref{35t1} gives something more. It follows
that the DBGA $\cS^*_\infty[F;Y]$ is minimal differential calculus
over the real graded commutative ring $\cS^0_\infty[F;Y]$.

Let $\vt\rfloor\f$, $\vt\in \gd\cS^0_\infty[F;Y]$, $\f\in
\cS^1_\infty[F;Y]$, denote the corresponding interior product.
Extended to the DBGA $\cS^*_\infty[F;Y]$, it obeys the rule
\be
\vt\rfloor(\f\w\si)=(\vt\rfloor \f)\w\si
+(-1)^{|\f|+[\f][\vt]}\f\w(\vt\rfloor\si), \qquad \f,\si\in
\cS^*_\infty[F;Y].
\ee

Restricted to a coordinate chart (\ref{j3}) of $J^\infty Y$, the
algebra $\cS^*_\infty[F;Y]$ is a free $\cS^0_\infty[F;Y]$-module
generated by one-forms $dx^\la$, $\thh^A_\La$. Due to the
isomorphism stated in Theorem \ref{35t1}, any graded derivation
$\vt\in\gd\cS^0_\infty[F;Y]$ takes a local form
\mar{gg3}\beq
\vt=\vt^\la \dr_\la + \vt^A\dr_A + \op\sum_{0<|\La|}\vt^A_\La
\dr^\La_A, \label{gg3}
\eeq
where
\be
\dr^\La_A(s_\Si^B)=\dr^\La_A\rfloor ds_\Si^B=\dl_A^B\dl^\La_\Si
\ee
up to permutations of multi-indices $\La$ and $\Si$. Its
coefficients $\vt^\la$, $\vt^A$, $\vt^A_\La$ are local smooth
functions of finite jet order possessing the transformation law
\mar{g71}\ben
&& \vt'^\la=\frac{\dr x'^\la}{\dr x^\m}\vt^\m, \qquad
\vt'^A=\frac{\dr s'^A}{\dr s^B}\vt^B + \frac{\dr s'^A}{\dr
x^\m}\vt^\m, \nonumber\\
&& \vt'^A_\La=\op\sum_{|\Si|\leq|\La|}\frac{\dr s'^A_\La}{\dr
s^B_\Si}\vt^B_\Si + \frac{\dr s'^A_\La}{\dr x^\m}\vt^\m.
\label{g71}
\een

Every graded derivation $\vt$ (\ref{gg3}) of a ring
$\cS^0_\infty[F;Y]$ yields a derivation (called the  Lie
derivative) $\bL_\vt$ of the BGDA $\cS^*_\infty[F;Y]$ given by the
relations
\be
&& \bL_\vt\f=\vt\rfloor d\f+ d(\vt\rfloor\f), \qquad \f\in
\cS^*_\infty[F;Y],\\
&& \bL_\vt(\f\w\si)=\bL_\vt(\f)\w\si
+(-1)^{[\vt][\f]}\f\w\bL_\vt(\si),
\ee
of the DBGA $\cS^*_\infty[F;Y]$.

The graded derivation $\vt$ (\ref{gg3}) is called contact if the
Lie derivative $\bL_\vt$ preserves the ideal of contact graded
forms of the DBGA $\cS^*_\infty[F;Y]$ generated by the contact
one-forms (\ref{kk}).

\begin{lemma} \mar{nnn} \label{nnn}
With respect to the local generating basis $(s^A)$ for the DBGA
$\cS^*_\infty[F;Y]$, any its contact graded derivation takes a
form
\mar{g105}\beq
\vt=\vt_H+\vt_V=\up^\la d_\la + [\up^A\dr_A +\op\sum_{|\La|>0}
d_\La(\up^A-s^A_\m\up^\m)\dr_A^\La], \label{g105}
\eeq
where $\vt_H$ and $\vt_V$ denotes the horizontal and vertical
parts of $\vt$ \cite{cmp04,book09}.
\end{lemma}

\begin{proof}
The expression (\ref{g105}) results from a direct computation
similar to that of the first part of B\"acklund's theorem
\cite{ibr}. One can then justify that local functions (\ref{g105})
satisfy the transformation law (\ref{g71}).
\end{proof}

A glance at the expression (\ref{g105}) shows that a contact
graded derivation $\vt$ is the infinite order jet prolongation
\mar{inf}\beq
\vt=J^\infty\up \label{inf}
\eeq
of its restriction
\mar{jj15}\beq
\up=\up^\la\dr_\la +\up^A\dr_A =\up_H + \up_V= \up^\la d_\la +
(u^A\dr_A - s^A_\la\dr^\la_A) \label{jj15}
\eeq
to the graded commutative ring $S^0[F;Y]$. We call the $\up$
(\ref{jj15}) the generalized graded vector field on a graded
manifold $(Y,\gA_F)$. This fails to be a graded vector field on
$(Y,\gA_F)$ because its component depends on jets of elements of
the local generating basis for $(Y,\gA_F)$ in general. At the same
tine, any graded vector field $u$ on $(Y,\gA_F)$ is the
generalized graded vector field (\ref{jj15}) generating the
contact graded derivation $J^\infty u$ (\ref{inf}).

In particular, the vertical contact graded derivation (\ref{jj15})
reads
\mar{j40}\beq
\vt= \up^A\dr_A +\op\sum_{|\La|>0} d_\La\up^A\dr_A^\La.
\label{j40}
\eeq

\begin{lemma}
Any vertical contact graded derivation (\ref{j40}) obeys the
relations
\mar{g6,'}\ben
&& \vt\rfloor d_H\f=-d_H(\vt\rfloor\f), \label{g6}\\
&& \bL_\vt(d_H\f)=d_H(\bL_\vt\f), \qquad \f\in\cS^*_\infty[F;Y].
\label{g6'}
\een
\end{lemma}

\begin{proof}
It is easily justified that, if $\f$ and $\f'$ satisfy the
relation (\ref{g6}), then $\f\w\f'$ does well. Then it suffices to
prove the relation (\ref{g6}) when $\f$ is a function and
$\f=\thh^A_\La$. The result follows from the equalities
\mar{0480,'}\ben
&& \vt\rfloor \thh^A_\La=\up^A_\La, \quad
d_H(\up^A_\La)=\up^A_{\la+\La}dx^\la, \quad
d_H\thh^A_\la=dx^\la\w\thh^A_{\la+\La}, \label{0480}\\
&& d_\la\circ \up^A_\La\dr_A^\La= \up^A_\La\dr_A^\La \circ d_\la.
\label{0480'}
\een
The relation (\ref{g6'}) is a corollary of the equality
(\ref{g6}).
\end{proof}

The vertical contact graded derivation $\vt$ (\ref{j40}) is said
to be nilpotent if
\mar{g133}\ben
&&\bL_\vt(\bL_\vt\f)= \op\sum_{0\leq|\Si|,0\leq|\La|}
(\up^B_\Si\dr^\Si_B(\up^A_\La)\dr^\La_A + \label{g133}\\
&& \qquad (-1)^{[s^B][\up^A]}\up^B_\Si\up^A_\La\dr^\Si_B \dr^\La_A)\f=0
\nonumber
\een
for any horizontal graded form $\f\in S^{0,*}_\infty$.

\begin{lemma} \label{041} \mar{041} The vertical contact graded
derivation (\ref{j40}) is  nilpotent only if it is odd and iff the
equality
\be
\bL_\vt(\up^A)=\op\sum_{0\leq|\Si|} \up^B_\Si\dr^\Si_B(\up^A)=0
\ee
holds for all $\up^A$.
\end{lemma}

\begin{proof} There is the relation
\mar{0490}\beq
d_\la\circ \up^A_\La\dr_A^\la= \up^A_\La\dr_A^\la \circ d_\la.
\label{0490}
\eeq
Then the result follows from the equality (\ref{g133}) where one
puts $\f=s^A$ and $\f=s^A_\La s^B_\Si$.
\end{proof}

\begin{remark} \mar{zxz} \label{zxz}
If there is no danger of confusion, the common symbol $\up$
further stands for a generalized graded vector field $\up$, the
contact graded derivation $\vt$ (\ref{inf}) determined by $\up$,
and the Lie derivative $\bL_\vt$. We agree to call all these
operators, simply, a graded derivation of the structure algebra of
a Lagrangian system.
\end{remark}

\begin{remark} \label{rr35} \mar{rr5}
For the sake of convenience, right graded derivations
\mar{rgh}\beq
\op\up^\lto =\rdr_A\up^A \label{rgh}
\eeq
also are considered. They act on graded functions and differential
forms $\f$ on the right by the rules
\be
&& \op\up^\lto(\f)=d\f\lfloor\op\up^\lto +d(\f\lfloor\op\up^\lto), \\
&&\op\up^\lto(\f\w\f')=(-1)^{[\f']}\op\up^\lto(\f)\w\f'+
\f\w\op\up^\lto(\f'),\\
&& \thh_{\La A}\lfloor\rdr^{\Si B}=\dl^A_B\dl^\Si_\La.
\ee
\end{remark}

\subsection{Lagrangian symmetries and conservation laws}

Let $(\cS^*_\infty[F;Y], L)$ be a Grassmann-graded Lagrangian
system. A contact graded derivation $\vt$ (\ref{g105}) is called
the variational symmetry of a graded Lagrangian  $L$ if a Lie
derivative $\bL_\vt L$ of $L$ along $\vt$ is $d_H$-exact, i.e.,
\mar{35f1}\beq
\bL_\vt L=d_H\si. \label{35f1}
\eeq

A corollary of the variational decomposition (\ref{g99}) is the
first variational formula for a graded Lagrangian
\cite{jmp05,cmp04,book09}.

\begin{theorem} \label{j44} \mar{j44}
The Lie derivative of a graded Lagrangian along any contact graded
derivation (\ref{g105}) obeys the first variational formula
\mar{g107}\beq
\bL_\vt L= \up_V\rfloor\dl L +d_H(h_0(\vt\rfloor \Xi_L)) + d_V
(\up_H\rfloor\om)\cL, \label{g107}
\eeq
where $\Xi_L$ is the Lepage equivalent (\ref{g99'}) of a
Lagrangian $L$.
\end{theorem}

\begin{proof}
The formula (\ref{g107}) comes from the decomposition (\ref{g99})
and the relations (\ref{g6}) -- (\ref{g6'}) as follows:
\be
&& \bL_\vt L=\vt\rfloor dL + d(\vt\rfloor L) =[\vt_V\rfloor dL
-d_V\cL\w \up_H\rfloor\om] +[d_H(\up_H\rfloor L) + \\
&& \qquad d_V(\cL
\up_H\rfloor\om)]= \vt_V\rfloor dL + d_H(\up_H\rfloor L) + d_V
(\up_H\rfloor\om)\cL= \\
&& \qquad   \up_V\rfloor\dl L -\vt_V\rfloor d_H\Xi_L + d_H(\up_H\rfloor L)
+ d_V (\up_H\rfloor\om)\cL
=  \\
&& \qquad \up_V\rfloor\dl L +d_H(\vt_V\rfloor\Xi_L + \up_H\rfloor L)
+d_V (\up_H\rfloor\om)\cL,
\ee
where
\be
\vt_V\rfloor\Xi_L=h_0(\up_V\rfloor\Xi_L)
\ee
since $\Xi_L-L$ is a one-contact form and
\be
\vt_H\rfloor L=h_0(\up_H\rfloor \Xi_L).
\ee
\end{proof}

A glance at the expression (\ref{g107}) shows the following.

\begin{lemma} \mar{35l10} \label{35l10}
(i) A contact graded derivation $\vt$ is a variational symmetry
only if it is projected onto $X$.

(ii) Any projectable contact graded derivation is a variational
symmetry of a variationally trivial graded Lagrangian. It follows
that, if $\vt$ is a variational symmetry of a graded Lagrangian
$L$, it also is a variational symmetry of a Lagrangian $L+L_0$,
where $L_0$ is a variationally trivial graded Lagrangian.

(iii) A contact graded derivations $\vt$ is a variational symmetry
iff its vertical part $\up_V$ (\ref{g105}) is well.

(iv) It is a variational symmetry iff the graded density
$\up_V\rfloor \dl L$ is $d_H$-exact.
\end{lemma}

Variational symmetries of a graded Lagrangian $L$ constitute a
real vector subspace $\ccG_L$ of the graded derivation module
$\gd\cS^0_\infty[F;Y]$. By virtue of item (ii) of Lemma
\ref{35l10}, the Lie superbracket
\be
\bL_{[\vt,\vt']}= [\bL_\vt,\bL_{\vt'}]
\ee
of variational symmetries is a variational symmetry and,
therefore, their vector space $\ccG_L$ is a real Lie superalgebra.

An immediate corollary of the first variational formula
(\ref{g107}) is the following first Noether theorem.

\begin{theorem} \label{j45} \mar{j45} If the contact graded derivation $\vt$
(\ref{g105}) is a variational symmetry of a graded Lagrangian $L$,
the first variational formula (\ref{g107}) leads to the weak
conservation law
\mar{35f2}\beq
0\ap d_H(h_0(\vt\rfloor\Xi_L)-\si) \label{35f2}
\eeq
of the current
\mar{09200}\beq
\cJ_\vt=\cJ^\m_\vt\om_\m=\si- h_0(\vt\rfloor\Xi_L), \qquad
\om_\mu=\dr_\mu\rfloor\om. \label{09200}
\eeq
on the shell Ker$\,\dl L$ (\ref{eq}).
\end{theorem}

Obviously, the conserved current (\ref{09200}) is defined up to a
$d_H$-closed graded horizontal $(n-1)$-form, e.g. a total
differential $d_H U$ of some graded horizontal $(n-2)$-form
\mar{002}\beq
U=\frac12 U^{\nu\m}\om_{\nu\m}, \qquad
\om_{\nu\m}=\dr_\nu\rfloor\om_\m, \label{002}
\eeq
called the superpotential.

A variational symmetry $\vt$ of a graded Lagrangian $L$ is called
its exact symmetry or, simply, a symmetry if
\mar{22f0}\beq
\bL_\vt L=0. \label{22f0}
\eeq
In this case, the weak conservation law (\ref{35f2}) takes a form
\mar{22f5}\beq
0\ap d_H(h_0(\vt\rfloor\Xi_L))=-d_H\cJ_\vt, \label{22f5}
\eeq
where
\mar{22f6}\beq
\cJ_\vt=\cJ^\m_\vt\om_\m= -h_0(\vt\rfloor\Xi_L) \label{22f6}
\eeq
is called the symmetry current.  Of course, the symmetry current
(\ref{22f6}) also is defined with the accuracy of a $d_H$-closed
term.

Let $\vt$ be an exact symmetry of a Lagrangian $L$. Whenever $L_0$
is a variationally trivial Lagrangian, $\vt$ is a variational
symmetry of the Lagrangian $L+L_0$ such that the weak conservation
law (\ref{35f2}) for this Lagrangian is reduced to the weak
conservation law (\ref{22f5}) for a Lagrangian $L$ as follows:
\be
\bL_\vt(L+L_0)=d_H\si\ap d_H\si -d_H\cJ_\up.
\ee

\begin{remark} In accordance with the standard terminology,
variational and exact symmetries generated by generalized graded
vector fields (\ref{jj15}) are called generalized symmetries
because they depend on derivatives of variables. Generalized
symmetries of differential equations and Lagrangian systems have
been intensively investigated \cite{bry,fat,cmp04,ibr,kras,olv}.
Accordingly, by variational symmetries and symmetries one means
only those generated by vector fields $\up$ on a graded bundle
$(Y,\gA_F)$. We agree to call them classical symmetries. In this
case, the relation
\mar{22f9}\beq
\bL_\vt\cE_L=\dl(\bL_\vt L)=0 \label{22f9}
\eeq
holds \cite{book,olv}. It follows that $\vt$ also is a symmetry of
the Euler -- Lagrange operator $\cE_L$ of $L$. However, the
equality (\ref{22f9}) fails to be true in the case of generalized
symmetries.
\end{remark}

\begin{example} \mar{vvv} \label{vvv} Following Example
\ref{first}, let us consider first order Lagrangian theory on a
fibre bundle $Y\to X$. Given its Lagrangian $L$ (\ref{23f2}), let
$\up$ (\ref{jj15}) be its classical symmetry given by a
projectable vector field
\be
\up=u=u^\la\dr_\la + u^i\dr_i=u_H + u_V= u^\la d_\la + (u^i\dr_i -
y^i_\la\dr^\la_i)
\ee
on a fibre bundle $Y\to X$. In this case, it is sufficient to
consider the first order jet prolongation
\mar{23f41}\beq
J^1u= u^\la\dr_\la + u^i\dr_i + (d_\la u^i - y^i_\m d_\la
u^\m)\dr^\la_i. \label{23f41}
\eeq
of $u$ onto $J^1Y$, bun not the infinite order one (\ref{inf}).
Then the first variational formula (\ref{g107}) takes a form
\mar{23f42}\beq
\bL_{J^1u}L= u_V\rfloor\cE_L + d_H(h_0(u\rfloor H_L)),
\label{23f42}
\eeq
where $\Xi_L=H_L$ is the Poincar\'e -- Cartan form (\ref{303}).
Its coordinate expression reads
\mar{J4}\ben
&& \dr_\la u^\la\cL +[u^\la\dr_\la+
u^i\dr_i +(d_\la u^i -y^i_\m\dr_\la u^\m)\dr^\la_i]\cL =  \label{J4}\\
&& \qquad (u^i-y^i_\la u^\la)\cE_i
- d_\la[\pi^\la_i(u^\m y^i_\m -u^i) -u^\la\cL]. \nonumber
\een
If $u$ is an exact symmetry of $L$, we obtain the weak
conservation law (\ref{22f5}):
\mar{K4}\beq
0\ap - d_\la[\pi^\la_i(u^\m y^i_\m-u^i )-u^\la\cL], \label{K4}
\eeq
of the symmetry current (\ref{22f6}):
\mar{Q30}\beq
\cJ_u =[\pi^\la_i(u^\m y^i_\m-u^i )-u^\la\cL]\om_\la \label{Q30}
\eeq
along a vector field $u$.
\end{example}

\subsection{Gauge symmetries}

Treating gauge symmetries of Lagrangian field theory, one is
traditionally based on an example of Yang -- Mills gauge theory of
principal connections on principal bundles (Section 6.1). This
notion of gauge symmetries has been generalized to Lagrangian
theory on an arbitrary fibre bundle $Y\to X$
\cite{jpa05,jmp09,book09}. Here, we extend this notion to
Lagrangian theory on graded bundles in a general setting
(Definition \ref{sgauge}).

Let $(\cS^*_\infty[F;Y], L)$ be a Grassmann-graded Lagrangian
system on a graded bundle $(X,Y, \gA_F)$ with the local generating
basis $(s^A)$. Let
\be
E=E^0\op\oplus_X E^1
\ee
be a graded vector bundle over $X$ possessing an even part $E^0\to
X$ and the odd one $E^1\to X$. We regard it as a composite bundle
\mar{psp}\beq
E\to E^0\to X \label{psp}
\eeq
and consider a graded bundle $(X,E^0,\gA_E)$ modelled over this
composite bundle. Then we define the product (\ref{tyt}) of graded
bundles $(X,Y, \gA_F)$ and $(X,E^0,\gA_E)$ over the product
(\ref{wqw}) of the composite bundles $E$ (\ref{psp}) and $F$
(\ref{su5}). It reads
\mar{olo}\beq
(X, E^0\op\times_X Y, \gA_{E\op\times_X F}) \label{olo}
\eeq
Let us consider the corresponding DBGA
\mar{tyt'}\beq
\cS^*_\infty[E\op\times_XF;E^0\op\times_X Y] \label{tyt'}
\eeq
together with the monomorphisms (\ref{opo}) of DBGAs
\mar{tgv}\beq
\cS^*_\infty[F;Y]\to \cS^*_\infty[E\op\times_XF;E^0\op\times_X Y],
\qquad \cS^*_\infty[E;E^0]\to
\cS^*_\infty[E\op\times_XF;E^0\op\times_X Y] \label{tgv}
\eeq

Given a Lagrangian $L\in \cS^{0,n}_\infty[F;Y]$, let us define its
pull-back
\mar{kjk}\beq
L\in \cS^{0,n}_\infty[F;Y]\subset
\cS^*_\infty[E\op\times_XF;E^0\op\times_X Y], \label{kjk}
\eeq
and consider an extended Lagrangian system
\mar{mkm}\beq
(\cS^*_\infty[E\op\times_XF;E^0\op\times_X Y],L) \label{mkm}
\eeq
provided with the local generating basis $(s^A,c^r)$.

\begin{definition} \mar{sgauge} \label{sgauge}
A gauge transformation of the Lagrangian $L$ (\ref{kjk}) is
defined to be the contact graded derivation $\vt$ (\ref{inf}) of
the ring $\cS^0_\infty[E\op\times_XF;E^0\op\times_X Y]$
(\ref{tyt'}) such that a derivation $\vt$ equals zero on a subring
\be
\cS^0_\infty[E;E^0]\subset
\cS^0_\infty[E\op\times_XF;E^0\op\times_X Y],
\ee
A gauge transformation $\vt$ is called the gauge symmetry if it is
a variational symmetry of the Lagrangian $L$ (\ref{kjk}), i.e., a
density $\up_V\rfloor \dl L$ is $d_H$-exact.
\end{definition}

In view of the fist condition in Definition \ref{sgauge}, the
variables $c^r$ of the extended Lagrangian system (\ref{mkm}) can
be treated as gauge parameters of a gauge symmetry $\vt$.

Furthermore, we additionally assume that a gauge symmetry $\vt$ is
linear in gauge parameters $c^r$ and their jets $c^r_\La$ (see
Remark \ref{s8}). Then the generalized graded vector field $\up$
(\ref{jj15}) generating a gauge symmetry $\vt$ reads
\mar{gg2}\beq
\up=\left(\op\sum_{0\leq|\La|\leq m}
\up^{\la\La}_r(x^\m)c^r_\La\right)\dr_\la +
\left(\op\sum_{0\leq|\La|\leq m}
\up^{A\La}_r(x^\m,s^B_\Si)c^r_\La\right)\dr_A. \label{gg2}
\eeq
In accordance with Remark \ref{zxz}, we also call it the gauge
symmetry.

By virtue of item (iii) of Lemma \ref{35l10}, the generalized
vector field $\up$ (\ref{gg2}) is a gauge symmetry iff its
vertical part is so. Therefore, we can restrict our consideration
to vertical gauge symmetries.

\begin{remark} \mar{qwe} \label{qwe}
Let $E=E^0$, i.e., gauge parameters are even. A glance at the
expression (\ref{gg2}) shows that, in this case, a gauge symmetry
$\up$ is a linear differential operator on a real space of
sections of a vector bundle $E^0\to X$ with values in a real space
$\ccG_L$ of variational symmetries of a Lagrangian $L$
\cite{jpa05,jmp05,book09}.
\end{remark}

\subsection{Gauge conservation laws}

Being a variational symmetry, the gauge symmetry  $\up$
(\ref{gg2}) defines the weak conservation law (\ref{35f2}). The
peculiarity of this conservation law is that the conserved current
$\cJ_\up$ (\ref{09200}) is reduced to the superpotential
(\ref{002}) as follows.

\begin{theorem} \mar{supp} \label{supp} If $\up$ (\ref{gg2}) is a gauge
symmetry of a Lagrangian $L$, the corresponding conserved current
$\cJ_\up$ (\ref{09200}) is linear in gauge parameters (up to a
$d_h$-closed term), and it takes a form
\mar{005}\beq
\cJ_\up=W +d_HU=(W^\m + d_\nu U^{\nu\m})\om_\m, \label{005}
\eeq
where a term $W$ vanishes on-shell and $U^{\nu\m}=-U^{\m\nu}$ is
the superpotential (\ref{002}) which takes the form (\ref{spp}).
\end{theorem}

\begin{proof}
Let the gauge symmetry $\up$ (\ref{gg2}) be at most of jet order
$N$ in parameters. Then the conserved current $\cJ_\up$
(\ref{09200}), being linear in gauge parameters, is decomposed
into a sum
\mar{g2g}\ben
&& \cJ_\up^\m= J^{\m\m_1\ldots\m_M}_rc^r_{\m_1\ldots\m_M} +
\op\sum_{1<k< M} J^{\m\m_k\ldots\m_M}_rc^r_{\m_k\ldots\m_M} + \label{g2g}\\
&& \qquad J^{\m\m_M}_rc^r_{\m_M} +J^\m_rc^r, \qquad N\leq
M,\nonumber
\een
and the first variational formula (\ref{g107}) takes a form
\be
&& 0=[ \op\sum_{k=1}^N
\up_V^A{}_r^{\m_k\ldots\m_N}c^r_{\m_k\ldots\m_N}
+\up_V^A{}_rc^r]\cE_A -\\
&& \qquad d_\m(\op\sum_{k=1}^M
J^{\m\m_k\ldots\m_M}_rc^r_{\m_k\ldots\m_M} +J^\m_rc^r).
\ee
This equality provides the following set of equalities for each
$c^r_{\m\m_1\ldots\m_M}$, $c^r_{\m_k\ldots\m_M}$
$(k=1,\ldots,M-N-1)$, $c^r_{\m_k\ldots\m_N}$ $(k=1,\ldots,N-1)$,
$c^r_\m$ and $c^r$:
\mar{g4g,-6}\ben
&& 0=J^{(\m\m_1)\ldots\m_M}_r, \label{g4g}\\
&& 0=J^{(\m_k\m_{k+1})\ldots\m_M}_r +d_\nu
J^{\nu\m_k\ldots\m_M}_r, \qquad 1\leq k<M-N,
\label{g4g'}\\
&& 0=\up_V^A{}_r^{\m_k\ldots\m_N}\cE_A-
J^{(\m_k\m_{k+1})\ldots\m_N}_r
-d_\nu J^{\nu\m_k\ldots\m_N}_r,\qquad 1\leq k<N, \label{g5g}\\
&& 0= \up_V^A{}_r^\m\cE_A - J^\m_r - d_\nu J^{\nu\m}_r, \label{g6g}
\een
where $(\m\nu)$ means symmetrization of indices in accordance with
the splitting
\be
J^{\m_k\m_{k+1}\ldots\m_N}_r=J^{(\m_k\m_{k+1})\ldots\m_N}_r+
J^{[\m_k\m_{k+1}]\ldots\m_N}_r.
\ee
We also have the equality
\mar{g7g}\beq
0= \up_V^A{}_r\cE_A - d_\m J^\m_r, \label{g7g}
\eeq
With the equalities (\ref{g4g}) -- (\ref{g6g}), the decomposition
(\ref{g2g}) takes a form
\be
&& \cJ_\up^\m= J^{[\m\m_1]\ldots\m_M}_rc^r_{\m_1\ldots\m_M} +\\
&& \qquad \op\sum_{1< k\leq M-N} [(J^{[\m\m_k]\ldots\m_M}_r -
d_\nu
J^{\nu\m\m_k\ldots\m_M}_r)c^r_{\m_k\ldots\m_M}]+ \\
&& \qquad \op\sum_{1<k< N}[(\up_V^iA{}_r^{\m\m_k\ldots\m_N}\cE_A -
d_\nu J^{\nu\m\m_k\ldots\m_N}_r +
 J^{[\m\m_k]\ldots\m_N}_r)c^r_{\m_k\ldots\m_N}]+\\
&& \qquad (\up_V^A{}_r^{\m\m_N}\cE_A -d_\nu J^{\nu\m\m_N}_r +
J^{[\m\m_N]}_r)c^r_{\m_N} + (\up_V^A{}_r^\m\cE_A - d_\nu
J^{\nu\m}_r)c^r.
\ee
A direct computation
\be
&& \cJ_\up^\m=
d_\nu(J^{[\m\nu]\m_2\ldots\m_M}_rc^r_{\m_2\ldots\m_M}) -
d_\nu J^{[\m\nu]\m_2\ldots\m_M}_rc^r_{\m_2\ldots\m_M}+\\
&&  \qquad \op\sum_{1< k\leq M-N}
[d_\nu(J^{[\m\nu]\m_{k+1}\ldots\m_M}_rc^r_{\m_{k+1}\ldots\m_M})
-\\
&&\qquad d_\nu
J^{[\m\nu]\m_{k+1}\ldots\m_M}_rc^r_{\m_{k+1}\ldots\m_M}- d_\nu
J^{\nu\m\m_k\ldots\m_M}_rc^r_{\m_k\ldots\m_M}]+\\
&& \qquad \op\sum_{1<k<N} [(\up_V^A{}_r^{\m\m_k\ldots\m_N}\cE_A
- d_\nu J^{\nu\m\m_k\ldots\m_N}_r)c^r_{\m_k\ldots\m_N} +\\
&&  \qquad
d_\nu(J^{[\m\nu]\m_{k+1}\ldots\m_N}_rc^r_{\m_{k+1}\ldots\m_N})
-d_\nu J^{[\m\nu]\m_{k+1}\ldots\m_N}_rc^r_{\m_{k+1}\ldots\m_N}]+\\
&& \qquad [(\up_V^A{}_r^{\m\m_N}\cE_A - d_\nu
J^{\nu\m\m_N}_r)c^r_{\m_N} + d_\nu (J^{[\m\nu]}_rc^r)
- d_\nu J^{[\m\nu]}_rc^r] +\\
&& \qquad (\up_V^A{}_r^\m\cE_A - d_\nu J^{\nu\m}_r)c^r\\
&& = d_\nu(J^{[\m\nu]\m_2\ldots\m_M}_rc^r_{\m_2\ldots\m_M})+\\
&&  \qquad \op\sum_{1< k\leq M-N}
[d_\nu(J^{[\m\nu]\m_{k+1}\ldots\m_M}_rc^r_{\m_{k+1}\ldots\m_M}) -
d_\nu
J^{(\nu\m)\m_k\ldots\m_M}_rc^r_{\m_k\ldots\m_M}]+\\
&& \qquad \op\sum_{1<k<N} [(\up_V^A{}_r^{\m\m_k\ldots\m_N}\cE_A
- d_\nu J^{(\nu\m)\m_k\ldots\m_N}_r)c^r_{\m_k\ldots\m_N} +\\
&&  \qquad d_\nu(J^{[\m\nu]\m_{k+1}\ldots\m_N}_rc^r_{\m_{k+1}\ldots\m_N})]+\\
&& \qquad [(\up_V^A{}_r^{\m\m_N}\cE_A - d_\nu
J^{(\nu\m)\m_N}_r)c^r_{\m_N} + d_\nu (J^{[\m\nu]}_rc^r)]
+(\up_V^A{}_r^\m\cE_A - d_\nu J^{(\nu\m)}_r)c^r
\ee
leads to the expression
\mar{g8g}\ben
&& \cJ_\up^\m=(\op\sum_{1<k\leq N}\up_V^i{}_r^{\m\m_k\ldots\m_N}
c^r_{\m_k\ldots\m_N}+ \up_V^A{}_r^\m c^r)\cE_A -\label{g8g}\\
&& \qquad (\op\sum_{1<k\leq M}d_\nu
J^{(\nu\m)\m_k\ldots\m_M}c^r_{\m_k\ldots\m_M}+ d_\nu
J^{(\nu\m)}_rc^r)- \nonumber\\
&& \qquad d_\nu(\op\sum_{1<k\leq
M}J^{[\nu\m]\m_k\ldots\m_M}c^r_{\m_k\ldots\m_M} +
J^{[\nu\m]}_rc^r).\nonumber
\een
The first summand of this expression vanishes on-shell. Its second
one contains the terms $d_\nu J^{(\nu\m_k)\m_{k+1}\ldots\m_M}$,
$k=1,\ldots, M$. By virtue of the equalities (\ref{g4g'}) --
(\ref{g5g}), every $d_\nu J^{(\nu\m_k)\m_{k+1}\ldots\m_M}$ is
expressed into the terms vanishing on-shell and the term $d_\nu
J^{(\nu\m_{k-1})\m_k\ldots\m_M}$. Iterating the procedure and
bearing in mind the equality (\ref{g4g}), one can easily show that
the second summand of the expression (\ref{g8g}) also vanishes
on-shell. Thus a symmetry current takes the form (\ref{005}) where
\mar{spp}\beq
U^{\nu\mu}=- \op\sum_{1<k\leq
M}J^{[\nu\m]\m_k\ldots\m_M}c^r_{\m_k\ldots\m_M} -
J^{[\nu\m]}_rc^r). \label{spp}
\eeq
\end{proof}

\begin{example} \mar{xcx} \label{xcx}
If a gauge symmetry
\mar{g20g}\beq
\up=(\up^\la_r c^r + \up^{\la\m}_r c^r_\m)\dr_\la + (\up^A_r c^r +
\up^{A\m}_r c^r_\m)\dr_A \label{g20g}
\eeq
of a graded Lagrangian $L$ depends at most on the first jets of
gauge parameters, then the decomposition (\ref{g8g}) takes a form
\mar{g21g}\ben
&& \cJ_\up^\m=\up_V^A{}_r^\m \cE_A c^r -d_\nu(J^{[\nu\m]}_r c^r) =
\label{g21g}\\
&& \qquad
(\up^{i\m}_ar- s^A_\la \up^{\la\m}_r)c^r\cE_A +
d_\nu[(\up^{A[\m}_r- s^A_\la \up^{\la[\m}_r)c^r\dr^{\nu]}_A\cL+
\up^{[\nu\m]}_rc^r\cL].\nonumber
\een
\end{example}

\section{Second Noether theorems}

Let $(\cS^*_\infty[F;Y],L)$ be a Grassmann-graded Lagrangian
system. Describing Noether and higher-stage identities of its
Euler -- Lagrange operator, we follow the general analysis of
Noether and higher-stage Noether identities of differential
operators on fibre bundles in Section 5.5. In the case of an Euler
-- Lagrange operator as a variation of a Lagrangian, one can
however formulate the second Noether theorems (Theorems \ref{w35}
-- \ref{825}) which associate to these identities the gauge and
higher-stage gauge symmetries of a Lagrangian system
\cite{book09,sard13}.

\subsection{Noether and higher-stage Noether identities}

Without a lose of generality, let a Lagrangian $L$ be even.

Its Euler -- Lagrange operator $\dl L$ (\ref{0709'}) is assumed to
be at least of order 1 in order to guarantee that transition
functions of $Y$ do not vanish on-shell. This Euler -- Lagrange
operator $\dl L\in \bE_1\subset \cS^{1,n}_\infty[F;Y]$ takes its
values into the graded vector bundle
\mar{41f33}\beq
\ol{VF}=V^*F\op\ot_F\op\w^n T^*X\to F, \label{41f33}
\eeq
where $V^*F$ is the vertical cotangent bundle of $F\to X$. It
however is not a vector bundle over $Y$. Therefore, we restrict
our consideration to the case of the pull-back composite bundle
$F$ (\ref{su5}) that is
\mar{41f1}\beq
F=Y\op\times_X F^1\to Y\to X, \label{41f1}
\eeq
where $F^1\to X$ is a vector bundle.

\begin{remark} \mar{rtr} \label{rtr}
Let us introduce the following notation. Given the vertical
tangent bundle $VE$ of a fibre bundle $E\to X$, by its
density-dual bundle is meant the fibre bundle
\mar{41f2}\beq
\ol{VE}=V^*E\op\ot_E \op\w^n T^*X. \label{41f2}
\eeq
If $E\to X$ is a vector bundle, we have
\mar{41f3}\beq
\ol{VE}=\ol E\op\times_X E, \qquad \ol E=E^*\op\ot_X\op\w^n T^*X,
\label{41f3}
\eeq
where $\ol E$ is called the density-dual of $E$. Let
\mar{grd}\beq
E=E^0\op\oplus_X E^1 \label{grd}
\eeq
be a graded vector bundle over $X$. Its graded density-dual is
defined to be
\be
\ol E=\ol E^1\op\oplus_X \ol E^0
\ee
with an even part $\ol E^1\to X$ and the odd one $\ol E^1\to X$.
Given the graded vector bundle $E$ (\ref{grd}), we consider the
product (\ref{olo}) of graded bundles over the product (\ref{wqw})
of the composite bundles $E$ (\ref{psp}) and $F$ (\ref{su5}) and
the corresponding DBGA (\ref{tyt'}) which we denote:
\mar{41f5'}\beq
P^*_\infty[F\op\times_XE;Y]=\cS^*_\infty[F\op\times_XE;Y\op\times_X
E^0]. \label{41f5'}
\eeq
In particular, we treat the composite bundle $F$ (\ref{su5}) as a
graded vector bundle over $Y$ possessing only an odd part. The
density-dual (\ref{41f2}) of the vertical tangent bundle $VF$ of
$F\to X$ is $\ol{VF}$ (\ref{41f33}). If $F$ is the pull-back
bundle (\ref{41f1}), then
\mar{41f4}\beq
\ol{VF}=\ol F^1\op\oplus_Y ((V^*Y\op\ot_Y\op\w^n T^*X)\op\oplus_Y
F^1 )\label{41f4}
\eeq
is a graded vector bundle over $Y$. This bundle can be seen as the
product (\ref{wqw}) of composite bundles
\be
\ol{VF^1}=\ol F^1\op\oplus_X F^1\to \ol F^1\to X, \qquad
\ol{VY}\to Y \to X,
\ee
and we consider the corresponding graded bundle (\ref{tyt}) and
the DBGA (\ref{bvb}) which we denote
\mar{41f5}\beq
\cP^*_\infty [\ol {VF};Y]=\cS^*_\infty [\ol {VF};Y\op\times_X\ol
F^1]=\cS^*_\infty [\ol{VF}^1\op\times_X \ol{VY}; Y\op\times_X\ol
F^1]. \label{41f5}
\eeq
\end{remark}

\begin{lemma} \label{41l1} \mar{41l1} One can associate to any
Grassmann-graded Lagrangian system $(\cS^*_\infty[F;Y],L)$ the
chain complex (\ref{v042}) whose one-boundaries vanish on-shell.
\end{lemma}

\begin{proof} Let us consider the density-dual $\ol{VF}$ (\ref{41f4})
of the vertical tangent bundle $VF\to F$, and let us enlarge the
original DBGA $\cS^*_\infty[F;Y]$ to the DBGA
$\cP^*_\infty[\ol{VF};Y]$ (\ref{41f5}) with the local generating
basis
\be
(s^A, \ol s_A), \qquad [\ol s_A]=([A]+1){\rm mod}\,2.
\ee
Following the terminology of Lagrangian BRST theory
\cite{barn,gom}, we agree to call its elements $\ol s_A$ the
Noether antifields of antifield number Ant$[\ol s_A]= 1$. The DBGA
$\cP^*_\infty[\ol{VF};Y]$ is endowed with the nilpotent right
graded derivation
\mar{41f6}\beq
\ol\dl=\rdr^A \cE_A, \label{41f6}
\eeq
where $\cE_A$ are the variational derivatives (\ref{0709'}). Then
we have a chain complex
\mar{v042}\beq
0\lto \im\ol\dl \llr^{\ol\dl} \cP^{0,n}_\infty[\ol{VF};Y]_1
\llr^{\ol\dl} \cP^{0,n}_\infty[\ol{VF};Y]_2 \label{v042}
\eeq
of graded densities of antifield number $\leq 2$. Its
one-boundaries $\ol\dl\Phi$, $\Phi\in
\cP^{0,n}_\infty[\ol{VF};Y]_2$, by very definition, vanish
on-shell.
\end{proof}

Any one-cycle
\mar{0712}\beq
\Phi= \op\sum_{0\leq|\La|} \Phi^{A,\La}\ol s_{\La A} \om \in
\cP^{0,n}_\infty[\ol{VF};Y]_1\label{0712}
\eeq
of the complex (\ref{v042}) is a differential operator on a bundle
$\ol{VF}$ such that it is linear on fibres of $\ol{VF}\to F$ and
its kernel contains the graded Euler -- Lagrange operator $\dl L$
(\ref{0709'}), i.e.,
\mar{0713}\beq
\ol\dl\Phi=0, \qquad \op\sum_{0\leq|\La|} \Phi^{A,\La}d_\La
\cE_A\om=0. \label{0713}
\eeq
Refereing to Definition \ref{46d1} of Noether identities of a
differential operator in Section 5.5, one can say that the
one-cycles (\ref{0712}) define the Noether identities (\ref{0713})
of an Euler -- Lagrange operator $\dl L$, which we agree to call
the Noether identities of a Grassmann-graded Lagrangian system
$(\cS^*_\infty[F;Y],L)$.

In particular, one-chains $\Phi$ (\ref{0712}) are necessarily
Noether identities if they are boundaries. Therefore, these
Noether identities are called trivial.  They are of the form
\be
&& \Phi= \op\sum_{0\leq|\La|,|\Si|} T^{(A\La)(B\Si)}d_\Si\cE_B\ol
s_{\La A}\om,\\
&& T^{(A\La)(B\Si)}=-(-1)^{[A][B]} T^{(B\Si)(A\La)}.
\ee
Accordingly, non-trivial Noether identities modulo the trivial
ones are associated to elements of the first homology
$H_1(\ol\dl)$ of the complex (\ref{v042}). A Lagrangian $L$ is
called degenerate if there are non-trivial Noether identities.

Non-trivial Noether identities can obey first-stage Noether
identities. In order to describe them, let us assume that a module
$H_1(\ol \dl)$ is finitely generated. Namely, there exists a
graded projective $C^\infty(X)$-module $\cC_{(0)}\subset H_1(\ol
\dl)$ of finite rank possessing a local basis $\{\Delta_r\om\}$:
\mar{41f7}\beq
\Delta_r\om=\op\sum_{0\leq|\La|} \Delta_r^{A,\La}\ol s_{\La
A}\om,\qquad \Delta_r^{A,\La}\in \cS^0_\infty[F;Y], \label{41f7}
\eeq
such that any element $\Phi\in H_1(\ol \dl)$ factorizes as
\mar{xx2}\beq
\Phi= \op\sum_{0\leq|\Xi|} \Phi^{r,\Xi} d_\Xi \Delta_r \om, \qquad
\Phi^{r,\Xi}\in \cS^0_\infty[F;Y], \label{xx2}
\eeq
through elements (\ref{41f7}) of $\cC_{(0)}$. Thus, all
non-trivial Noether identities (\ref{0713}) result from the
Noether identities
\mar{v64}\beq
\ol\dl\Delta_r= \op\sum_{0\leq|\La|} \Delta_r^{A,\La} d_\La
\cE_A=0, \label{v64}
\eeq
called the complete Noether identities. Clearly, the factorization
(\ref{xx2}) is independent of specification of a local basis
$\{\Delta_r\om\}$. Note that, being representatives of $H_1(\ol
\dl)$, the graded densities $\Delta_r\om$ (\ref{41f7}) are not
$\ol\dl$-exact.

A Lagrangian system whose non-trivial Noether identities are
finitely generated is called finitely degenerate. Hereafter,
degenerate Lagrangian systems only of this type are considered.

\begin{lemma} \label{41l2} \mar{41l2}
If the homology $H_1(\ol\dl)$ of the complex (\ref{v042}) is
finitely generated in the above mentioned sense, this complex can
be extended to the one-exact chain complex (\ref{v66}) with a
boundary operator whose nilpotency conditions are equivalent to
the complete Noether identities (\ref{v64}).
\end{lemma}

\begin{proof}
By virtue of Serre -- Swan Theorem \ref{vv0}, a graded module
$\cC_{(0)}$ is isomorphic to a module of sections of the
density-dual $\ol E_0$ of some graded vector bundle $E_0\to X$.
Let us enlarge $\cP^*_\infty[\ol{VF};Y]$ to the DBGA
\mar{41f14}\beq
\ol\cP^*_\infty\{0\}=\cP^*_\infty[\ol{VF}\op\times_X \ol E_0;Y]=
\cS^*_\infty[\ol{VF}\op\times_X \ol E_0;\ol E^1_0\op\times_X\ol
F^1\op\times_X Y] \label{41f14}
\eeq
possessing the local generating basis $(s^A,\ol s_A, \ol c_r)$
where $\ol c_r$ are Noether antifields of Grassmann parity
\be
[\ol c_r]=([\Delta_r]+1){\rm mod}\,2
\ee
and antifield number ${\rm Ant}[\ol c_r]=2$. The DBGA
(\ref{41f14}) is provided with the odd right graded derivation
\mar{41f10}\beq
\dl_0=\ol\dl + \rdr^r\Delta_r \label{41f10}
\eeq
which is nilpotent iff the complete Noether identities (\ref{v64})
hold. Then $\dl_0$ (\ref{41f10}) is a boundary operator of a chain
complex
\mar{v66}\beq
0\lto \im\ol\dl \op\lto^{\ol\dl}
\cP^{0,n}_\infty[\ol{VF};Y]_1\op\lto^{\dl_0}
\ol\cP^{0,n}_\infty\{0\}_2 \op\lto^{\dl_0}
\ol\cP^{0,n}_\infty\{0\}_3 \label{v66}
\eeq
of graded densities of antifield number $\leq 3$. Let $H_*(\dl_0)$
denote its homology. We have
\be
H_0(\dl_0)=H_0(\ol\dl)=0.
\ee
Furthermore, any one-cycle $\Phi$ up to a boundary takes the form
(\ref{xx2}) and, therefore, it is a $\dl_0$-boundary
\be
\Phi= \op\sum_{0\leq|\Si|} \Phi^{r,\Xi} d_\Xi \Delta_r\om
=\dl_0\left(\op\sum_{0\leq|\Si|} \Phi^{r,\Xi}\ol c_{\Xi
r}\om\right).
\ee
Hence, $H_1(\dl_0)=0$, i.e., the complex (\ref{v66}) is one-exact.
\end{proof}

Let us consider the second homology $H_2(\dl_0)$ of the complex
(\ref{v66}). Its two-chains  read
\mar{41f9}\beq
\Phi= G + H= \op\sum_{0\leq|\La|} G^{r,\La}\ol c_{\La r}\om +
\op\sum_{0\leq|\La|,|\Si|} H^{(A,\La)(B,\Si)}\ol s_{\La A}\ol
s_{\Si B}\om. \label{41f9}
\eeq
Its two-cycles define the first-stage Noether identities
\mar{v79}\beq
\dl_0 \Phi=0, \qquad   \op\sum_{0\leq|\La|} G^{r,\La}d_\La\Delta_r
\om =-\ol\dl H. \label{v79}
\eeq
Conversely, let the equality (\ref{v79}) hold. Then it is a cycle
condition of the two-chain (\ref{41f9}).

\begin{remark}
Note that this definition of first-stage Noether identities is
independent on specification of a generating module $\cC_{(0)}$.
Given a different one, there exists a chain isomorphism between
the corresponding complexes (\ref{v66}).
\end{remark}

The first-stage Noether identities (\ref{v79}) are trivial either
if a two-cycle $\Phi$ (\ref{41f9}) is a $\dl_0$-boundary or its
summand $G$ vanishes on-shell. Therefore, non-trivial first-stage
Noether identities fails to exhaust the second homology
$H_2(\dl_0)$ of the complex (\ref{v66}) in general.

\begin{lemma} \label{v134'} \mar{v134'}
Non-trivial first-stage Noether identities modulo the trivial ones
are identified with elements of the homology $H_2(\dl_0)$ iff any
$\ol\dl$-cycle $\f\in \ol\cP^{0,n}_\infty\{0\}_2$ is a
$\dl_0$-boundary.
\end{lemma}

\begin{proof}
It suffices to show that, if a summand $G$ of a two-cycle $\Phi$
(\ref{41f9}) is $\ol\dl$-exact, then $\Phi$ is a boundary. If
$G=\ol\dl \Psi$, let us write
\mar{v169'}\beq
\Phi=\dl_0\Psi +(\ol \dl-\dl_0)\Psi + H. \label{v169'}
\eeq
Hence, the cycle condition (\ref{v79}) reads
\be
\dl_0\Phi=\ol\dl((\ol\dl-\dl_0)\Psi + H)=0.
\ee
Since any $\ol\dl$-cycle $\f\in \ol\cP^{0,n}_\infty\{0\}_2$, by
assumption, is $\dl_0$-exact, then
\be
(\ol \dl-\dl_0)\Psi + H
\ee
is a $\dl_0$-boundary. Consequently, $\Phi$ (\ref{v169'}) is
$\dl_0$-exact. Conversely, let $\Phi\in
\ol\cP^{0,n}_\infty\{0\}_2$ be a $\ol\dl$-cycle, i.e.,
\be
\ol\dl\Phi= 2\Phi^{(A,\La)(B,\Sigma)}\ol s_{\La A} \ol\dl\ol
s_{\Sigma B}\om= 2\Phi^{(A,\La)(B,\Sigma)}\ol s_{\La A} d_\Si
\cE_B\om=0.
\ee
It follows that
\be
\Phi^{(A,\La)(B,\Sigma)} \ol\dl\ol s_{\Sigma B}=0
\ee
for all indices $(A,\La)$. Omitting a $\ol\dl$-boundary term, we
obtain
\be
\Phi^{(A,\La)(B,\Sigma)} \ol s_{\Sigma B}= G^{(A,\La)(r,\Xi)}d_\Xi
\Delta_r.
\ee
Hence, $\Phi$ takes a form
\be
\Phi=G'^{(A,\La)(r,\Xi)} d_\Xi\Delta_r \ol s_{\La A}\om.
\ee
Then there exists a three-chain
\be
\Psi= G'^{(A,\La)(r,\Xi)} \ol c_{\Xi r} \ol s_{\La A}\om
\ee
such that
\mar{41f12}\beq
\dl_0\Psi=\Phi +\si = \Phi + G''^{(A,\La)(r,\Xi)}d_\La\cE_A \ol
c_{\Xi r} \om. \label{41f12}
\eeq
Owing to the equality $\ol\dl\Phi=0$, we have $\dl_0\si=0$. Thus,
$\si$ in the expression (\ref{41f12}) is $\ol\dl$-exact
$\dl_0$-cycle. By assumption, it is $\dl_0$-exact, i.e.,
$\si=\dl_0\psi$.   Consequently, a $\ol\dl$-cycle $\Phi$ is a
$\dl_0$-boundary
\be
\Phi=\dl_0\Psi -\dl_0\psi.
\ee
\end{proof}

\begin{remark}
It is easily justified that the two-cycle $\Phi$ (\ref{41f9}) is
$\dl_0$-exact iff $\Phi$ up to a $\ol\dl$-boundary takes a form
\be
\Phi= \op\sum_{0\leq |\La|, |\Si|} G'^{(r,\Si)(r',\La)}
d_\Si\Delta_r d_\La\Delta_{r'}\om.
\ee
\end{remark}

A degenerate Lagrangian system is called reducible if it admits
non-trivial first stage Noether identities.

If the condition of Lemma \ref{v134'} is satisfied, let us assume
that non-trivial first-stage Noether identities are finitely
generated as follows. There exists a graded projective
$C^\infty(X)$-module $\cC_{(1)}\subset H_2(\dl_0)$ of finite rank
possessing a local basis $\{\Delta_{r_1}\om\}$:
\mar{41f13}\beq
\Delta_{r_1}\om=\op\sum_{0\leq|\La|} \Delta_{r_1}^{r,\La}\ol
c_{\La r}\om + h_{r_1}\om,   \label{41f13}
\eeq
such that any element $\Phi\in H_2(\dl_0)$ factorizes as
\mar{v80'}\beq
\Phi= \op\sum_{0\leq|\Xi|} \Phi^{r_1,\Xi} d_\Xi \Delta_{r_1}\om,
\qquad \Phi^{r_1,\Xi}\in \cS^0_\infty[F;Y], \label{v80'}
\eeq
through elements (\ref{41f13}) of $\cC_{(1)}$. Thus, all
non-trivial first-stage Noether identities (\ref{v79}) result from
the equalities
\mar{v82'}\beq
 \op\sum_{0\leq|\La|} \Delta_{r_1}^{r,\La} d_\La \Delta_r +\ol\dl
h_{r_1} =0, \label{v82'}
\eeq
called the complete first-stage Noether identities.  Note that, by
virtue of the condition of Lemma \ref{v134'}, the first summands
of the graded densities $\Delta_{r_1}\om$ (\ref{41f13}) are not
$\ol\dl$-exact.

A degenerate Lagrangian system is called finitely reducible if it
admits finitely generated non-trivial first-stage Noether
identities.

\begin{lemma} \label{v139'} \mar{v139'} The one-exact complex
(\ref{v66}) of a finitely reducible Lagrangian system is extended
to the two-exact one (\ref{v87'}) with a boundary operator whose
nilpotency conditions are equivalent to the complete Noether
identities (\ref{v64}) and the complete first-stage Noether
identities (\ref{v82'}).
\end{lemma}

\begin{proof}
By virtue of Serre -- Swan Theorem \ref{vv0}, a graded module
$\cC_{(1)}$ is isomorphic to a module of sections of the
density-dual $\ol E_1$ of some graded vector bundle $E_1\to X$.
Let us enlarge the DBGA $\ol\cP^*_\infty\{0\}$ (\ref{41f14}) to
the DBGA
\be
\ol\cP^*_\infty\{1\}=\cP^*_\infty[\ol{VF}\op\times_X \ol
E_0\op\times_X\ol E_1;Y]
\ee
possessing the local generating basis $\{s^A,\ol s_A, \ol c_r, \ol
c_{r_1}\}$ where $\ol c_{r_1}$ are  first stage Noether antifields
of Grassmann parity
\be
[\ol c_{r_1}]=([\Delta_{r_1}]+1){\rm mod}\,2
\ee
and antifield number Ant$[\ol c_{r_1}]=3$. This DBGA is provided
with the odd right graded derivation
\mar{v205'}\beq
\dl_1=\dl_0 + \rdr^{r_1} \Delta_{r_1} \label{v205'}
\eeq
which is nilpotent iff the complete Noether identities (\ref{v64})
and the complete first-stage Noether identities (\ref{v82}) hold.
Then $\dl_1$ (\ref{v205'}) is a boundary operator of the chain
complex
\mar{v87'}\beq
0\lto \im\ol\dl \op\lto^{\ol\dl}
\cP^{0,n}_\infty[\ol{VF};Y]_1\op\lto^{\dl_0}
\ol\cP^{0,n}_\infty\{0\}_2 \op\lto^{\dl_1}
\ol\cP^{0,n}_\infty\{1\}_3 \op\lto^{\dl_1}
\ol\cP^{0,n}_\infty\{1\}_4 \label{v87'}
\eeq
of graded densities of antifield number $\leq 4$. Let $H_*(\dl_1)$
denote its homology. It is readily observed that
\be
H_0(\dl_1)=H_0(\ol\dl), \qquad H_1(\dl_1)=H_1(\dl_0)=0.
\ee
By virtue of the expression (\ref{v80'}), any two-cycle of the
complex (\ref{v87'}) is a boundary
\be
 \Phi= \op\sum_{0\leq|\Xi|} \Phi^{r_1,\Xi} d_\Xi \Delta_{r_1}\om
=\dl_1\left(\op\sum_{0\leq|\Xi|} \Phi^{r_1,\Xi} \ol c_{\Xi
r_1}\om\right).
\ee
It follows that $H_2(\dl_1)=0$, i.e., the complex (\ref{v87'}) is
two-exact.
\end{proof}

If the third homology $H_3(\dl_1)$ of the complex (\ref{v87'}) is
not trivial, its elements correspond to second-stage Noether
identities which the complete first-stage ones satisfy, and so on.
Iterating the arguments, one comes to the following.

A degenerate Grassmann-graded Lagrangian system
$(\cS^*_\infty[F;Y],L)$ is called $N$-stage reducible if it admits
finitely generated non-trivial $N$-stage Noether identities, but
no non-trivial $(N+1)$-stage ones. It is characterized as follows
\cite{jmp05a,book09}.

$\bullet$ There are graded vector bundles $E_0,\ldots, E_N$ over
$X$, and a DBGA $\cP^*_\infty[\ol{VF};Y]$ is enlarged to a DBGA
\mar{v91}\beq
\ol\cP^*_\infty\{N\}=\cP^*_\infty[\ol{VF}\op\times_X \ol
E_0\op\times_X\cdots\op\times_X \ol E_N;Y] \label{v91}
\eeq
with the local generating basis
\be
(s^A,\ol s_A, \ol c_r, \ol c_{r_1}, \ldots, \ol c_{r_N})
\ee
where $\ol c_{r_k}$ are $k$-stage Noether antifields of antifield
number Ant$[\ol c_{r_k}]=k+2$.

$\bullet$ The DBGA (\ref{v91}) is provided with a nilpotent right
graded derivation
\mar{v92,'}\ben
&&\dl_{\rm KT}=\dl_N=\ol\dl +
\op\sum_{0\leq|\La|}\rdr^r\Delta_r^{A,\La}\ol s_{\La A} +
\op\sum_{1\leq k\leq N}\rdr^{r_k} \Delta_{r_k},
\label{v92}\\
&& \Delta_{r_k}\om= \op\sum_{0\leq|\La|}
\Delta_{r_k}^{r_{k-1},\La}\ol c_{\La r_{k-1}}\om +
\label{v92'}\\
&& \qquad \op\sum_{0\leq |\Si|, |\Xi|}(h_{r_k}^{(r_{k-2},\Si)(A,\Xi)}\ol
c_{\Si r_{k-2}}\ol s_{\Xi A}+...)\om \in
\ol\cP^{0,n}_\infty\{k-1\}_{k+1}, \nonumber
\een
of  antifield number -1. The index $k=-1$ here stands for $\ol
s_A$. The nilpotent derivation $\dl_{\rm KT}$ (\ref{v92}) is
called the Koszul -- Tate operator.

$\bullet$ With this graded derivation, a module
$\ol\cP^{0,n}_\infty\{N\}_{\leq N+3}$ of densities of antifield
number $\leq (N+3)$ is decomposed into the exact  Koszul -- Tate
chain complex
\mar{v94}\ben
&& 0\lto \im \ol\dl \llr^{\ol\dl}
\cP^{0,n}_\infty[\ol{VF};Y]_1\llr^{\dl_0}
\ol\cP^{0,n}_\infty\{0\}_2\llr^{\dl_1}
\ol\cP^{0,n}_\infty\{1\}_3\cdots
\label{v94}\\
&& \qquad
 \llr^{\dl_{N-1}} \ol\cP^{0,n}_\infty\{N-1\}_{N+1}
\llr^{\dl_{\rm KT}} \ol\cP^{0,n}_\infty\{N\}_{N+2}\llr^{\dl_{\rm
KT}} \ol\cP^{0,n}_\infty\{N\}_{N+3} \nonumber
\een
which satisfies the following homology regularity condition.

\begin{condition} \label{v155} \mar{v155} Any $\dl_{k<N}$-cycle
\be
\f\in \ol\cP_\infty^{0,n}\{k\}_{k+3}\subset
\ol\cP_\infty^{0,n}\{k+1\}_{k+3}
\ee
is a $\dl_{k+1}$-boundary.
\end{condition}

\begin{remark}
The exactness of the complex (\ref{v94}) means that any
$\dl_{k<N}$-cycle $\f\in \cP_\infty^{0,n}\{k\}_{k+3}$, is a
$\dl_{k+2}$-boundary, but not necessary a $\dl_{k+1}$-one.
\end{remark}

$\bullet$ The nilpotentness $\dl_{\rm KT}^2=0$ of the Koszul --
Tate operator (\ref{v92}) is equivalent to the complete
non-trivial Noether identities (\ref{v64}) and the complete
non-trivial $(k\leq N)$-stage Noether identities
\mar{v93}\ben
&& \op\sum_{0\leq|\La|} \Delta_{r_k}^{r_{k-1},\La}d_\La
\left(\op\sum_{0\leq|\Si|} \Delta_{r_{k-1}}^{r_{k-2},\Si}\ol
c_{\Si
r_{k-2}}\right) = \label{v93}\\
&& \qquad -  \ol\dl\left(\op\sum_{0\leq |\Si|,
|\Xi|}h_{r_k}^{(r_{k-2},\Si)(A,\Xi)}\ol c_{\Si r_{k-2}}\ol s_{\Xi
A}\right). \nonumber
\een
This item means the following.

\begin{lemma} Any $\dl_k$-cocycle $\Phi\in
\cP^{0,n}_\infty\{k\}_{k+2}$ is a $k$-stage Noether identity, and
{\it vice versa}.
\end{lemma}

\begin{proof} Any $(k+2)$-chain $\Phi\in \cP^{0,n}_\infty\{k\}_{k+2}$ takes a form
\mar{v156'}\ben
&& \Phi= G+H=\op\sum_{0\leq|\La|} G^{r_k,\La}\ol c_{\La r_k}\om +
\label{v156'}\\
&& \qquad \op\sum_{0\leq \Si, 0\leq\Xi}(H^{(A,\Xi)(r_{k-1},\Si)}\ol s_{\Xi
A}\ol c_{\Si r_{k-1}}+...)\om.\nonumber
\een
If it is a $\dl_k$-cycle, then
\mar{v145'}\ben
&& \op\sum_{0\leq|\La|} G^{r_k,\La}d_\La
\left(\op\sum_{0\leq|\Si|} \Delta_{r_k}^{r_{k-1},\Si}\ol c_{\Si
r_{k-1}}\right) + \label{v145'}\\
&& \qquad \ol\dl\left(\op\sum_{0\leq \Si, 0\leq\Xi}H^{(A,\Xi)(r_{k-1},\Si)}\ol
s_{\Xi A}\ol c_{\Si r_{k-1}}\right)=0 \nonumber
\een
are the $k$-stage Noether identities. Conversely, let the
condition (\ref{v145'}) hold. Then it can be extended to a cycle
condition as follows. It is brought into the form
\be
&& \dl_k\left(\op\sum_{0\leq|\La|} G^{r_k,\La}\ol c_{\La r_k} +
\op\sum_{0\leq \Si, 0\leq\Xi}H^{(A,\Xi)(r_{k-1},\Si)}\ol
s_{\Xi A}\ol c_{\Si r_{k-1}}\right)=\\
&& \qquad  -\op\sum_{0\leq|\La|} G^{r_k,\La}d_\La h_{r_k}
+\op\sum_{0\leq \Si, 0\leq\Xi}H^{(A,\Xi)(r_{k-1},\Si)}\ol s_{\Xi
A}d_\Si \Delta_{r_{k-1}}.
\ee
A glance at the expression (\ref{v92'}) shows that the term in the
right-hand side of this equality belongs to
$\cP^{0,n}_\infty\{k-2\}_{k+1}$. It is a $\dl_{k-2}$-cycle and,
consequently, a $\dl_{k-1}$-boundary $\dl_{k-1}\Psi$ in accordance
with Condition \ref{v155}. Then the equality (\ref{v145'}) is a
$\ol c_{\Si r_{k-1}}$-dependent part of the cycle condition
\be
\dl_k\left(\op\sum_{0\leq|\La|} G^{r_k,\La}\ol c_{\La r_k} +
\op\sum_{0\leq \Si, 0\leq\Xi}H^{(A,\Xi)(r_{k-1},\Si)}\ol s_{\Xi
A}\ol c_{\Si r_{k-1}} -\Psi\right)=0,
\ee
but $\dl_k\Psi$ does not make a contribution to this condition.
\end{proof}

\begin{lemma}
Any trivial $k$-stage Noether identity is a $\dl_k$-boundary
$\Phi\in \cP^{0,n}_\infty\{k\}_{k+2}$.
\end{lemma}

\begin{proof}
The $k$-stage Noether identities (\ref{v145'}) are trivial either
if a $\dl_k$-cycle $\Phi$ (\ref{v156'}) is a $\dl_k$-boundary or
its summand $G$ vanishes on-shell. Let us show that, if the
summand $G$ of $\Phi$ (\ref{v156'}) is $\ol\dl$-exact, then $\Phi$
is a $\dl_k$-boundary. If $G=\ol\dl \Psi$, one can write
\mar{vv169}\beq
\Phi=\dl_k\Psi +(\ol \dl-\dl_k)\Psi + H. \label{vv169}
\eeq
Hence, the $\dl_k$-cycle condition reads
\be
\dl_k\Phi=\dl_{k-1}((\ol\dl-\dl_k)\Psi + H)=0.
\ee
By virtue of Condition \ref{v155}, any $\dl_{k-1}$-cycle $\f\in
\ol\cP^{0,n}_\infty\{k-1\}_{k+2}$ is $\dl_k$-exact. Then
\be
(\ol \dl-\dl_k)\Psi + H
\ee
is a $\dl_k$-boundary. Consequently, $\Phi$ (\ref{v156'}) is
$\dl_k$-exact.
\end{proof}

\begin{lemma} All non-trivial $k$-stage Noether identity (\ref{v145'}), by
assumption, factorize as
\be
\Phi= \op\sum_{0\leq|\Xi|} \Phi^{r_k,\Xi} d_\Xi \Delta_{r_k}\om,
\qquad \Phi^{r_1,\Xi}\in \cS^0_\infty[F;Y],
\ee
through the complete ones (\ref{v93}).
\end{lemma}

It may happen that a Grassmann-graded Lagrangian field system
possesses non-trivial Noether identities of any stage. However, we
restrict our consideration to $N$-reducible Lagrangian systems for
a finite integer $N$. In this case, the Koszul -- Tate operator
(\ref{v92}) and the gauge operator (\ref{w108'}) below contain
finite terms.

\subsection{Inverse second Noether theorem}

Different variants of the second Noether theorem have been
suggested in order to relate reducible Noether identities and
gauge symmetries \cite{barn,jmp05,jpa05,fulp}. The inverse second
Noether theorem (Theorem \ref{w35}), that we formulate in homology
terms, associates to the Koszul -- Tate complex (\ref{v94}) of
non-trivial Noether identities the cochain sequence (\ref{w108})
with the ascent operator $\bu$ (\ref{w108'}) whose components are
complete non-trivial gauge and higher-stage gauge symmetries of
Lagrangian system. Let us start with the following notation.

\begin{remark} \label{42n1} \mar{42n1}
Given the DBGA $\ol\cP^*_\infty\{N\}$ (\ref{v91}), we consider the
the DBGA
\mar{w5}\beq
P^*_\infty\{N\}=P^*_\infty[F\op\times_XE_0\op\times_X\cdots
\op\times_X E_N;Y], \label{w5}
\eeq
possessing the local generating basis
\be
(s^A, c^r, c^{r_1}, \ldots, c^{r_N}), \qquad [c^{r_k}]=([\ol
c_{r_k}]+1){\rm mod}\,2,
\ee
and the DBGA
\mar{w6}\beq
\cP^*_\infty\{N\}=\cP^*_\infty[\ol{VF}\op\times_X
E_0\op\times_X\cdots \op\times_X E_N \op\times_X \ol
E_0\op\times_X\cdots\op\times_X \ol E_N;Y] \label{w6}
\eeq
with the local generating basis
\be
(s^A, \ol s_A, c^r, c^{r_1}, \ldots, c^{r_N},\ol c_r, \ol c_{r_1},
\ldots, \ol c_{r_N}),
\ee
(see Remark \ref{rtr} for the notation). Their elements $c^{r_k}$
are called $k$-stage ghosts of ghost number gh$[c^{r_k}]=k+1$ and
antifield number
\be
{\rm Ant}[c^{r_k}]=-(k+1).
\ee
A $C^\infty(X)$-module $\cC^{(k)}$ of $k$-stage ghosts is the
density-dual of a module $\cC_{(k+1})$ of $(k+1)$-stage Noether
antifields. In accordance with Remark \ref{pkp}, the DBGAs
$\ol\cP^*_\infty\{N\}$ (\ref{v91}) and the BGDA $P^*(N)$
(\ref{w5}) are subalgebras of the DBGA $\cP^*_\infty\{N\}$
(\ref{w6}). The Koszul -- Tate operator $\dl_{\rm KT}$ (\ref{v92})
is naturally extended to a graded derivation of the DBGA
$\cP^*_\infty\{N\}$.
\end{remark}

\begin{remark} \label{42n10} \mar{42n10} Any
graded differential form $\f\in \cS^*_\infty[F;Y]$ and any finite
tuple $(f^\La)$, $0\leq |\La|\leq k$, of local graded functions
$f^\La\in \cS^0_\infty[F;Y]$ obey the following relations
\cite{jpa05}:
\mar{qq1}\ben
&& \op\sum_{0\leq |\La|\leq k} f^\La d_\La \f\w \om= \op\sum_{0\leq
|\La|}(-1)^{|\La|}d_\La (f^\La)\f\w \om +d_H\si,
\label{qq1a}\\
&& \op\sum_{0\leq |\La|\leq k} (-1)^{|\La|}d_\La(f^\La \f)=
\op\sum_{0\leq |\La|\leq k} \eta (f)^\La d_\La \f, \label{qq1b}\\
&& \eta (f)^\La = \op\sum_{0\leq|\Si|\leq k-|\La|}(-1)^{|\Si+\La|}
\frac{(|\Si+\La|)!}{|\Si|!|\La|!} d_\Si f^{\Si+\La},
\label{qq1c}\\
&& \eta(\eta(f))^\La=f^\La. \label{qq1d}
\een
\end{remark}

\begin{theorem} \label{w35} \mar{w35} Given the Koszul -- Tate complex (\ref{v94}),
a module of graded densities $P_\infty^{0,n}\{N\}$ is decomposed
into the cochain sequence
\mar{w108,'}\ben
&& 0\to \cS^{0,n}_\infty[F;Y]\ar^{\bu}
P^{0,n}_\infty\{N\}^1\ar^{\bu}
P^{0,n}_\infty\{N\}^2\ar^{\bu}\cdots, \label{w108}\\
&& \bu=u + u^{(1)}+\cdots + u^{(N)}=\label{w108'}\\
&& \qquad u^A\frac{\dr}{\dr s^A} +
u^r\frac{\dr}{\dr c^r} +\cdots  + u^{r_{N-1}}\frac{\dr}{\dr
c^{r_{N-1}}}, \nonumber
\een
graded in ghost number. Its ascent operator $\bu$ (\ref{w108'}) is
an odd graded derivation of ghost number 1 where $u$ (\ref{w33})
is a variational symmetry of a graded Lagrangian $L$ and the
graded derivations $u_{(k)}$ (\ref{w38}), $k=1,\ldots, N$, obey
the relations (\ref{w34}).
\end{theorem}

\begin{proof} Given the Koszul -- Tate operator (\ref{v92}), let us extend an original
Lagrangian $L$ to the Lagrangian
\mar{w8}\beq
L_e=L+L_1=L + \op\sum_{0\leq k\leq N} c^{r_k}\Delta_{r_k}\om=L
+\dl_{\rm KT}( \op\sum_{0\leq k\leq N} c^{r_k}\ol c_{r_k}\om)
\label{w8}
\eeq
of zero antifield number. It is readily observed that the Koszul
-- Tate operator $\dl_{\rm KT}$ is an exact symmetry of the
extended Lagrangian $L_e\in \cP^{0,n}_\infty\{N\}$ (\ref{w8}).
Since the graded derivation $\dl_{\rm KT}$ is vertical, it follows
from the first variational formula (\ref{g107}) that
\mar{w16}\ben
&& \left[\frac{\op\dl^\lto \cL_e}{\dl \ol s_A}\cE_A
+\op\sum_{0\leq k\leq N} \frac{\op\dl^\lto \cL_e}{\dl \ol
c_{r_k}}\Delta_{r_k}\right]\om = \label{w16}\\
&& \qquad \left[\up^A\cE_A + \op\sum_{0\leq k\leq N}\up^{r_k}\frac{\dl
\cL_e}{\dl c^{r_k}}\right]\om= d_H\si, \nonumber \\
&& \up^A= \frac{\op\dl^\lto \cL_e}{\dl \ol s_A}=u^A+w^A
=\op\sum_{0\leq|\La|} c^r_\La\eta(\Delta^A_r)^\La +
 \op\sum_{1\leq i\leq N}\op\sum_{0\leq|\La|}
c^{r_i}_\La\eta(\op\dr^\lto{}^A(h_{r_i}))^\La, \nonumber\\
&& \up^{r_k}=\frac{\op\dl^\lto \cL_e}{\dl \ol c_{r_k}} =u^{r_k}+
w^{r_k}= \op\sum_{0\leq|\La|}
c^{r_{k+1}}_\La\eta(\Delta^{r_k}_{r_{k+1}})^\La + \nonumber\\
&& \qquad \op\sum_{k+1<i\leq N} \op\sum_{0\leq|\La|}
c^{r_i}_\La\eta(\op\dr^\lto{}^{r_k}(h_{r_i}))^\La. \nonumber
\een
The equality (\ref{w16}) is split into the set of equalities
\mar{w19,20}\ben
&& \frac{\op\dl^\lto (c^r\Delta_r)}{\dl \ol s_A}\cE_A \om
=u^A\cE_A \om=d_H\si_0, \label{w19}\\
&&  \left[\frac{\op\dl^\lto (c^{r_k}\Delta_{r_k})}{\dl \ol s_A}\cE_A
+\op\sum_{0\leq i<k} \frac{\op\dl^\lto (c^{r_k}\Delta_{r_k})}{\dl
\ol c_{r_i}}\Delta_{r_i}\right] \om= d_H\si_k,  \label{w20}
\een
where $k=1,\ldots,N$. A glance at the equality (\ref{w19}) shows
that, by virtue of the first variational formula (\ref{g107}), the
odd graded derivation
\mar{w33}\beq
u= u^A\frac{\dr}{\dr s^A}, \qquad u^A =\op\sum_{0\leq|\La|}
c^r_\La\eta(\Delta^A_r)^\La, \label{w33}
\eeq
of $P^0_\infty\{0\}$ is a variational symmetry of a graded
Lagrangian $L$. Every equality (\ref{w20}) falls into a set of
equalities graded by the polynomial degree in antifields. Let us
consider that of them linear in antifields $\ol c_{r_{k-2}}$. We
have
\be
&& \frac{\op\dl^\lto}{\dl \ol
s_A}\left(c^{r_k}\op\sum_{0\leq|\Si|,|\Xi|}h_{r_k}^{(r_{k-2},\Si)(A,\Xi)}
\ol
c_{\Si r_{k-2}}\ol s_{\Xi A}\right)\cE_A\om + \\
&& \qquad \frac{\op\dl^\lto}{\dl \ol
c_{r_{k-1}}}\left(c^{r_k}\op\sum_{0\leq|\Si|}\Delta_{r_k}^{r'_{k-1},\Si}\ol
c_{\Si r'_{k-1}}\right)\op\sum_{0\leq|\Xi|}
\Delta_{r_{k-1}}^{r_{k-2},\Xi}\ol c_{\Xi r_{k-2}}\om= d_H\si_k.
\ee
This equality is brought into the form
\be
 && \op\sum_{0\leq|\Xi|}
(-1)^{|\Xi|}d_\Xi\left(c^{r_k}\op\sum_{0\leq|\Si|}
h_{r_k}^{(r_{k-2},\Si)(A,\Xi)} \ol c_{\Si r_{k-2}}\right)\cE_A \om
+ \\
&& \qquad u^{r_{k-1}}\op\sum_{0\leq|\Xi|} \Delta_{r_{k-1}}^{r_{k-2},\Xi}\ol
c_{\Xi r_{k-2}} \om= d_H\si_k.
\ee
Using the relation (\ref{qq1a}), we obtain the equality
\mar{ddd1}\ben
&& \op\sum_{0\leq|\Xi|} c^{r_k}\op\sum_{0\leq|\Si|}
h_{r_k}^{(r_{k-2},\Si)(A,\Xi)} \ol c_{\Si r_{k-2}} d_\Xi\cE_A\om
+ \label{ddd1}\\
&& \qquad u^{r_{k-1}}\op\sum_{0\leq|\Xi|}
\Delta_{r_{k-1}}^{r_{k-2},\Xi}\ol c_{\Xi r_{k-2}}\om=
d_H\si'_k.\nonumber
\een
The variational derivative of both its sides with respect to $\ol
c_{r_{k-2}}$ leads to the relation
\mar{w34}\ben
&&\op\sum_{0\leq|\Si|} d_\Si u^{r_{k-1}}\frac{\dr}{\dr
c^{r_{k-1}}_\Si} u^{r_{k-2}} =\ol\dl(\al^{r_{k-2}}),\label{w34}\\
&& \al^{r_{k-2}} = -\op\sum_{0\leq|\Si|}
\eta(h_{r_k}^{(r_{k-2})(A,\Xi)})^\Si d_\Si(c^{r_k} \ol s_{\Xi A}),
\nonumber
\een
which the odd graded derivation
\mar{w38}\beq
u^{(k)}= u^{r_{k-1}}\frac{\dr}{\dr
c^{r_{k-1}}}=\op\sum_{0\leq|\La|}
c^{r_k}_\La\eta(\Delta^{r_{k-1}}_{r_k})^\La\frac{\dr}{\dr
c^{r_{k-1}}}, \quad k=1,\ldots,N, \label{w38}
\eeq
satisfies. Graded derivations $u$ (\ref{w33}) and $u^{(k)}$
(\ref{w38}) are assembled into the ascent operator $\bu$
(\ref{w108'}) of the cochain sequence (\ref{w108}).
\end{proof}

A glance at the expression (\ref{w33}) shows that the variational
symmetry $u$ is a graded derivation of a ring
\be
P^0_\infty[0]=S^0_\infty[F\op\times_X E_0; Y\op\times_X E^0_0]
\ee
which satisfies Definition \ref{sgauge}. Consequently,  $u$
(\ref{w33}) is a gauge symmetry of a graded Lagrangian $L$ which
is associated to the complete non-trivial Noether identities
(\ref{v64}). Therefore, it is a non-trivial gauge symmetry.
Moreover, it is complete in the following sense. Let
\be
\op\sum_{0\leq|\Xi|} C^RG^{r,\Xi}_R d_\Xi \Delta_r \om
\ee
be some projective $C^\infty(X)$-module of finite rank of
non-trivial Noether identities (\ref{xx2}) parameterized by the
corresponding ghosts $C^R$. We have the equalities
\be
&& 0=\op\sum_{0\leq|\Xi|} C^RG^{r,\Xi}_R d_\Xi
\left(\op\sum_{0\leq|\La|}\Delta_r^{A,\La}d_\La \cE_A\right) \om=\\
&& \qquad \op\sum_{0\leq|\La|}\left(\op\sum_{0\leq|\Xi|}\eta(G^r_R)^\Xi C^R_\Xi\right)
\Delta_r^{A,\La}d_\La \cE_A\om+d_H(\si)=\\
&& \qquad
\op\sum_{0\leq|\La|}(-1)^{|\La|}d_\La\left(\Delta_r^{A,\La}\op\sum_{0\leq|\Xi|}\eta(G^r_R)^\Xi
C^R_\Xi\right)\cE_A \om +d_H\si =\\
&& \qquad
\op\sum_{0\leq|\La|}\eta(\Delta_r^A)^\La
d_\La\left(\op\sum_{0\leq|\Xi|}\eta(G^r_R)^\Xi C^R_\Xi\right)\cE_A
\om
+d_H\si=\\
&&\qquad \op\sum_{0\leq|\La|}u_r^{A,\La}d_\La\left(\op\sum_{0\leq|\Xi|}\eta(G^r_R)^\Xi
C^R_\Xi\right)\cE_A \om +d_H\si.
\ee
It follows that the graded derivation
\be
d_\La\left(\op\sum_{0\leq|\Xi|}\eta(G^r_R)^\Xi
C^R_\Xi\right)u_r^{A,\La}\frac{\dr}{\dr s^A}
\ee
is a variational symmetry of a graded Lagrangian $L$ and,
consequently, its gauge symmetry parameterized by ghosts $C^R$. It
factorizes through the gauge symmetry (\ref{w33}) by putting
ghosts
\be
c^r= \op\sum_{0\leq|\Xi|}\eta(G^r_R)^\Xi C^R_\Xi.
\ee
Thus, we come to the following definition.

Turn now to the relation (\ref{w34}). For $k=1$, it takes a form
\be
\op\sum_{0\leq|\Si|} d_\Si u^r\frac{\dr}{\dr c^r_\Si} u^A =\ol
\dl(\al^A)
\ee
of a first-stage gauge symmetry condition on-shell which the
non-trivial gauge symmetry $u$ (\ref{w33}) satisfies. Therefore,
one can treat the odd graded derivation
\be
u^{(1)}= u^r\frac{\dr}{\dr c^r}, \qquad u^r=\op\sum_{0\leq|\La|}
c^{r_1}_\La\eta(\Delta^r_{r_1})^\La,
\ee
as a first-stage gauge symmetry associated to the complete
first-stage Noether identities
\be
 \op\sum_{0\leq|\La|} \Delta_{r_1}^{r,\La}d_\La
\left(\op\sum_{0\leq|\Si|} \Delta_r^{A,\Si}\ol s_{\Si A}\right) =
- \ol\dl\left(\op\sum_{0\leq |\Si|,
|\Xi|}h_{r_1}^{(B,\Si)(A,\Xi)}\ol s_{\Si B}\ol s_{\Xi A}\right).
\ee

Iterating the arguments, one comes to the relation (\ref{w34})
which provides a $k$-stage gauge symmetry condition which is
associated to the complete non-trivial $k$-stage Noether
identities (\ref{v93}).

The odd graded derivation $u_{(k)}$ (\ref{w38}) is called the
$k$-stage gauge symmetry. It is non-trivial and complete as
follows. Let
\be
\op\sum_{0\leq|\Xi|} C^{R_k}G^{r_k,\Xi}_{R_k} d_\Xi \Delta_{r_k}
\om
\ee
be a projective $C^\infty(X)$-module of finite rank of non-trivial
$k$-stage Noether identities (\ref{xx2}) factorizing through the
complete ones (\ref{v92'}) and parameterized by the corresponding
ghosts $C^{R_k}$. One can show that it defines a $k$-stage gauge
symmetry factorizing through $u^{(k)}$ (\ref{w38}) by putting
$k$-stage ghosts
\be
c^{r_k}= \op\sum_{0\leq|\Xi|}\eta(G^{r_k}_{R_k})^\Xi C^{R_k}_\Xi.
\ee

Thus, components of the ascent operator $\bu$ (\ref{w108'}) in
inverse second Noether Theorem \ref{w35} are complete non-trivial
gauge and higher-stage gauge symmetries. Therefore, we agree to
call this operator the gauge operator.

\begin{remark}
With the gauge operator (\ref{w108'}), the extended Lagrangian
$L_e$ (\ref{w8}) takes a form
\mar{lmp2}\beq
L_e= L+\bu( \op\sum_{0\leq k\leq N} c^{r_{k-1}}\ol c_{r_{k-1}})
\om + L^*_1 + d_H\si, \label{lmp2}
\eeq
where $L^*_1$ is a term of polynomial degree in antifields
exceeding 1.
\end{remark}

\subsection{Direct second Noether theorem}

The correspondence between of complete non-trivial gauge and
higher-stage gauge symmetries to  complete non-trivial Noether and
higher-stage Noether identities is unique due to the following
direct second Noether theorem.

\begin{theorem} \mar{825} \label{825}
(i) If $u$ (\ref{w33}) is a gauge symmetry, the variational
derivative of the $d_H$-exact density $u^A\cE_A\om$ (\ref{w19})
with respect to ghosts $c^r$ leads to the equality
\mar{gg11}\ben
&& \dl_r(u^A\cE_A
\om)= \op\sum_{0\leq |\La|}(-1)^{|\La|}d_\La[u^{A\La}_r\cE_A]= \label{gg11}\\
&&\op\sum_{0\leq|\La|}(-1)^{|\La|}d_\La(\eta(\Delta^A_r)^\La\cE_A)=
\op\sum_{0\leq|\La|}(-1)^{|\La|} \eta(\eta(\Delta^A_r))^\La
d_\La\cE_A=0, \nonumber
\een
which reproduces the complete Noether identities (\ref{v64}) by
means of the relation (\ref{qq1d}).

(ii) Given the $k$-stage gauge symmetry condition (\ref{w34}), the
variational derivative of the equality (\ref{ddd1}) with respect
to ghosts $c^{r_k}$ leads to the equality, reproducing the
$k$-stage Noether identities (\ref{v93})
 by means of the relations (\ref{qq1b}) -- (\ref{qq1d}).
\end{theorem}

\begin{example} \mar{vbv} \label{vbv} If the gauge symmetry $u$ (\ref{gg2}) is of second
jet order in gauge parameters, i.e.,
\mar{0656}\beq
u_V=(u_r^A c^r +u^{A\m}_r c^r_\m + u_r^{A\nu\m}c^r_{\nu\m})\dr_A,
\label{0656}
\eeq
the corresponding Noether identities (\ref{gg11}) take a form
\mar{0657}\beq
u^A_r\cE_A - d_\m(u^{A\m}_r\cE_A) + d_{\nu\m}(u_r^{A\nu\m}
\cE_A)=0, \label{0657}
\eeq
and {\it vice versa}.
\end{example}

\begin{remark} \mar{asd} \label{asd}
A glance at the expression (\ref{0657}) shows that, if the gauge
symmetry $u_V$ (\ref{0656}) is independent of jets of gauge
parameters, then all variational derivatives of a Lagrangian
equals zero, i.e., this Lagrangian is variationally trivial.
Therefore, such gauge symmetries usually are not considered. At
the same time, let a Lagrangian $L$ be variationally trivial. Its
variational derivatives $\cE_A\equiv 0$ obey irreducible complete
Noether identities
\mar{a2a}\beq
\ol\dl\Delta_A=0, \qquad \Delta_A=\ol s_A. \label{a2a}
\eeq
By the formula (\ref{w108'}), the associated irreducible gauge
symmetry is given by the gauge operator
\mar{a3a}\beq
\bu=c^A\frac{\dr}{\dr s^A}. \label{a3a}
\eeq
\end{remark}

\begin{remark} \mar{s8} \label{s8}
One can consider gauge symmetries which need not be be linear in
gauge parameters. Let us call then the generalized gauge
symmetries. However, direct second Noether Theorem \ref{gg11} is
not relevant to generalized gauge symmetries because, in this
case, an Euler -- Lagrange operator satisfies the identities
depending on gauge parameters.
\end{remark}

\subsection{BRST operator}

In contrast with the Koszul -- Tate operator (\ref{v92}), the
gauge operator $\bu$ (\ref{w108}) need not be nilpotent. Following
the basic example of Yang -- Mills gauge theory (Section 6.1), let
us study its extension to a nilpotent graded derivation
\mar{w109}\ben
&& \bbc=\bu+ \g=\bu + \op\sum_{1\leq k\leq N+1}\g^{(k)}=
\bu + \op\sum_{1\leq k\leq N+1}\g^{r_{k-1}}\frac{\dr}{\dr
c^{r_{k-1}}} \label{w109} \\
&& \qquad =\left(u^A\frac{\dr}{\dr s^A}+ \g^r\frac{\dr}{\dr c^r}\right) +
\op\sum_{0\leq k\leq N-1} \left(u^{r_k}\frac{\dr}{\dr c^{r_k}}+
\g^{r_{k+1}}\frac{\dr}{\dr c^{r_{k+1}}}\right) \nonumber
\een
of ghost number 1 by means of antifield-free terms $\g^{(k)}$ of
higher polynomial degree in ghosts $c^{r_i}$ and their jets
$c^{r_i}_\La$, $0\leq i<k$. We call $\bbc$ (\ref{w109}) the BRST
operator, where $k$-stage gauge symmetries are extended to
$k$-stage BRST transformations acting both on $(k-1)$-stage and
$k$-stage ghosts \cite{jmp09,book09}. If the BRST operator exists,
the cochain sequence (\ref{w108}) is brought into the BRST complex
\be
0\to \cS^{0,n}_\infty[F;Y]\ar^{\bbc}
P^{0,n}_\infty\{N\}^1\ar^{\bbc}
P^{0,n}_\infty\{N\}^2\ar^{\bbc}\cdots.
\ee

There is  the following necessary condition of the existence of
such a BRST extension.

\begin{theorem} \label{lmp6} \mar{lmp6} The gauge operator
(\ref{w108}) admits the nilpotent extension (\ref{w109}) only if
the gauge symmetry conditions (\ref{w34}) and the higher-stage
Noether identities (\ref{v93}) are satisfied off-shell.
\end{theorem}

\begin{proof}
It is easily justified that, if the graded derivation $\bbc$
(\ref{w109}) is nilpotent, then  the right hand sides of the
equalities (\ref{w34}) equal zero, i.e.,
\mar{850}\beq
u^{(k+1)}(u^{(k)})=0, \qquad 0\leq k\leq N-1, \qquad u^{(0)}=u.
\label{850}
\eeq
Using the relations (\ref{qq1a}) -- (\ref{qq1d}), one can show
that, in this case, the right hand sides of the higher-stage
Noether identities (\ref{v93}) also equal zero \cite{jmp05}. It
follows that the summand $G_{r_k}$ of each cocycle $\Delta_{r_k}$
(\ref{v92'}) is $\dl_{k-1}$-closed. Then its summand $h_{r_k}$
also is $\dl_{k-1}$-closed and, consequently, $\dl_{k-2}$-closed.
Hence it is $\dl_{k-1}$-exact by virtue of Condition \ref{v155}.
Therefore, $\Delta_{r_k}$ contains only the term $G_{r_k}$ linear
in antifields.
\end{proof}

It follows at once from the equalities (\ref{850}) that the
higher-stage gauge operator
\be
u_{\rm HS}=\bu-u=u^{(1)}+\cdots + u^{(N)}
\ee
is nilpotent, and
\mar{a1a}\beq
\bu(\bu)=u(\bu). \label{a1a}
\eeq
Therefore, the nilpotency condition for the BRST operator $\bbc$
(\ref{w109}) takes a form
\mar{851}\beq
\bbc(\bbc)=(u+\g)(\bu) +(u+u_{\rm HS}+\g)(\g)=0. \label{851}
\eeq
Let us denote
\be
&& \g^{(0)}=0, \\
&& \g^{(k)}=\g^{(k)}_{(2)} +\cdots + \g^{(k)}_{(k+1)},
\quad k=1,\ldots, N+1,\\
&& \g^{r_{k-1}}_{(i)}= \op\sum_{k_1+\cdots+ k_i=k+1
-i}\left(\op\sum_{0\leq |\La_{k_j}|}
\g^{r_{k-1},\La_{k_1},\ldots,\La_{k_i}}_{(i)r_{k_1},\ldots,r_{k_i}}c^{r_{k_1}}_{\La_{k_1}}
\cdots c^{r_{k_i}}_{\La_{k_i}}\right), \\
&& \g^{(N+2)}=0,
\ee
where $\g^{(k)}_{(i)}$ are terms of polynomial degree $2\leq i\leq
k+1$ in ghosts. Then the nilpotent property (\ref{851}) of $\bbc$
falls into a set of equalities
\mar{w110,3}\ben
&& u^{(k+1)}(u^{(k)})
=0, \qquad 0\leq k\leq N-1,  \label{w110}\\
&& (u +\g^{(k+1)}_{(2)})(u^{(k)}) + u_{\rm HS}(\g^{(k)}_{(2)})=0,
\qquad 0\leq k\leq N+1, \label{w111} \\
&& \g_{(i)}^{(k+1)}(u^{(k)}) + u (\g_{(i-1)}^{(k)}) +
u_{\rm HS}(\g^{k}_{(i)}) + \label{w113}\\
&&\qquad \op\sum_{2\leq m\leq
i-1}\g_{(m)}(\g_{(i-m+1)}^{(k)})=0, \qquad  i-2\leq k\leq N+1,
\nonumber
\een
of ghost polynomial degree 1, 2 and $3\leq i\leq N+3$,
respectively.

The equalities (\ref{w110}) are exactly the gauge symmetry
conditions (\ref{850}) in Theorem \ref{lmp6}.

The equality (\ref{w111}) for $k=0$ reads
\mar{852}\beq
(u+ \g^{(1)})(u)=0, \qquad \op\sum_{0\leq |\La|} (d_\La(u^A)
\dr_A^\La u^B+ d_\La(\g^r)u^{B,\La}_r)=0. \label{852}
\eeq
It takes a form of the Lie antibracket
\mar{s12}\beq
[u,u]=-2\g^{(1)}(u)=-2\op\sum_{0\leq |\La|}
d_\La(\g^r)u^{B,\La}_r\dr_B \label{s12}
\eeq
of the odd gauge symmetry $u$. Its right-hand side is non-linear
in ghosts. Following Remark \ref{s8}, we treat it as a generalized
gauge symmetry factorizing through the gauge symmetry $u$. Thus,
we come to the following.

\begin{theorem} \label{830} \mar{830}
The gauge operator (\ref{w108}) admits the nilpotent extension
(\ref{w109}) only if the Lie antibracket of the odd gauge symmetry
$u$ (\ref{w33}) is a generalized gauge symmetry factorizing
through $u$.
\end{theorem}

The equalities (\ref{w111}) -- (\ref{w113}) for $k=1$ take a form
\mar{853,4}\ben
&& (u +\g^{(2)}_{(2)})(u^{(1)}) + u^{(1)}(\g^{(1)})=0, \label{853} \\
&& \g_{(3)}^{(2)}(u^{(1)}) + (u + \g^{(1)})(\g^{(1)})=0. \label{854}
\een
In particular, if a Lagrangian system is irreducible, i.e.,
$u^{(k)}=0$ and $\bu=u$, the BRST operator reads
\mar{hhh}\beq
 \bbc= u+\g^{(1)}=u^A\dr_A + \g^r\dr_r=
\op\sum_{0\leq|\La|} u^{A,\La}_r c^r_\La \dr_A +
\op\sum_{0\leq|\La|,|\Xi|}\g^{r,\La,\Xi}_{pq}c^p_\La c^q_\Xi
\dr_r. \label{hhh}
\eeq
In this case, the nilpotency conditions (\ref{853}) - (\ref{854})
are reduced to the equality
\mar{0691}\beq
(u + \g^{(1)})(\g^{(1)})=0. \label{0691}
\eeq
Furthermore, let a gauge symmetry $u$ be affine in fields $s^A$
and their jets. It follows from the nilpotency condition
(\ref{852}) that the BRST term $\g^{(1)}$ is independent of
original fields and their jets.  Then the relation (\ref{0691})
takes a form of the Jacobi identity
\mar{s11}\beq
\g^{(1)})(\g^{(1)})=0 \label{s11}
\eeq
for coefficient functions $\g^{r,\La,\Xi}_{pq}(x)$ in the Lie
antibracket (\ref{s12}).

The relations (\ref{s12}) and (\ref{s11}) motivate us to think of
the equalities (\ref{w111}) -- (\ref{w113}) in a general case of
reducible gauge symmetries as being {\it sui generis} generalized
commutation relations and Jacobi identities of gauge symmetries,
respectively \cite{jmp09,book09}. Based on Theorem \ref{830}, we
therefore say that non-trivial gauge symmetries are algebraically
closed (in the terminology of \cite{gom}) if the gauge operator
$\bu$ (\ref{w108'}) admits the nilpotent BRST extension $\bbc$
(\ref{w109}).

\begin{example} \label{44e1} \mar{44e1}
A Lagrangian system is called  abelian if its gauge symmetry $u$
is  abelian and the higher-stage gauge symmetries are independent
of original fields, i.e., if $u(\bu)=0$. It follows from the
relation (\ref{a1a}) that, in this case, the gauge operator itself
is the BRST operator $\bu=\bbc$. In particular, a Lagrangian
system with a variationally trivial Lagrangian (Remark \ref{asd})
is abelian, and $\bu$ (\ref{a3a}) is the BRST operator. The
topological BF theory exemplifies a reducible abelian Lagrangian
system (Section 6.4).
\end{example}

\subsection{Lagrangian BRST theory}

The DBGA $P^*_\infty\{N\}$ (\ref{w6}) is a particular
field-antifield theory of the following type
\cite{barn,lmp08,gom}.

Let us consider a pull-back composite bundle
\be
W=Z\op\times_X Z'\to Z\to X
\ee
where $Z'\to X$ is a vector bundle. Let us regard it as a graded
vector bundle over $Z$ possessing only odd part. The density-dual
$\ol{VW}$ of the vertical tangent bundle $VW$ of $W\to X$ is a
graded vector bundle
\be
\ol{VW}=((\ol Z'\op\oplus_Z V^*Z)\op\ot_Z\op\w^n T^*X)\op\oplus_Y
Z'
\ee
over $Z$ (cf. (\ref{41f4})). Let us consider the DBGA
$\cP^*_\infty[\ol{VW};Z]$ (\ref{41f5}) with the local generating
basis
\be
(z^a,\ol z_a), \qquad [\ol z_a]=([z^a]+1){\rm mod}\,2.
\ee
Its elements $z^a$ and $\ol z_a$ are called fields and antifields,
respectively.

Graded densities of this DBGA are endowed with the antibracket
\mar{f11}\beq
\{\gL \om,\gL'\om\}=\left[\frac{\op\dl^\lto \gL}{\dl \ol
z_a}\frac{\dl \gL'}{\dl z^a} +
(-1)^{[\gL']([\gL']+1)}\frac{\op\dl^\lto \gL'}{\dl \ol
z_a}\frac{\dl \gL}{\dl z^a}\right]\om. \label{f11}
\eeq
With this antibracket, one associates to any (even) Lagrangian
$\gL\om$ the odd vertical graded derivations
\mar{w37,lmp1}\ben
&&\up_\gL=\op\cE^\lto{}^a\dr_a=\frac{\op\dl^\lto \gL}{\dl \ol z_a}
\frac{\dr}{\dr z^a}, \label{w37}\\
&&\ol\up_\gL=\rdr^a\cE_a=\frac{\op\dr^\lto}{\dr \ol z_a}\frac{\dl
\gL}{\dl z^a}, \label{w37'}\\
&& \vt_\gL=\up_\gL+ \ol\up_\gL^l=(-1)^{[a]+1}\left(\frac{\dl \gL}{\dl
\ol z^a}\frac{\dr}{\dr z_a}+\frac{\dl \gL}{\dl z^a}\frac{\dr}{\dr
\ol z_a}\right),  \label{lmp1}
\een
such that
\be
\vt_\gL(\gL'\om)=\{\gL\om,\gL' \om\}.
\ee

\begin{theorem} \label{w39} \mar{w39} The following conditions are
equivalent.

(i) The antibracket of a Lagrangian $\gL\om$ is $d_H$-exact, i.e.,
\mar{w44}\beq
\{\gL\om,\gL\om\}=2\frac{\op\dl^\lto \gL}{\dl \ol z_a}\frac{\dl
\gL}{\dl z^a}\om =d_H\si. \label{w44}
\eeq

(ii) The graded derivation $\up$ (\ref{w37}) is a variational
symmetry of a Lagrangian $\gL\om$.

(iii) The graded derivation $\ol\up$ (\ref{w37'}) is a variational
symmetry of $\gL\om$.

(iv) The graded derivation $\vt_\gL$ (\ref{lmp1}) is nilpotent.
\end{theorem}

\begin{proof} By virtue of the first variational formula
(\ref{g107}), conditions (ii) and (iii) are equivalent to
condition (i). The equality (\ref{w44}) is equivalent to that the
odd density $\op\cE^\lto{}^a\cE_a\om$ is variationally trivial.
Replacing right variational derivatives $\op\cE^\lto{}^a$ with
$(-1)^{[a]+1}\cE^a$, we obtain
\be
2\op\sum_a (-1)^{[a]}\cE^a\cE_a \om=d_H\si.
\ee
The variational operator acting on this relation results in the
equalities
\be
&&
\op\sum_{0\leq|\La|}(-1)^{[a]+|\La|}d_\La(\dr^\La_b(\cE^a\cE_a))=\\
&& \qquad
\op\sum_{0\leq|\La|}(-1)^{[a]}[\eta(\dr_b\cE^a)^\La\cE_{\La a} +
\eta(\dr_b\cE_a)^\La\cE^a_\La)]=0, \\
&& \op\sum_{0\leq|\La|}(-1)^{[a]+|\La|}d_\La(\dr^{\La
b}(\cE^a\cE_a)) = \\
&& \qquad \op\sum_{0\leq|\La|}(-1)^{[a]}[\eta(\dr^b\cE^a)^\La\cE_{\La a} +
\eta(\dr^b\cE_a)\cE^a_\La] = 0.
\ee
Due to the identity
\be
(\dl\circ\dl)(L)=0, \qquad
\eta(\dr_B\cE_A)^\La=(-1)^{[A][B]}\dr_A^\La\cE_B,
\ee
we obtain
\be
&& \op\sum_{0\leq|\La|}(-1)^{[a]}[(-1)^{[b]([a]+1)}\dr^{\La
a}\cE_b\cE_{\La
a} + (-1)^{[b][a]}\dr_a^\La\cE_b\cE^a_\La]=0, \\
&& \op\sum_{0\leq|\La|}(-1)^{[a]+1}[(-1)^{([b]+1)([a]+1)}\dr^{\La
a}\cE^b\cE_{\La a} + (-1)^{([b]+1)[a]}\dr_a^\La\cE^b\cE^a_\La]=0
\ee
for all $\cE_b$ and $\cE^b$. This is exactly condition (iv).
\end{proof}

The equality (\ref{w44}) is called the classical master equation.
For instance, any variationally trivial Lagrangian satisfies the
master equation. A solution of the master equation (\ref{w44}) is
called non-trivial if both the derivations (\ref{w37}) and
(\ref{w37'}) do not vanish.

Being an element of the DBGA $\cP^*_\infty\{N\}$ (\ref{w6}), an
original Lagrangian $L$ obeys the master equation (\ref{w44}) and
yields the graded derivations $\up_L=0$ (\ref{w37}) and
$\ol\up_L=\ol\dl$ (\ref{w37'}), i.e., it is a trivial solution of
the master equation.

The graded derivations (\ref{w37}) -- (\ref{w37'}) associated to
the extended Lagrangian $L_e$ (\ref{lmp2}) are extensions
\be
&& \up_e= \bu+ \frac{\op\dl^\lto \cL^*_1}{\dl \ol s_A}\frac{\dr}{\dr
s^A} + \op\sum_{0\leq k\leq N} \frac{\op\dl^\lto \cL^*_1}{\dl \ol
c_{r_k}}\frac{\dr}{\dr c^{r_k}}, \\
&& \ol\up_e= \dl_{\rm KT} +
\frac{\rdr }{\dr \ol s_A}\frac{\dl \cL_1}{\dl s^A}
\ee
of the gauge and Koszul -- Tate operators, respectively. However,
the Lagrangian $L_e$ need not satisfy the master equation.
Therefore, let us consider its extension
\mar{w61}\beq
L_E=L_e+L'=L+L_1+L_2+\cdots \label{w61}
\eeq
by means of even densities $L_i$, $i\geq 2$, of zero antifield
number and polynomial degree $i$ in ghosts. The corresponding
graded derivations (\ref{w37}) -- (\ref{w37'}) read
\mar{w102,3}\ben
&& \up_E= \up_e+ \frac{\op\dl^\lto \cL'}{\dl \ol
s_A}\frac{\dr}{\dr s^A} + \op\sum_{0\leq k\leq N}
\frac{\op\dl^\lto \cL'}{\dl \ol
c_{r_k}}\frac{\dr}{\dr c^{r_k}}, \label{w102} \\
&& \ol\up_E= \ol\up_e + \frac{\rdr }{\dr \ol
s_A}\frac{\dl\cL'}{\dl s^A} + \op\sum_{0\leq k\leq N} \frac{\rdr
}{\dr \ol c_{r_k}}\frac{\dl \cL'}{\dl c^{r_k}}. \label{w103}
\een
The Lagrangian $L_E$ (\ref{w61}) where $L+L_1=L_e$ is called a
proper extension of an original Lagrangian $L$. The following is a
corollary of Theorem \ref{w39}.

\begin{corollary} \label{w120} \mar{w120} A Lagrangian $L$ is extended
to a  proper  solution $L_E$ (\ref{w61}) of the master equation
only if the gauge operator $\bu$ (\ref{w108}) admits a nilpotent
extension.
\end{corollary}

By virtue of condition (iv) of Theorem \ref{w39}, this nilpotent
extension is the derivation $\vt_E=\up_E +\ol\up^l_E$
(\ref{lmp1}), called the  KT-BRST operator. With this operator,
the module of densities $P^{0,n}_\infty\{N\}$ is split into the
KT-BRST complex
\mar{lmp12}\ben
&& \cdots\ar \cP^{0,n}_\infty\{N\}_2\ar \cP^{0,n}_\infty\{N\}_1\ar
\cP^{0,n}_\infty\{N\}_0\ar \label{lmp12}\\
&& \qquad  \cP^{0,n}_\infty\{N\}^1\ar \cP^{0,n}_\infty\{N\}^2\ar
\cdots.\nonumber
\een
Putting all ghosts zero, we obtain a cochain morphism of this
complex onto the Koszul -- Tate complex, extended to
$\ol\cP^{0,n}_\infty\{N\}$ and reversed into the cochain one.
Letting all Noether antifields zero, we come to a cochain morphism
of the KT-BRST complex (\ref{lmp12}) onto the cochain sequence
(\ref{w108}), where the gauge operator is extended to the
antifield-free part of the KT-BRST operator.

\begin{theorem} \label{w130} \mar{w130} If the gauge
operator $\bu$ (\ref{w108}) can be extended to the BRST operator
$\bbc$ (\ref{w109}), then the master equation has a non-trivial
proper solution
\mar{w133}\ben
&& L_E=L_e + \op\sum_{1\leq k\leq N}\g^{r_{k-1}}\ol
c_{r_{k-1}}\om=\label{w133}\\
&& \qquad  L+\bbc\left( \op\sum_{0\leq k\leq N} c^{r_{k-1}}\ol c_{r_{k-1}}\right) \om
+d_H\si. \nonumber
\een
\end{theorem}

\begin{proof} By virtue of Theorem \ref{lmp6}, if the BRST operator
$\bbc$ (\ref{w109}) exists, the densities $\Delta_{r_k}$
(\ref{v92'}) contain only the terms $G_{r_k}$ linear in
antifields. It follows that the extended Lagrangian $L_e$
(\ref{w8}) and, consequently, the Lagrangian $L_E$ (\ref{w133})
are affine in antifields. In this case, we have
\be
u^A=\op\dl^\lto{}^A(\cL_e), \qquad
u^{r_k}=\op\dl^\lto{}^{r_k}(\cL_e)
\ee
for all indices $A$ and $r_k$ and, consequently,
\be
\bbc^A=\op\dl^\lto{}^A(\cL_E), \qquad
\bbc^{r_k}=\op\dl^\lto{}^{r_k}(\cL_E),
\ee
i.e., $\bbc=\up_E$ is the graded derivation (\ref{w102}) defined
by the Lagrangian $L_E$. Its nilpotency condition takes a form
\be
\bbc(\op\dl^\lto{}^A(\cL_E))=0,\qquad
\bbc(\op\dl^\lto{}^{r_k}(\cL_E))=0.
\ee
Hence, we obtain
\be
\bbc(\cL_E)=\bbc(\op\dl^\lto{}^A(\cL_E)\ol s_A +
\op\dl^\lto{}^{r_k}(\cL_E) \ol c_{r_k})=0,
\ee
i.e., $\bbc$ is a variational symmetry of $L_E$. Consequently,
$L_E$ obeys the master equation.
\end{proof}

For instance, let a gauge symmetry $u$ be abelian, and let the
higher-stage gauge symmetries be independent of original fields,
i.e., $u(\bu)=0$. Then  $\bu=\bbc$ and $L_E=L_e$.

The proper solution $L_E$ (\ref{w133}) of the master equation is
called the  BRST extension of an original Lagrangian $L$.

\subsection{Appendix. Noether identities of differential operators}

Noether identities of a Lagrangian system in Section 4.1 are
particular Noether identities of differential operators which are
described in homology terms as follows \cite{oper}.

Let $E\to X$ be a vector bundle, and let $\cE$ be a $E$-valued
$k$-order differential operator on a fibre bundle $Y\to X$. It is
represented by a section $\cE^a$ of the pull-back bundle
\be
J^kY\op\times E\to J^kY
\ee
endowed with bundle coordinates $(x^\la, y^j_\Si,\chi^a)$,
$0\leq|\Si|\leq k$ \cite{bry,book,kras}.

\begin{definition} \label{46d1} \mar{46d1}
One says that a differential operator $\cE$ obeys Noether
identities if there exist an $r$-order differential operator
$\Phi$ on the pull-back bundle
\mar{46f1}\beq
E_Y=Y\op\times_X \label{46f1}E\to X
\eeq
such that its restriction onto $E$ is a linear differential
operator and its kernel contains $\cE$, i.e.,
\mar{45f1}\beq
\Phi=\op\sum_{0\leq|\La|} \Phi_a^\La \chi^a_\La, \qquad
\op\sum_{0\leq|\La|} \Phi_a^\La \cE^a_\La=0. \label{45f1}
\eeq
\end{definition}

Any differential operator admits Noether identities, e.g.,
\mar{v44}\beq
\Phi= \op\sum_{0\leq|\La|,|\Si|} T_{ab}^{\La\Si}d_\Si\cE^b
\chi^a_\La, \qquad T_{ab}^{\La\Si}=- T_{ba}^{\Si\La}. \label{v44}
\eeq
Therefore, they  must be separated into the trivial and
non-trivial ones.

\begin{lemma} \label{v203} \mar{v203}
One can associate to $\cE$ a chain complex whose boundaries vanish
on Ker$\cE$.
\end{lemma}

\begin{proof}
Let us consider the composite graded manifold $(Y,\gA_{E_Y})$
modelled over the vector bundle $E_Y\to Y$. Let
$\cS^0_\infty[E_Y;Y]$ be the ring of graded functions on the
infinite order jet manifold $J^\infty Y$ possessing the local
generating basis $(y^i,\ve^a)$ of Grassmann parity $[\ve^a]=1$. It
is provided with the nilpotent graded derivation
\mar{vv}\beq
\ol\dl=\rdr_a\cE^a. \label{vv}
\eeq
whose definition is independent of the choice of the local
generating basis. Then we have the chain complex
\mar{v042x}\beq
0\lto \im\dl \llr^{\ol\dl} \cS^0_\infty[E_Y;Y]_1 \llr^{\ol\dl}
\cS^0_\infty[E_Y;Y]_2 \label{v042x}
\eeq
of graded functions of antifield number $k\leq 2$. Its
one-boundaries $\ol\dl\Phi$, $\Phi\in \cS^0_\infty[E_Y;Y]_2$,  by
very definition, vanish on Ker$\cE$.
\end{proof}

Every one-cycle
\mar{0712x}\beq
\Phi= \op\sum_{0\leq|\La|} \Phi_a^\La \ve^a_\La\in
\cS^0_\infty[E_Y;Y]_1 \label{0712x}
\eeq
of the complex (\ref{v042x}) defines a linear differential
operator on pull-back bundle $E_Y$ (\ref{46f1}) such that it is
linear on $E$ and its kernel contains $\cE$, i.e.,
\mar{0713x}\beq
\dl \Phi=0, \qquad \op\sum_{0\leq|\La|} \Phi_a^\La d_\La \cE^a =0.
\label{0713x}
\eeq
In accordance with Definition \ref{46d1}, the one-cycles
(\ref{0712x}) define the Noether identities (\ref{0713x}) of a
differential operator $\cE$. These Noether identities are trivial
if a cycle is a boundary, i.e., it takes a form (\ref{v44}).
Accordingly, non-trivial Noether identities modulo the trivial
ones are associated to elements of the homology $H_1(\dl)$ of the
complex (\ref{0712x}).

A differential operator is called degenerate if it obeys
non-trivial Noether identities.

One can say something more if the $\cO_\infty^0$-module $H_1(\dl)$
is finitely generated, i.e., it possesses the following particular
structure. There are elements $\Delta\in H_1(\dl)$ making up a
projective $C^\infty(X)$-module $\cC_{(0)}$ of finite rank which,
by virtue of the Serre -- Swan theorem, is isomorphic to the
module of sections of some vector bundle $E_0\to X$. Let
$\{\Delta^r\}$:
\mar{71}\beq
\Delta^r=\op\sum_{0\leq|\La|} \Delta_a^{\La r} \ve^a_\La, \qquad
\Delta_a^{\La r}\in \cO_\infty^0, \label{v71}
\eeq
be local bases for this $C^\infty(X)$-module. Then every element
$\Phi\in H_1(\dl)$ factorizes as
\mar{v63}\beq
\Phi= \op\sum_{0\leq|\Xi|} G_r^\Xi d_\Xi \Delta^r, \qquad
G_r^\Xi\in \cO_\infty Y, \label{v63}
\eeq
through elements of $\cC_{(0)}$, i.e., any Noether identity
(\ref{0713x}) is a corollary of the Noether identities
\mar{v64x}\beq
 \op\sum_{0\leq|\La|} \Delta_a^{\La r} d_\La \cE^a=0,
\label{v64x}
\eeq
called complete Noether identities.

\begin{remark} \label{46n1} \mar{46n1}
Given an integer $N\geq 1$, let $E_1, \ldots, E_N$ be vector
bundles over $X$. Let us denote
\be
\cP^0_\infty\{N\}=\cS_\infty^0[E_{N-1}\op\oplus_X\cdots\op\oplus_X
E_1\op\oplus_XE_Y;Y\op\times_X E_0\op\oplus_X\cdots\op\oplus_X
E_N]
\ee
if $N$ is even and
\be
\cP^0_\infty\{N\}=\cS_\infty^0[E_N\op\oplus_X\cdots\op\oplus_X
E_1\op\oplus_XE_Y;Y\op\times_X E_0\op\oplus_X\cdots\op\oplus_X
E_{N-1}]
\ee
if $N$ is odd.
\end{remark}

\begin{lemma} \label{v137} \mar{v137}
If the homology $H_1(\dl)$ of the complex (\ref{v042x}) is
finitely generated, this complex  can be extended to the one-exact
complex (\ref{v66x}) with a boundary operator whose nilpotency
conditions are equivalent to complete Noether identities.
\end{lemma}

\begin{proof}
Let us consider the graded commutative ring $\cP_\infty^0\{0\}$.
It possesses the local generating basis $\{y^i, \ve^a, \ve^r\}$ of
Grassmann parity $[\ve^r]=0$ and antifield number Ant$[\ve^r]=2$.
This ring is provided with the nilpotent graded derivation
\mar{v204}\beq
\dl_0=\dl + \rdr_r\Delta^r. \label{v204}
\eeq
Its nilpotency conditions are equivalent to the complete Noether
identities (\ref{v64x}). Then the module $\cP_\infty^0\{0\}_{\leq
3}$ of graded functions of antifield number $\leq 3$ is decomposed
into the chain complex
\mar{v66x}\beq
0\lto \im\dl \llr^\dl \cS_\infty^0[E_Y;Y]_1\llr^{\dl_0}
\cP_\infty^0\{0\}_2\llr^{\dl_0} \cP_\infty^0\{0\}_3. \label{v66x}
\eeq
Let $H_*(\dl_0)$ denote its homology. We have
\be
H_0(\dl_0)=H_0(\dl)=0.
\ee
Furthermore, any one-cycle $\Phi$ up to a boundary takes the form
(\ref{v63}) and, therefore, it is a $\dl_0$-boundary
\be
\Phi= \op\sum_{0\leq|\Si|} G_r^\Xi d_\Xi \Delta^r
=\dl_0\left(\op\sum_{0\leq|\Si|} G_r^\Xi \ve^r_\Xi\right).
\ee
Hence, $H_1(\dl_0)=0$, i.e., the complex (\ref{v66x}) is
one-exact.
\end{proof}

Let us consider the second homology $H_2(\dl_0)$ of the complex
(\ref{v66x}). Its two-chains  read
\mar{v77}\beq
\Phi= G + H= \op\sum_{0\leq|\La|} G_r^\La \ve^r_\La  +
\op\sum_{0\leq|\La|,|\Si|} H_{ab}^{\La\Si} \ve^a_\La \ve^b_\Si.
\label{v77}
\eeq
Its two-cycles define the first-stage Noether identities
\mar{v79x}\beq
\dl_0 \Phi=0, \qquad \op\sum_{0\leq|\La|} G_r^\La d_\La\Delta^r
+\dl H=0. \label{v79x}
\eeq
Conversely, let the equality (\ref{v79x}) hold. Then it is a cycle
condition of the two-chain (\ref{v77}). The first-stage Noether
identities (\ref{v79x}) are trivial either if a two-cycle $\Phi$
(\ref{v77}) is a boundary or its summand $G$ vanishes on Ker$\cE$.

\begin{lemma} \label{v134} \mar{v134}
First-stage Noether identities can be identified with nontrivial
elements of the homology $H_2(\dl_0)$ iff any $\dl$-cycle $\Phi\in
\cS_\infty^0[E_Y;Y]_2$ is a $\dl_0$-boundary.
\end{lemma}

\begin{proof} The proof is similar to that of Lemma \ref{v134'}
\cite{oper}.
\end{proof}

A degenerate differential operator is called reducible if there
exist non-trivial first-stage Noether identities.

If the condition of Lemma \ref{v134} is satisfied, let us assume
that non-trivial first-stage Noether identities are finitely
generated as follows. There exists a graded projective
$C^\infty(X)$-module $\cC_{(1)}\subset H_2(\dl_0)$ of finite rank
possessing a local basis $\Delta_{(1)}:$
\be
\Delta^{r_1}=\op\sum_{0\leq|\La|} \Delta^{\La r_1}_r \ve^r_\La +
h^{r_1},
\ee
such that any element $\Phi\in H_2(\dl_0)$ factorizes as
\mar{v80}\beq
\Phi= \op\sum_{0\leq|\Xi|} \Phi_{r_1}^\Xi d_\Xi \Delta^{r_1}
\label{v80}
\eeq
through elements of $\cC_{(1)}$. Thus, all non-trivial first-stage
Noether identities (\ref{v79x}) result from the equalities
\mar{v82}\beq
 \op\sum_{0\leq|\La|} \Delta^{r_1\La}_r d_\La \Delta^r +\dl
h^{r_1} =0, \label{v82}
\eeq
called the complete first-stage Noether identities.

\begin{lemma} \label{v139} \mar{v139} If non-trivial first-stage Noether
identities are finitely generated, the one-exact complex
(\ref{v66x}) is extended to the two-exact one (\ref{v87}) with a
boundary operator whose nilpotency conditions are equivalent to
complete Noether and first-stage Noether identities.
\end{lemma}

\begin{proof} By virtue of the Serre -- Swan theorem, the module
$\cC_{(1)}$ is isomorphic to a module of sections of some  vector
bundle $E_1\to X$. Let us consider the ring $\cP^0_\infty\{1\}$ of
graded functions on $J^\infty Y$ possessing the local generating
bases $\{y^i,\ve^a,\ve^r,\ve^{r_1}\}$ of Grassmann parity
$[\ve^{r_1}]=1$ and antifield number Ant$[\ve^{r_1}]=3$. It can be
provided with the nilpotent graded derivation
\mar{v205}\beq
\dl_1=\dl_0 + \rdr_{r_1} \Delta^{r_1}. \label{v205}
\eeq
Its nilpotency conditions are equivalent to the complete Noether
identities (\ref{v64x}) and the complete first-stage Noether
identities (\ref{v82}). Then the module $\cP^0_\infty\{1\}_{\leq
4}$ of graded functions of antifield number $\leq 4$ is decomposed
into the chain complex
\mar{v87}\beq
0\lto \im\dl \llr^\dl \cS_\infty[E_Y;Y]_1\llr^{\dl_0}
\cP^0_\infty\{0\}_2\llr^{\dl_1} \cP^0_\infty\{1\}_3 \llr^{\dl_1}
\cP^0_\infty\{1\}_4. \label{v87}
\eeq
Let $H_*(\dl_1)$ denote its homology. It is readily observed that
\be
H_0(\dl_1)=H_0(\dl)=0, \qquad H_1(\dl_1)=H_1(\dl_0)=0.
\ee
By virtue of the expression (\ref{v80}), any two-cycle of the
complex (\ref{v87}) is a boundary
\be
 \Phi= \op\sum_{0\leq|\Xi|} \Phi_{r_1}^\Xi d_\Xi \Delta^{r_1}
=\dl_1\left(\op\sum_{0\leq|\Xi|} \Phi_{r_1}^\Xi
\ve_\Xi^{r_1}\right).
\ee
It follows that $H_2(\dl_1)=0$, i.e., the complex (\ref{v87}) is
two-exact.
\end{proof}

If the third homology $H_3(\dl_1)$ of the complex (\ref{v87}) is
not trivial, its elements correspond to second-stage Noether
identities, and so on. Iterating the arguments, we come to the
following.

A degenerate differential operator $\cE$ is called $N$-stage
reducible if it admits finitely generated non-trivial $N$-stage
Noether identities, but no non-trivial $(N+1)$-stage ones. It is
characterized as follows \cite{oper}.

$\bullet$ There are graded vector bundles $E_0,\ldots, E_N$ over
$X$, and the graded commutative ring $\cS^0_\infty[E_Y;Y]$ is
enlarged to the graded commutative ring $\ol\cP^0_\infty\{N\}$
with the local generating basis
\be
(y^i, \ve^a, \ve^r, \ve^{r_1}, \ldots, \ve^{r_N})
\ee
of Grassmann parity $[\ve^{r_k}]=(k+1)\,$mod2 and antifield number
Ant$[\ve_\La^{r_k}]=k+2$.

$\bullet$ The graded commutative ring $\ol\cP^0_\infty\{N\}$ is
provided with the nilpotent right graded derivation
\mar{v92x}\ben
&&\dl_{\rm KT}=\dl_N=\dl_0 + \op\sum_{1\leq k\leq N}\rdr_{r_k} \Delta^{r_k},
\label{v92x}\\
&& \Delta^{r_k}=\op\sum_{0\leq|\La|}
\Delta^{\La r_k}_{r_{k-1}} \ve_\La^{r_{k-1}} + \op\sum_{0\leq \Si,
0\leq\Xi}(h^{\Xi\Si r_k}_{a r_{k-2}} \ve^a_\Xi
\ve_\Si^{r_{k-2}}+...), \nonumber
\een
of antifield number -1.

$\bullet$ With this graded derivation, the module
$\cP^0_\infty\{N\}_{\leq N+3}$ of graded functions of antifield
number $\leq (N+3)$ is decomposed into the exact  Koszul -- Tate
complex
\mar{v92xx}\ben
&&0\lto \im \dl \llr^\dl \cS^0_\infty[E_Y;Y]_1\llr^{\dl_0}
\cP_\infty\{0\}_2\llr^{\dl_1} \cP_\infty^0\{1\}_3\cdots \label{v92xx}\\
&& \qquad
 \llr^{\dl_{N-1}} \cP_\infty^0\{N-1\}_{N+1}
\llr^{\dl_{\rm KT}} \cP^0_\infty\{N\}_{N+2}\llr^{\dl_{\rm KT}}
\cP_\infty\{N\}_{N+3}, \nonumber
\een
which satisfies the following homology regularity condition.

\begin{condition} \label{v155x} \mar{v155x} Any
$\dl_{k<N-1}$-cycle
\be
\Phi\in \cP_\infty^0\{k\}_{k+3}\subset \cP_\infty^0\{k+1\}_{k+3}
\ee
is a $\dl_{k+1}$-boundary.
\end{condition}

$\bullet$ The nilpotentness $\dl_{\rm KT}^2=0$ of the Koszul --
Tate operator (\ref{v92x}) is equivalent to the complete
non-trivial Noether identities (\ref{v64x}) and the  complete
non-trivial $(k\leq N)$-stage Noether identities
\mar{46f11}\ben
&& \op\sum_{0\leq|\La|} \Delta^{\La r_k}_{r_{k-1}}d_\La
\left(\op\sum_{0\leq|\Si|} \Delta^{\Si r_{k-1}}_{r_{k-2}}
\ve_\Si^{r_{k-2}}\right) + \label{46f11}\\
&& \qquad \dl\left(\op\sum_{0\leq \Si, \Xi}h^{\Xi\Si r_k}_{a r_{k-2}}
\ve_\Xi^a \ve_\Si^{r_{k-2}}\right)=0. \nonumber
\een

Let us study the following example of reducible Noether identities
of a differential operator which is relevant to topological BF
theory (Section 6.4).

\begin{example} \label{46e1} \mar{46e1}
Let us consider the fibre bundles
\mar{iio1}\beq
Y=X\times \mathbb R, \qquad E=\op\w^{n-1} TX, \qquad 2<n,
\label{iio1}
\eeq
coordinated by $(x^\la,y)$ and $(x^\la,\chi^{\m_1\ldots
\m_{n-1}})$, respectively. We study the $E$-valued differential
operator
\mar{v183}\beq
\cE^{\m_1\ldots \m_{n-1}} = - \e^{\m\m_1\ldots \m_{n-1}}y_\m,
\label{v183}
\eeq
where $\e$ is the Levi -- Civita symbol. It defines the first
order differential equation
\mar{iio}\beq
d_Hy=0 \label{iio}
\eeq
on the fibre bundle $Y$ (\ref{iio1}).

Putting
\be
E_Y= \mathbb R\op\times_X \op\w^{n-1} TX,
\ee
let us consider the graded commutative ring $\cS^*_\infty[E_Y;Y]$
of graded functions on $J^\infty Y$. It possesses the local
generating basis $(y,\ve^{\m_1\ldots \m_{n-1}})$ of Grassmann
parity $[\ve^{\m_1\ldots \m_{n-1}}]=1$ and antifield number
Ant$[\ve^{\m_1\ldots \m_{n-1}}]=1$. With the nilpotent derivation
\be
\ol\dl=\frac{\rdr}{\dr \ve^{\m_1\ldots \m_{n-1}}} \cE^{\m_1\ldots
\m_{n-1}},
\ee
we have the complex (\ref{v042x}). Its one-chains read
\be
\Phi= \op\sum_{0\leq |\La|} \Phi^\La_{\m_1\ldots \m_{n-1}}
\ve^{\m_1\ldots \m_{n-1}}_\La,
\ee
and the cycle condition $\ol\dl\Phi=0$ takes a form
\mar{v189}\beq
\Phi^\La_{\m_1\ldots \m_{n-1}} \cE^{\m_1\ldots \m_{n-1}}_\La=0.
\label{v189}
\eeq
This equality is satisfied iff
\be
\Phi^{\la_1\ldots \la_k}_{\m_1\ldots\m_{n-1}}\e^{\m\m_1\ldots
\m_{n-1}}=- \Phi^{\m\la_2\ldots
\la_k}_{\m_1\ldots\m_{n-1}}\e^{\la_1\m_1\ldots \m_{n-1}}.
\ee
It follows that $\Phi$ factorizes as
\be
\Phi= \op\sum_{0\leq |\Xi|} G_{\nu_2\ldots\nu_{n-1}}^\Xi
d_\Xi\Delta^{\nu_2\ldots\nu_{n-1}}\om
\ee
through graded functions
\mar{v190}\ben
&&\Delta^{\nu_2\ldots\nu_{n-1}}=\Delta^{\la,\nu_2\ldots\nu_{n-1}}_{\al_1\ldots\al_{n-1}}
\ve^{\al_1\ldots\al_{n-1}}_\la=
\label{v190}\\
&& \qquad \dl^\la_{\al_1}\dl^{\nu_2}_{\al_2}\cdots
\dl^{\nu_{n-1}}_{\al_{n-1}} \ve^{\al_1\ldots\al_{n-1}}_\la=
d_{\nu_1}\ve^{\nu_1\nu_2\ldots\nu_{n-1}}, \nonumber
\een
which provide the complete Noether identities
\mar{v191}\beq
d_{\nu_1}\cE^{\nu_1\nu_2\ldots\nu_{n-1}}=0. \label{v191}
\eeq
They can be written in the form
\mar{iio2}\beq
d_Hd_Hy=0. \label{iio2}
\eeq

The graded functions (\ref{v190}) form a basis for a projective
$C^\infty(X)$-module of finite rank which is isomorphic to the
module of sections of the vector bundle
\be
E_0=\op\w^{n-2} TX.
\ee
Therefore, let us extend the graded commutative ring
$\cS^0_\infty[E_Y;Y]$ to that $\cP^*_\infty\{0\}$ (see Remark
\ref{46f1}) possessing the local generating basis
\be
(y, \ve^{\m_1\ldots \m_{n-1}}, \ve^{\m_2\ldots \m_{n-1}}\},
\ee
where $\ve^{\m_2\ldots \m_{n-1}}$ are even Noether antifields of
antifield number 2. We have the nilpotent graded derivation
\be
\dl_0= \ol\dl + \frac{\rdr}{\dr \ve^{\m_2\ldots \m_{n-1}}}
\Delta^{\m_2\ldots \m_{n-1}}
\ee
of $\cP^0_\infty\{0\}$. Its nilpotency is equivalent to the
complete Noether identities (\ref{v191}). Then we obtain the
one-exact complex (\ref{v66x}).

Iterating the arguments, let us consider the vector bundles
\be
&& E_k=\op\w^{n-k-2} TX, \qquad k=1,\ldots, n-3,\\
&& E_{N=n-2}=X\times \mathbb R
\ee
and the graded commutative ring $\cP^0_\infty\{N\}$ (see Remark
\ref{46n1}), possessing the local generating basis
\be
(y,\ve^{\m_1\ldots \m_{n-1}}, \ve^{\m_2\ldots
\m_{n-1}},\ldots,\ve^{\m_{n-1}},\ve)
\ee
of Grassmann parity
\be
[\ve^{\m_{k+2}\ldots \m_{n-1}}]=k\,{\rm mod}\,2, \qquad [\ve]=n,
\ee
and of antifield number
\be
{\rm Ant}[\ve^{\m_{k+2}\ldots \m_{n-1}}]=k+2, \qquad {\rm
Ant}[\ve]=n.
\ee
It is provided with the nilpotent graded derivation
\mar{va202}\ben
&& \dl_{\rm KT}=\dl_0 + \op\sum_{1\leq k\leq n-3}\frac{\rdr}{\dr
\ve^{\m_{k+2}\ldots \m_{n-1}}} + \frac{\rdr}{\dr \ve}
d_{\m_{n-1}}\ve^{\m_{n-1}}, \label{va202}\\
&& \Delta^{\m_{k+2}\ldots \m_{n-1}}=d_{\m_{k+1}}
\ve^{\m_{k+1}\m_{k+2}\ldots \m_{n-1}}, \nonumber
\een
of antifield number -1. Its nilpotency results from the complete
Noether identities (\ref{v191}) and the equalities
\mar{v212}\beq
d_{\m_{k+2}}\Delta^{\m_{k+2}\ldots \m_{n-1}}=0, \qquad
k=0,\ldots,n-3, \label{v212}
\eeq
which are the $(k+1)$-stage Noether identities (\ref{46f11}). Then
the Koszul -- Tate complex (\ref{v92xx}) reads
\mar{v203x}\ben
&&0\lto \im \ol\dl \llr^{\ol\dl} \cS^0_\infty[E_Y;Y]_1\llr^{\dl_0}
\cP^0_\infty\{0\}_2\llr^{\dl_1} \cP^0_\infty\{1\}_3\cdots
\label{v203x}\\
&& \qquad
 \llr^{\dl_{n-3}} \cP^0_\infty\{n-3\}_{n-1}
\llr^{\dl_{\rm KT}} \cP^0_\infty\{n-2\}_n\llr^{\dl_{\rm KT}}
\cP^0_\infty\{n-2\}_{n+1}. \nonumber
\een
It obeys Condition \ref{v155x} as follows.

\begin{lemma} \label{v220} \mar{v220}
Any $\dl_k$-cycle $\Phi\in \cP^0_\infty\{k\}_{k+3}$ up to a
$\dl_k$-boundary takes a form
\mar{v218}\ben
&& \Phi=\op\sum_{(k_1+\cdots +k_i+3i=k+3)}\sum_{(0\leq|\La_1|,\ldots,
|\La_i|)}G^{\La_1\cdots \La_i}_{\m^1_{k_1+2}\ldots
\m^1_{n-1};\ldots; \m^i_{k_i+2}\ldots \m^i_{n-1}} \label{v218} \\
&&\qquad d_{\La_1} \Delta^{\m^1_{k_1+2}\ldots \m^1_{n-1}}\cdots
d_{\La_i} \Delta^{\m^i_{k_i+2}\ldots \m^i_{n-1}}, \qquad
k_j=-1,0,1,\ldots, n-3, \nonumber
\een
where $k_j=-1$ stands for $\ve^{\m_1\ldots\m_{n-1}}$ and
\be
\Delta^{\m_1\ldots\m_{n-1}}=\cE^{\m_1\ldots\m_{n-1}}.
\ee
It follows that $\Phi$ is a $\dl_{k+1}$-boundary.
\end{lemma}

\begin{proof}
Let us choose some basis element $\ve^{\m_{k+2}\ldots \m_{n-1}}$
and denote it, simply, by $\ve$. Let $\Phi$ contain a summand
$\f_1 \ve$, linear in $\ve$. Then the cycle condition reads
\be
\dl_k\Phi=\dl_k(\Phi-\f_1\ve) + (-1)^{[\ve]}\dl_k(\f_1) \ve + \f
\Delta=0, \qquad \Delta=\dl_k\ve.
\ee
It follows that $\Phi$ contains a summand $\psi\Delta$ such that
\be
(-1)^{[\ve]+1}\dl_k(\psi)\Delta +\f\Delta=0.
\ee
This equality implies the relation
\mar{v213}\beq
\f_1=(-1)^{[\ve]+1}\dl_k(\psi) \label{v213}
\eeq
because the reduction conditions (\ref{v212}) involve total
derivatives of $\Delta$, but not $\Delta$. Hence,
\be
\Phi=\Phi' +\dl_k(\psi \ve),
\ee
where $\Phi'$ contains no term linear in $\ve$. Furthermore, let
$\ve$ be even and $\Phi$ have a summand $\sum \f_r \ve^r$
polynomial in $\ve$. Then the cycle condition leads to the
equalities
\be
\f_r\Delta=-\dl_k\f_{r-1}, \qquad r\geq 2.
\ee
Since $\f_1$ (\ref{v213}) is $\dl_k$-exact, then $\f_2=0$ and,
consequently, $\f_{r>2}=0$. Thus, a cycle $\Phi$ up to a
$\dl_k$-boundary contains no term polynomial in $c$. It reads
\mar{v217}\ben
&& \Phi=\op\sum_{(k_1+\cdots +k_i+3i=k+3)}\sum_{(0<|\La_1|,\ldots,
|\La_i|)}G^{\La_1\cdots \La_i}_{\m^1_{k_1+2}\ldots
\m^1_{n-1};\ldots; \m^i_{k_i+2}\ldots \m^i_{n-1}} \nonumber\\
&& \qquad \ve^{\m^1_{k_1+2}\ldots \m^1_{n-1}}_{\La_1}\cdots
\ve_{\La_i}^{\m^i_{k_i+2}\ldots \m^i_{n-1}}. \label{v217}
\een
However, the terms polynomial in $\ve$ may appear under general
coordinate transformations
\be
\ve'^{\nu_{k+2}\ldots \nu_{n-1}}=\det\left(\frac{\dr x^\al}{\dr
x'^\bt}\right) \frac{\dr x'^{\nu_{k+2}}}{\dr x^{\m_{k+2}}}\cdots
\frac{\dr x'^{\nu_{n-1}}}{\dr x^{\m_{n-1}}}\ve^{\m_{k+2}\ldots
\m_{n-1}}
\ee
of a chain $\Phi$ (\ref{v217}). In particular, $\Phi$ contains the
summand
\be
\op\sum_{k_1+\cdots +k_i+3i=k+3}F_{\nu^1_{k_1+2}\ldots
\nu^1_{n-1};\ldots; \nu^i_{k_i+2}\ldots \nu^i_{n-1}}
\ve'^{\nu^1_{k_1+2}\ldots \nu^1_{n-1}}\cdots
\ve'^{\nu^i_{k_i+2}\ldots \nu^i_{n-1}},
\ee
which must vanish if $\Phi$ is a cycle. This takes place only if
$\Phi$ factorizes through the graded densities
$\Delta^{\m_{k+2}\ldots \m_{n-1}}$ (\ref{va202}) in accordance
with the expression (\ref{v218}).
\end{proof}

Following the proof of Lemma \ref{v220}, one also can show that
any $\dl_k$-cycle $\Phi\in \cP^0_\infty\{k\}_{k+2}$ up to a
boundary takes a form
\be
\Phi=\op\sum_{0\leq|\La|}G^\La_{\m_{k+2}\ldots \m_{n-1}} d_\La
\Delta^{\m_{k+2}\ldots \m_{n-1}},
\ee
i.e., the homology $H_{k+2}(\dl_k)$ of the complex (\ref{v203x})
is finitely generated by the cycles $\Delta^{\m_{k+2}\ldots
\m_{n-1}}$.
\end{example}

\section{Classical field models}

As was mentioned above, classical field theory of even and odd
fields is formulated adequately a Lagrangian theory on graded
bundles \cite{book09,sard08,sard13}. This Section provides some
examples of relevant field models.

\subsection{Gauge theory on principal bundles}

In classical gauge theory, gauge fields are conventionally
described as principal connections on principal bundles
\cite{book09,book00,book13}. We consider their first order Yang --
Mills Lagrangian theory.

Principal connections on a principal bundle $\pi_P:P\to X$ with a
structure Lie group $G$ are connections on $P$ which are
equivariant with respect to the right action
\mar{1}\ben
&& G: G\op\times_X P \ar_X P,  \label{1}\\
&& G: p\to pg, \qquad \pi_P(p)=\pi_P(pg), \qquad p\in P, \nonumber
\een
of a structure group $G$ on $P$. In order to describe them, we
follow the definition of connections on a fibre bundle $Y\to X$ as
global sections of the affine jet bundle $J^1Y\to X$
\cite{book,book00,sau}.

Let $J^1P$ be a first order jet manifold of a principal $G$-bundle
$P\to X$. Then connections on a principal bundle $P\to X$ are
global sections $A: P\to J^1P$ of an affine jet bundle $J^1P\to
P$. In order to describe principal connections on $P\to X$, let us
consider the jet prolongation
\mar{53f1}\beq
G\ni g: j^1_xp\to (j^1_xp)g =j^1_x(pg). \label{53f1}
\eeq
the action ({1}) of $G$ onto $J^1P$. Taking the quotient of an
affine jet bundle $J^1P\to P$ by $G$ (\ref{53f1}), we obtain an
affine bundle
\mar{B1}\beq
C=J^1P/G\to X\label{B1}
\eeq
modelled over a vector bundle
\be
\ol C=T^*X\op\ot_X V_GP\to X.
\ee
Hence, there is the canonical vertical splitting
\be
VC= C\op\ot_X \ol C
\ee
of the vertical tangent bundle $VC$ of $C\to X$. Principal
connections on a principal bundle $P\to X$ are identified with
global sections of the fibre bundle $C\to X$ (\ref{B1}), called
the bundle of principal connections. Given an atlas
\mar{51f2}\beq
\Psi_P=\{(U_\al,\psi^P_\al),\vr_{\al\bt}\} \label{51f2}
\eeq
of a principal bundle $P$, the bundle of principal connections $C$
(\ref{B1}) is endowed with bundle coordinates $(x^\la,a^m_\m)$
possessing the transformation rule
\be
\vr(a^m_\m)\ve_m=(a^m_\nu{\rm Ad}_{\vr^{-1}}(\ve_m) +
R^m_\nu\ve_m)\frac{\dr x^\nu}{\dr x'^\m}.
\ee
If $G$ is a matrix group, this transformation rule reads
\be
\vr(a^m_\m)\ve_m = (a^m_\nu\vr^{-1}(\ve_m)\vr -
\dr_\m(\vr^{-1})\vr)\frac{\dr x^\nu}{\dr x'^\m}.
\ee
A glance at this expression shows that the bundle of principal
connections $C$ fails to be a bundle with a structure group $G$.

We consider first order Lagrangian theory on a fibre bundle $Y=C$
(see Example \ref{first}). Its structure algebra
$\cS^*_\infty[F;Y]=\cS^*_\infty[C]$ (\ref{5.77a}) is the graded
differential algebra $\cS^*_\infty[C]=\cO^*_\infty C$ (\ref{ppp})
of exterior forms on jet manifolds $J^rC$ of $C\to X$. Its first
order Lagrangian (\ref{23f2}) is a density
\mar{57f1}\beq
L=\cL\om: J^1C\to \op\w^n T^*X \label{57f1}
\eeq
on a first order jet manifold $J^1C$ possessing the adapted
coordinates $(x^\m,a^m_\m, a^m_{\la\m})$. The corresponding Euler
-- Lagrange operator (\ref{305}) reads
\mar{57f2}\beq
\cE_L= \cE_m^\m\thh^m_\m\w\om=(\dr_m^\m- d_\la\dr^{\la\m}_m)\cL
\thh^m_\m\w\om. \label{57f2}
\eeq
Its kernel defines the Euler -- Lagrange equation
\mar{57f2'}\beq
\cE_m^\m=(\dr_m^\m- d_\la\dr^{\la\m}_m)\cL=0. \label{57f2'}
\eeq

In classical gauge theory, gauge transformations are defined as
vertical principal automorphisms of a principal bundle $P$ which
are equivariant with respect to the action (\ref{1}) of a
structure group $G$ i.e.,
\mar{55ff1}\beq
\Phi_P(pg)=\Phi_P(p)g, \qquad g\in G, \qquad p\in P. \label{55ff1}
\eeq

In order to describe gauge symmetries of gauge theory on a
principal bundle $P$, it is sufficient to consider (local)
one-parameter groups of principal automorphisms of $P$ and their
infinitesimal generators which are $G$-invariant projectable
vector fields $\xi$ on $P$, called the principal vector fields.
They are represented by sections of the quotient
\mar{b1.205}\beq
V_GP=VP/G \label{b1.205}
\eeq
of the vertical tangent bundle $VP$ of $P\to X$ with respect to
the tangent prolongation of the action ({1}) of $G$ on $P$. It is
a $P$-associated bundle whose typical fibre is a right Lie algebra
$\cG_r$ of $G$ subject to the adjoint representation of a
structure group $G$. Therefore, $V_GP$ (\ref{b1.205}) is called
the Lie algebra bundle. Given the bundle atlas $\Psi_P$
(\ref{51f2}) of $P$, the Lie algebra bundle $V_GP$ is provided
with bundle coordinates $(U_\al; x^\m,\chi^m)$ with respect to the
fibre frames $\{e_m=\psi_\al^{-1}(x)(\ve_m)\}$, where $\{\ve_m\}$
is a basis for the Lie algebra $\cG_r$. These coordinates possess
the transition functions
\be
\vr(\chi^m) \ve_m=\chi^m {\rm Ad}_{\vr^{-1}}(\ve_m).
\ee
Then sections of the Lie algebra bundle $V_\G$ read
\mar{b3106}\beq
\xi=\xi^r e_r. \label{b3106}
\eeq
They form a finite-dimensional Lie $C^\infty(X)$-algebra with
respect to the Lie bracket
\mar{1129'}\beq
[\xi,\eta]= c_{pq}^r\xi^p\eta^q e_r, \label{1129'}
\eeq
where $c_{pq}^r$ are the structure constants of a Lie algebra
$\cG_r$.

Any (local) one-parameter group of principal automorphism $\Phi_P$
(\ref{55ff1}) of a principal bundle $P$ admits the jet
prolongation $J^1\Phi_P$ to a one-parameter group of
$G$-equivariant automorphism of the jet manifold $J^1P$ which, in
turn, yields a one-parameter group of principal automorphisms
$\Phi_C$ of the bundle of principal connections $C$ (\ref{B1}).
Its infinitesimal generator is a vector field on $C$. As a
consequence, any principal vector field $\xi$ (\ref{b3106}) yields
a principal vector field
\mar{279}\beq
u_\xi = (\dr_\m\xi^r + c_{pq}^r a_\m^p\xi^q)\dr_r^\m. \label{279}
\eeq
on $C$ \cite{book00}.

A glance at the expression (\ref{279}) shows that one can think of
the principal vector fields $u_\xi$ as being a linear first order
differential operator on a vector space of sections of the Lie
algebra bundle $V_GC$ (\ref{b1.205}) with values in a vector space
of vertical vector fields on the bundle of principal connections
$C$ (\ref{B1}), i.e., $u_\xi$ (\ref{279}) are even gauge
transformations (see Remark \ref{qwe}) with even gauge parameter
functions $\xi$ (\ref{b3106}). However, since gauge symmetries in
the second Noether theorem \ref{w35}) are odd, we modify the
definition of gauge transformations in gauge theory in accordance
with Definition \ref{sgauge} as follows.

Let us treat $V_GP\to X$ as an odd vector bundle, and let
$(X,\gA_{V_GP}$ be the corresponding simple graded manifold. Then
let us consider the composite bundle
\mar{vgv}\beq
V_GP\op\times_X C\to C\to X, \label{vgv}
\eeq
coordinated by $(x^\m, a^m_\m, \chi^r)$, and the graded bundle
$(X,C,\gA_{V_GP\times_X C})$ (\ref{olo}) modelled over this
composite bundle together with the local generating basis $(x^\m,
a^m_\m, c^r)$ whose terms $c^r$ are odd. Let
\mar{yyy}\beq
S_\infty^*[V_GP\op\times_X C; C] \label{yyy}
\eeq
 be the DBGA
(\ref{tyt'}) together with the monomorphisms (\ref{tgv}):
\be
\cO^*_\infty C\to S_\infty^*[V_GP\op\times_X C; C], \qquad
S^*_\infty[V_GP;X]\to S_\infty^*[V_GP\op\times_X C; C].
\ee

By inspection of transition functions of the principal vector
field $u_\xi$ (\ref{279}), one can justify the existence of an odd
contact derivation $J^\infty u$ of the DBGA (\ref{yyy}) generated
by a generalized vector field
\mar{57f30'}\beq
u=(c_\m^r + c_{pq}^r a_\m^pc^q)\dr_r^\m \label{57f30'}
\eeq
on a graded bundle $(X,C,\gA_{V_GP\times_X C})$. The graded
derivation $u$ (\ref{57f30'}) obviously vanishes on a subring
\be
S^0_\infty[V_GP;X]\subset S_\infty^0[V_GP\op\times_X C; C].
\ee
Consequently, it is a gauge transformation of a Lagrangian system
\be
\cO^*_\infty C\subset S_\infty^*[V_GP\op\times_X C; C],
\ee
parameterized by odd ghosts $c^r$.

In Yang -- Mills gauge theory, its Lagrangian $L$ (\ref{57f1}) is
required to be invariant under the gauge transformation $u$
(\ref{57f30'}), i.e., it it is an exact gauge symmetry. The
corresponding condition reads
\mar{57f4}\ben
&& \bL_{J^1u}L=0, \label{57f4}\\
&& J^1u =u +(c_{\la\m}^r + c_{pq}^r a_\m^p c_\la^q
+c_{pq}^r a_{\la\m}^pc^q )\dr_r^{\la\m}, \nonumber
\een
(cf. (\ref{23f41})). In this case, the first variational formula
(\ref{J4}) for the Lie derivative (\ref{57f4}) takes a form
\mar{57f5}\beq
0= (c_\m^r + c_{pq}^r a_\m^pc^q)\cE_r^\m + d_\la[(c_\m^r +
c_{pq}^r a_\m^pc^q)\dr^{\la\m}_r\cL)]. \label{57f5}
\eeq
It leads to the gauge invariance conditions (\ref{g4g}) --
(\ref{g7g}) which read
\mar{57f7a-c}\ben
&& \dr_p^{\m\la}\cL + \dr_p^{\la\m}\cL = 0, \label{57f7a}\\
&& \cE_r^\m +d_\la \dr^{\la\m}_r\cL
+ c_{pr}^q a_\nu^p\dr^{\m\nu}_q\cL=0, \label{57f7b}\\
&& c_{pq}^r(a_\m^p \cE^\m_r+
d_\la( a_\m^p\dr^{\la\m}_r\cL))=0. \label{57f7c}
\een
One can regard the equalities (\ref{57f7a}) -- (\ref{57f7c}) as
the conditions of a Lagrangian $L$ to be gauge invariant. They are
brought into the form
\mar{57f8a-c}\ben
&& \dr_p^{\m\la}\cL + \dr_p^{\la\m}\cL = 0. \label{57f8a}\\
&& \dr_q^\m \cL + c_{pq}^r a_\nu^p \dr_r^{\m\nu}\cL  = 0, \label{57f8b} \\
&& c_{pq}^r(a_\m^p\dr_r^\m\cL + a^p_{\la\m} \dr_r^{\la\m}\cL)  =
0. \label{57f8c}
\een

In order to solve these equations, let us refer to the canonical
splitting of the jet manifold
\mar{296'}\ben
&& J^1C =C_+\op\oplus_C C_-=C_+\op\oplus_C (C\op\times_X\op\w^2T^*X\ot V_GP),
\nonumber\\
&&a_{\la\m}^r = \frac12(\cF_{\la\m}^r + \cS_{\la\m}^r)= \frac{1}{2}(a_{\la\m}^r + a_{\m\la}^r
 - c_{pq}^r a_\la^p a_\m^q) +
\label{296'}\\
&& \qquad   \frac{1}{2}
(a_{\la\m}^r - a_{\m\la}^r + c_{pq}^r a_\la^p a_\m^q), \nonumber
\een
and let us utilize the coordinates $(a^q_\m, \cF^r_{\la\m},
\cS^r_{\la\m})$ (\ref{296'}). With respect to these coordinates,
the equation (\ref{57f8a}) reads
\beq
\frac{\dr\cL}{\dr {\cal S}^p_{\m\la}}=0. \label{b3128}
\eeq
Then the equation (\ref{57f8b}) takes a form
\beq
\frac{\dr\cL}{\dr a^q_\m}=0. \label{b3129}
\eeq
A glance at the equalities (\ref{b3128}) and (\ref{b3129}) shows
that a gauge invariant Lagrangian factorizes through the strength
coordinates $\cF$ (\ref{296'}). Then the equation (\ref{57f8c}),
written as
\be
c^r_{pq}\cF^p_{\la\m}\frac{\dr\cL}{\dr \cF^r_{\la\m}}=0,
\ee
means that the gauge symmetry $u$ (\ref{57f30'}) of a Lagrangian
$L$ is exact. The following thus has been proved.

\begin{lemma} \mar{57t1} \label{57t1} The gauge theory Lagrangian
(\ref{57f1}) possesses the exact gauge symmetry $u$ (\ref{57f30'})
only if it factorizes through the strength coordinates $\cF$
(\ref{296'}).
\end{lemma}

A corollary of this result is the well-known Utiyama theorem
\cite{bruz}.

\begin{theorem} There is a unique gauge invariant quadratic first order
Lagrangian, called the Yang -- Mills Lagrangian,
\mar{5.1}\beq
L_{\rm YM}=\frac14a^G_{pq}g^{\la\m}g^{\bt\n}\cF^p_{\la
\beta}\cF^q_{\m\n}\sqrt{|g|}\,\om, \qquad  g=\det(g_{\m\nu}), \label{5.1}
\eeq
where $a^G$ is a $G$-invariant bilinear form on a Lie algebra
$\cG_r$ and $g$ is a world metric on $X$.
\end{theorem}

The Euler -- Lagrange operator (\ref{57f1}) of the Yang -- Mills
Lagrangian $L_{\rm YM}$ (\ref{5.1}) is
\mar{57f13}\beq
\cE_{\rm YM}=\cE^\m_r\thh^\m_r\w\om=(\dl^n_rd_\la
+c^n_{rp}a^p_\la)(a^G_{nq}g^{\m\al}g^{\la\bt}
\cF^q_{\al\bt}\sqrt{|g|})\thh_\m^r\w\om. \label{57f13}
\eeq
Its kernel (\ref{57f2'}) defines the Yang -- Mills equations
\mar{57f14}\beq
\cE^\m_r= (\dl^n_rd_\la
+c^n_{rp}a^p_\la)(a^G_{nq}g^{\m\al}g^{\la\bt}
\cF^q_{\al\bt}\sqrt{|g|})=0. \label{57f14}
\eeq

We call a Lagrangian system $(\cS^*_\infty[C], L_{\rm YM})$ the
Yang -- Mills gauge theory.

\begin{remark} In classical gauge theory, there are Lagrangians, e.g., the
Chern -- Simons one (\ref{csl}) (Section 6.3) which do not
factorize through the strength coordinates $\cF$, and whose gauge
symmetry $u$ (\ref{57f30'}) is variational, but not exact.
\end{remark}

Since the gauge symmetry $u$ (\ref{57f30'}) of the Yang -- Mills
Lagrangian (\ref{5.1}) is exact, the first variational formula
(\ref{57f5}) leads to a weak conservation law
\be
0\ap d_\la(-u^\m_r\dr^{\la\m}_r\cL_{\rm YM})
\ee
of the Noether current
\mar{57f11}\beq
\cJ^\la=-(\dr_\m\xi^r + c_{pq}^r a_\m^p\xi^q)
(a^G_{rq}g^{\m\al}g^{\la\bt} \cF^q_{\al\bt}\sqrt{|g|}).
\label{57f11}
\eeq
In accordance with Theorem \ref{supp}, the Noether current
(\ref{57f11}) is brought into the superpotential form (\ref{g21g})
which reads
\be
&& \cJ^\la= c^r\cE_r^\m + d_\nu(c^r\dr_r^{[\nu\m]}\cL_{\rm
YM}), \\
&& U^{\nu\m}= c^r a^G_{rq}g^{\nu\al}g^{\m\bt} \cF^q_{\al\bt}\sqrt{|g|}.
\ee

The gauge invariance conditions (\ref{57f7a}) -- (\ref{57f7c})
lead to the Noether identities  which the Euler -- Lagrange
operator $\cE_{\rm YM}$ (\ref{57f13}) of the Yang -- Mills
Lagrangian (\ref{5.1}) satisfies. These Noether identities are
associated to the gauge symmetry $u$ (\ref{57f30'}). In accordance
with the formula (\ref{0657}), they read
\mar{57f6}\beq
c^p_{rq}a^q_\m\cE_p^\m + d_\m\cE_r^\m=0. \label{57f6}
\eeq

\begin{lemma} \mar{57l1} \label{57l1} The Noether identities (\ref{57f6}) are
non-trivial.
\end{lemma}

\begin{proof}
Following the procedure in Section 5.2, let us consider the
density dual
\mar{vcc}\beq
\ol{VC}=V^*C\op\ot_C\op\w^nT^*X=(T^*X\op\ot_X
V_GP)^*\op\ot_C\op\w^nT^*X \label{vvc}
\eeq
of the vertical tangent bundle $VC$ of $C\to X$, and let us
enlarge the differential graded algebra $\cS^*_\infty[C]$ to the
DBGA (\ref{41f5}):
\be
\cP^*_\infty[\ol{VC};C]=\cS^*_\infty[\ol{VC};C],
\ee
possessing the local generating basis $(a^r_\m, \ol a^\m_r)$ where
$\ol a^\m_r$ are odd antifields. Providing this DBGA with the
nilpotent right graded derivation
\be
\ol\dl=\frac{\rdr}{\dr \ol a^\m_r} \cE_r^\m,
\ee
let us consider the chain complex (\ref{v042}). Its one-chains
\mar{57f21}\beq
\Delta_r=c^p_{rq}a^q_\m\ol a_p^\m + d_\m\ol a_r^\m \label{57f21}
\eeq
are $\ol\dl$-cycles which define the Noether identities
(\ref{57f6}). Clearly, they are not $\ol\dl$-boundaries.
Therefore, the Noether identities (\ref{57f6}) are non-trivial.

\end{proof}

\begin{lemma} \mar{57l2} \label{57l2} The Noether identities (\ref{57f6}) are
complete.
\end{lemma}

\begin{proof}
The second order Euler -- Lagrange operator $\cE_{\rm YM}$
(\ref{57f13}) takes its values into the space of sections of the
vector bundle
\be
(T^*X\op\ot_X V_GP)^*\op\ot_X\op\w^nT^*X\to X.
\ee
Let $\Phi$ be a first order differential operator on this vector
bundle such that
\be
\Phi\circ \cE_{\rm YM}=0.
\ee
This condition holds only if the highest derivative term of the
composition $\Phi^1\circ \cE_{\rm YM}^2$ of the first order
derivative term $\Phi^1$ of $\Phi$ and the second order derivative
term $\cE_{\rm YM}^2$ of $\cE_{\rm YM}$ vanishes. This is the case
only of
\be
\Phi^1=\Delta_r^1= d_\m\ol a_r^\m.
\ee
\end{proof}

The graded densities $\Delta_r\om$ (\ref{57f21}) constitute a
local basis for a $C^\infty(X)$-module $\cC_{(0)}$ isomorphic to a
module $\ol {V_GP}(X)$ of sections of the density dual $\ol
{V_GP}$ of the Lie algebra bundle $V_GP\to X$. Let us enlarge a
DBGA $\cP^*_\infty[\ol{VC};C]$ to a DBGA
\be
\ol\cP^*_\infty\{0\}=\cS^*_\infty[\ol{VC};C\op\times_X\ol {V_GP}]
\ee
possessing the local generating basis $(a^r_\m,\ol a^\m_r, \ol
c_r)$ where $\ol c_r$ are even Noether antifields.

\begin{lemma} \mar{57l3} \label{57l3} The Noether identities (\ref{57f6}) are
irreducible.
\end{lemma}

\begin{proof}
Providing the DBGA $\ol\cP^*_\infty\{0\}$ with the nilpotent odd
graded derivation
\be
\dl_0=\ol\dl + \frac{\rdr}{\dr \ol c_r}\Delta_r,
\ee
let us consider the chain complex (\ref{v66}). Let us assume that
$\Phi$ (\ref{41f9}) is a two-cycle of this complex, i.e., the
relation (\ref{v79}) holds. It is readily observed that $\Phi$
obeys this relation only if its first term $G$ is $\ol\dl$-exact,
i.e., the first-stage Noether identities (\ref{v79}) are trivial.
\end{proof}

It follows from Lemmas \ref{57l1} -- \ref{57l3} that Yang -- Mills
gauge theory is an irreducible degenerate Lagrangian theory
characterized by the complete Noether identities (\ref{57f6}).

Following inverse second Noether Theorem \ref{w35}, let us
consider a DBGA
\mar{vvc1}\beq
\cP^*_\infty\{0\}=\cS^*_\infty[\ol{VC}\op\oplus_C
V_GP;C\op\times_X\ol {V_GP}] \label{vvc1}
\eeq
with the local generating basis $(a^r_\m,\ol a^\m_r, c^r, \ol
c_r)$ where $c_r$ are odd ghosts. The gauge operator $\bu$
(\ref{w108'}) associated to the Noether identities (\ref{57f6})
reads
\mar{57f30}\beq
\bu=u=(c_\m^r + c_{pq}^r a_\m^pc^q)\dr_r^\m. \label{57f30}
\eeq
It is the odd gauge symmetry (\ref{57f30'}) of the Yang -- Mills
Lagrangian $L_{\rm YM}$ (\ref{5.1}). The gauge operator $\bu$
(\ref{57f30}) admits the nilpotent BRST extension (\ref{hhh}):
\be
\bbc= (c_\m^r + c_{pq}^r a_\m^pc^q)\frac{\dr}{\dr a_\m^r} -\frac12
c^r_{pq}c^pc^q\frac{\dr}{\dr c^r},
\ee
which is the well-known BRST operator in Yang - -Mills gauge
theory \cite{gom}. Then, by virtue of Theorem \ref{w130}, the Yang
-- Mills Lagrangian $L_{\rm YM}$ is extended to a proper solution
of the master equation
\be
L_E=L_{\rm YM}+ (c_\m^r + c_{pq}^r a_\m^pc^q)\ol a^\m_r\om
-\frac12 c^r_{pq}c^pc^q\ol c_r\om.
\ee

\subsection{Gauge gravitation theory on natural bundles}

Gauge transformations of Einstein's General Relativity and its
extensions, including gauge gravitation theory, are general
covariant transformations. These are bundle automorphisms of
so-called natural bundles. Therefore classical gravitation theory
can be described as a field theory on natural bundles over a
four-dimensional orientable manifold $X$, called the world
manifold \cite{book09,sard06,sard11}.

As well known, a connection $\G$ on a fibre bundle $Y\to X$
defines the horizontal lift $\G\tau$ onto $Y$ of any vector field
$\tau$ on $X$. There is the category of natural bundles
\cite{kol,terng} which admit the functorial lift $\wt\tau$ onto
$T$ of any vector field $\tau$ on $X$ such that
$\tau\mapsto\ol\tau$ is a monomorphism of the Lie algebra of
vector field on $X$ to that on $T$. One can think of the lift
$\wt\tau$ as being an infinitesimal generator of a local
one-parameter group of general covariant transformations of $T$.

Natural bundles are exemplified by tensor bundles over $X$. A
frame bundle $LX$ of linear frame in the tangent spaces to $X$ is
a natural bundle. It is a principal bundle with a structure group
$GL_4=GL^+(4, \mathbb R)$, and all bundles associated to $LX$ also
are natural are the natural ones. The bundle
\mar{gr14}\beq
C_K=J^1LX/GL_4 \label{gr14}
\eeq
of principal connections on $LX$ is not associated to $LX$, but it
also is a natural bundle \cite{book,book00}.

Dynamic variables of gauge gravitation theory are linear wold
connections and pseudo-Riemannian metrics on a world manifold.
Thus, it is a metric-affine gravitation theory
\cite{hehl,iva,sard11}.

Linear connections on $X$ (henceforth world connection) are
principal connections on the linear frame bundle $LX$  of $X$.
They are represented by sections of the bundle of linear
connections $C_K$ (\ref{gr14}). This is provided with bundle
coordinates $(x^\la,k_\la{}^\nu{}_\al)$ such that components
$k_\la{}^\nu{}_\al\circ K=K_\la{}^\nu{}_\al$ of a section $K$ of
$C_K\to X$ are coefficient of the linear connection
\be
K=dx^\la\ot (\dr_\la + K_\la{}^\m{}_\nu \dot x^\nu\dot\dr_\mu)
\ee
on $TX$ with respect to the holonomic bundle coordinates
$(x^\la,\dot x^\la)$.

In order to describe gravity, let us assume that the linear frame
bundle $LX$ admits a Lorentz structure, i.e., reduced principal
subbundles with the structure Lorentz group $SO(1,3)$. Global
sections of the corresponding quotient bundle
\mar{b3203}\beq
\Si= LX/SO(1,3)\to X \label{b3203}
\eeq
are pseudo-Riemannian (henceforth world) metrics on $X$. This fact
motivates us to treat a metric gravitational field as a Higgs
field \cite{iva,sard06,sard11}.

Thus, the total configuration space of gauge gravitation theory in
the absence of matter fields is the bundle product
\mar{grr}\beq
Q=\Si\op\times_X C_K \label{grr}
\eeq
coordinated by $(x^\la,\si^{\al\bt}, k_\mu{}^\al{}_\bt)$.

We consider first order Lagrangian theory on the fibre bundle $Q$
(\ref{grr})(see Example \ref{first}). Its structure algebra
$\cS^*_\infty[F;Y]=\cS^*_\infty[Q]$ (\ref{5.77a}) is the graded
differential algebra $\cS^*_\infty[C]=\cO^*_\infty Q$ (\ref{ppp})
of exterior forms on jet manifolds $J^rQ$ of $Q\to X$. Its first
order Lagrangian (\ref{23f2}) is a density
\mar{grav}\beq
L_G=\cL_G\om: J^1Q\to \op\w^n T^*X \label{grav}
\eeq
on a first order jet manifold $J^1Q$ possessing the adapted
coordinates
\be
(x^\la,\si^{\al\bt},
k_\mu{}^\al{}_\bt,\si^{\al\bt}_\la,k_{\la\mu}{}^\al{}_\bt ).
\ee
The corresponding Euler -- Lagrange operator (\ref{305}) reads
\mar{999}\beq
\cE_G=(\cE_{\al\bt} d\si^{\al\bt} + \cE^\m{}_\al{}^\bt
dk_\m{}^\al{}_\bt)\w\om. \label{999}
\eeq
Its kernel defines the Euler -- Lagrange equations
\mar{998}\beq
\cE_{\al\bt}=0, \qquad \cE^\m{}_\al{}^\bt =0. \label{998}
\eeq

The fibre bundle $Q$ (\ref{grr}) is a natural bundle admitting the
functorial lift
\mar{gr3}\ben
&& \wt\tau_{K\Si}=\tau^\m\dr_\m +(\si^{\nu\bt}\dr_\nu \tau^\al
+\si^{\al\nu}\dr_\nu \tau^\bt)\frac{\dr}{\dr \si^{\al\bt}} +
\label{gr3}\\
&& \qquad (\dr_\nu \tau^\al k_\m{}^\nu{}_\bt -\dr_\bt \tau^\nu
k_\m{}^\al{}_\nu -\dr_\mu \tau^\nu k_\nu{}^\al{}_\bt
+\dr_{\m\bt}\tau^\al)\frac{\dr}{\dr k_\mu{}^\al{}_\bt} \nonumber
\een
of vector fields $\tau$ on $X$ \cite{book09,book00,sard11}. These
lifts are generators of one-dimensional groups of general
covariant transformations.

A glance at the expression (\ref{gr3}) shows that one can think of
the vector fields $\wt\tau_{K\Si}$ as being a linear first order
differential operator on a vector space of vector fields on $X$
with values in a vector space of vector fields on the fibre bundle
$Q$ (\ref{grr}), i.e.,  $\wt\tau_{K\Si}$ (\ref{gr3}) are even
gauge transformations (see Remark \ref{qwe}) with even gauge
parameter functions $\tau$. By the same reasons as in Yang --
Mills gauge theory, we however modify the definition of gauge
transformations in gauge gravitation theory in accordance with
Definition \ref{sgauge} as follows.

Let us treat the tangent bundle $TX\to X$ as an odd vector bundle,
and let $(X,\gA_{TX}$ be the corresponding simple graded manifold.
Then let us consider the composite bundle
\mar{vgg}\beq
TX\op\times_X Q\to Q\to X, \label{vgg}
\eeq
coordinated by $(x^\la,\si^{\al\bt}, k_\mu{}^\al{}_\bt,\dot
x^\nu)$, and the graded bundle $(X,Q,\gA_{TX\times_X Q})$
(\ref{olo}) modelled over this composite bundle together with the
local generating basis
\be
(x^\la,\si^{\al\bt}, k_\mu{}^\al{}_\bt,c^\nu)
\ee
whose terms $c^\nu$ are odd. Let
\mar{yyg}\beq
S_\infty^*[TX\op\times_X Q; Q] \label{yyg}
\eeq
 be the DBGA (\ref{tyt'}) together with the monomorphisms (\ref{tgv}):
\be
\cO^*_\infty Q\to S_\infty^*[TX\op\times_X Q; Q], \qquad
S^*_\infty[TX;X]\to S_\infty^*[TX\op\times_X Q; Q].
\ee

By inspection of transition functions of the principal vector
field $\wt\tau_{K\Si}$ (\ref{gr3}), one can justify the existence
of an odd contact derivation $J^\infty u_G$ of the DBGA
(\ref{yyg}) generated by a generalized vector field
\mar{gr3'}\ben
&& u_G=c^\m\dr_\m +(\si^{\nu\bt}c_\nu^\al
+\si^{\al\nu}c_\nu^\bt)\frac{\dr}{\dr \si^{\al\bt}} +
\label{gr3'}\\
&& \qquad (c_\nu^\al k_\m{}^\nu{}_\bt -c_\bt^\nu
k_\m{}^\al{}_\nu -c_\mu^\nu k_\nu{}^\al{}_\bt
+c_{\m\bt}^\al)\frac{\dr}{\dr k_\mu{}^\al{}_\bt} \nonumber
\een
on a graded bundle $(X,Q,\gA_{TX\times_X Q})$. The graded
derivation $u$ (\ref{gr3'}) obviously vanishes on a subring
\be
S^0_\infty[TX;X]\subset S_\infty^0[TX\op\times_X Q; Q].
\ee
Consequently, it is a gauge transformation of a Lagrangian system
\be
\cO^*_\infty Q\subset S_\infty^*[TX\op\times_X Q; Q]
\ee
parameterized by odd ghosts $c^\m$.

We do not specify a gravitation Lagrangian $L_G$ on the jet
manifold $J^1Q$, but assume that the generalized vector field
(\ref{gr3'}) is its exact gauge symmetry. Then the Euler --
Lagrange operator (\ref{999}) of this Lagrangian obeys irreducible
Noether identities
\be
&&-(\si^{\al\bt}_\la +2\si^{\nu\bt}_\nu\dl^\al_\la)\cE_{\al\bt}
-2\si^{\nu\bt}d_\nu\cE_{\la\bt} +(-k_{\la\m}{}^\al{}_\bt
-k_{\nu\m}{}^\nu{}_\bt\dl^\al_\la + k_{\bt\m}{}^\al{}_\la +
k_{\m\la}{}^\al{}_\bt)\cE^\m{}_\al{}^\bt +\\
&& \qquad (-k_\m{}^\nu{}_\bt\dl^\al_\la
+k_\m{}^\al{}_\la\dl^\nu_\bt
+k_\la{}^\al{}_\bt\dl^\nu_\m)d_\nu\cE^\m{}_\al{}^\bt + d_{\m\bt}
\cE^\m{}_\la{}^\bt=0
\ee
\cite{ijgmmp05,book09}.

\begin{remark} \label{httu1} \mar{httu1} By analogy with Theorem \ref{57t1},
one can show that, if the first order Lagrangian $L_G$
(\ref{grav}) does not depend on the jet coordinates
$\si^{\al\bt}_\la$ and it possesses the exact gauge symmetry
(\ref{gr3'}), it factorizes through the curvature terms
\mar{0101}\beq
\cR_{\la\m}{}^\al{}_\bt = k_{\la\m}{}^\al{}_\bt -
k_{\m\la}{}^\al{}_\bt + k_\la{}^\g{}_\bt k_\m{}^\al{}_\g
-k_\m{}^\g{}_\bt k_\la{}^\al{}_\g. \label{0101}
\eeq
\end{remark}

Taking the vertical part of the generalized vector field $u_G$
(\ref{gr3'}), we obtain the gauge operator $\bu=u_G$ (\ref{w108'})
and its nilpotent BRST prolongation (\ref{hhh}):
\be
&&\bbc=u^{\al\bt}\frac{\dr}{\dr\si^{\al\bt}} +u_\m{}^\al{}_\bt
\frac{\dr}{\dr k_\mu{}^\al{}_\bt} +u^\la \frac{\dr}{\dr
c^\la}=(\si^{\nu\bt} c_\nu^\al +\si^{\al\nu}
c_\nu^\bt-c^\la\si_\la^{\al\bt})\frac{\dr}{\dr \si^{\al\bt}}+
\\
&& \qquad (c_\nu^\al k_\m{}^\nu{}_\bt -c_\bt^\nu k_\m{}^\al{}_\nu
-c_\mu^\nu k_\nu{}^\al{}_\bt +c_{\m\bt}^\al-c^\la
k_{\la\mu}{}^\al{}_\bt)\frac{\dr}{\dr k_\mu{}^\al{}_\bt} +
c^\la_\m c^\m\frac{\dr}{\dr c^\la},
\ee
but this differs from that in \cite{gron}. Accordingly, an
original Lagrangian $L_G$ is extended to a solution of the master
equation
\be
L_E= L_G + u^{\al\bt}\ol\si_{\al\bt}\om + u_\m{}^\al{}_\bt \ol
k^\m{}_\al{}^\bt\om + u^\la \ol c_\la\om,
\ee
where $\ol\si_{\al\bt}$, $\ol k^\m{}_\al{}^\bt$ and $\ol c_\la$
are the corresponding antifields.

\begin{remark} \label{hhtu} \mar{hhtu}
The  Hilbert -- Einstein Lagrangian $L_{\rm HE}$ of General
Relativity depends only on metric variables $\si^{\al\bt}$. It is
a reduced second order Lagrangian which differs from the first
order one $L'_{\rm HE}$ in a variationally trivial term. The gauge
transformations $u_G$ (\ref{gr3'}) is a variational (but not
exact) symmetry of the first order Lagrangian $L'_{\rm HE}$, and
its vertical part
\be
u_V=(\si^{\nu\bt} c_\nu^\al +\si^{\al\nu}
c_\nu^\bt-c^\la\si_\la^{\al\bt})\frac{\dr}{\dr \si^{\al\bt}}
\ee
is so. Then the corresponding Noether identities (\ref{0657}) take
the familiar form
\be
\nabla_\m \cE^\m_\la= (d_\m + \{_\m{}^\bt{}_\la\})\cE^\m_\bt =0,
\ee
where $\cE^\m_\la= \si^{\m\al}\cE_{\al\la}$ and
\mar{07103}\beq
\{_\m{}^\bt{}_\la\}= -\frac12\si^{\bt\nu}(d_\m\si_{\nu\la} +
d_\la\si_{\m\nu} - d_\nu\si_{\m\la}) \label{07103}
\eeq
are the Christoffel symbols expressed into function $\si_{\al\bt}$
of $\si^{\m\nu}$ given by the relations
$\si^{\m\al}\si_{\al\bt}=\dl^\m_\bt$.
\end{remark}

Since the gauge symmetry $u_G$ (\ref{gr3'}) is assumed to be an
exact symmetries of a metric-affine gravitation Lagrangian, let us
study the corresponding conservation law. This is the
energy-momentum conservation laws because the gauge symmetry $u_G$
is not vertical, and the corresponding energy-momentum current
reduces to a superpotential in accordance with Theorem \ref{supp})
\cite{giacqg,sard97,sard11}.

In view of Remark \ref{httu1}, let us assume that a gravitation
Lagrangian $L_G$ is independent of the jet variables
$\si_\la{}^{\al\bt}$ of a world metric and that it factorizes
through the curvature terms $\cR_{\la\m}{}^\al{}_\bt$
(\ref{0101}). Then the following relations take place:
\mar{K300',}\ben
&&  \pi^{\la\nu}{}_\al{}^\bt= -\pi^{\nu\la}{}_\al{}^\bt, \qquad
\pi^{\la\nu}{}_\al{}^\bt=\frac{\dr \cL_G}{\dr k_{\la\nu}{}^\al{}_\bt}, \label{K300'}\\
&&\frac{\dr\cL_G}{\dr k_\nu{}^\al{}_\bt}=
\pi^{\la\nu}{}_\al{}^\si k_\la{}^\bt{}_\si
-\pi^{\la\nu}{}_\si{}^\bt k_\la{}^\si{}_\al. \label{K300}
\een

Let us follow the compact notation
\be
&& y^A=k_\m{}^\al{}_\bt, \\
&& u_\m{}^\al{}_\bt{}^{\ve\si}_\g = \dl^\ve_\m \dl^\si_\bt \dl^\al_\g, \\
&&  u_\m{}^\al{}_\bt{}^\ve_\g= k_\m{}^\ve{}_\bt \dl^\al_\g -k_\m{}^\al{}_\g
\dl^\ve_\bt - k_\g{}^\al{}_\bt \dl^\ve_\m.
\ee
Then the generalized vector field (\ref{gr3'}) takes a form
\be
u_G =c^\la\dr_\la  + (\si^{\nu\bt}c_\nu^\al
+\si^{\al\nu}c_\nu^\bt)\dr_{\al\bt}+
  (u^A{}_\al^\bt c_\bt^\al
+u^A{}_\al^{\bt\m}c_{\bt\m}^\al)\dr_A.
\ee
We also have the equalities
\be
&& \pi^\la_A u^A{}_\al^{\bt\m} =\pi^{\la\m}{}_\al{}^\bt,\\
&& \pi^\ve_A u^A{}_\al^\bt = -\dr^\ve{}_\al{}^\bt\cL_G -
\pi^{\ve\bt}{}_\si{}^\g k_\al{}^\si{}_\g.
\ee

Let a Lagrangian  $L_G$ be invariant under general covariant
transformations, i.e.,
\be
\bL_{J^1u_G}L_G=0.
\ee
Then the first variational formula (\ref{J4}) takes a form
\mar{J4a}\ben
&& 0= (\si^{\nu\bt}c_\nu^\al +\si^{\al\nu}c_\nu^\bt -c^\la\si^{\al\bt}_\la)
\dl_{\al\bt}\cL_G +
\label{J4a}\\
&& \qquad (u^A{}_\al^\bt c_\bt^\al
+u^A{}_\al^{\bt\m}c_{\bt\m}^\al - c^\la y^A_\la) \dl_A
\cL_G -\nonumber\\
&& \qquad d_\la[ \pi^\la_A(y^A_\al c^\al -u^A{}_\al^\bt c_\bt^\al
 -u^A{}_\al^{\ve\bt}c_{\ve\bt}^\al) -c^\la\cL_G].
 \nonumber
\een
The first variational formula (\ref{J4a}) on-shell leads to the
weak conservation law
\mar{K8}\ben
&& 0\ap - d_\la[ \pi^\la_A(y^A_\al c^\al -u^A{}_\al^\bt c_\bt^\al
 -u^A{}_\al^{\ve\bt}c_{\ve\bt}^\al) -c^\la\cL_G],\label{K8}
\een
where
\mar{b3190}\beq
\cJ^\la= \pi^\la_A(y^A_\al c^\al -u^A{}_\al^\bt c_\bt^\al
 -u^A{}_\al^{\ve\bt} c_{\ve\bt}^\al)-c^\la\cL_G \label{b3190}
\eeq
is the  energy-momentum current of the metric-affine gravity.

Due to the arbitrariness of gauge parameters $c^\la$, the first
variational formula (\ref{J4a}) falls into the set of equalities
(\ref{g4g}) -- (\ref{g7g}) which read
\mar{b3173d,-b}\ben
&& \pi^{(\la\ve}{}_\g{}^{\si)}=0, \label{b3173d}\\
&& (u^A{}_\g^{\ve\si}\dr_A + u^A{}_\g^\ve\dr^\si_A)\cL_G= 0, \label{b3173c}\\
&& \dl^\bt_\al\cL_G + 2\si^{\bt\m}\dl_{\al\m}\cL_G +
u^A{}_\al^\bt\dl_A\cL_G
 + d_\m(\pi^\m_A  u^A{}_\al^\bt)
-y^A_\al\pi^\bt_A  =0 \label{b3173b} \\
&& \dr_\la\cL_G=0. \nonumber
\een
It is readily observed that the equalities (\ref{b3173d}) and
(\ref{b3173c}) hold due to the relations (\ref{K300'}) and
(\ref{K300}), respectively.

Substituting the term $y^A_\al\pi^\bt_A$ from the expression
(\ref{b3173b}) in the energy-momentum conservation law (\ref{K8}),
one brings this conservation law into the form
\mar{b3174}\ben
&& 0\ap -
d_\la[2\si^{\la\m}c^\al\dl_{\al\m}\cL_G +
u^A{}_\al^\la c^\al\dl_A\cL_G - \pi^\la_Au^A{}_\al^\bt c_\bt^\al + \label{b3174}\\
&& \qquad d_\m(\pi^{\la\m}{}_\al{}^\bt)
c_\bt^\al + d_\m(\pi^\m_A  u^A{}_\al^\la)c^\al -
d_\m(\pi^{\la\m}{}_\al{}^\bt \dr_\bt c^\al)]. \nonumber
\een
After separating the variational derivatives, the energy-momentum
conservation law (\ref{b3174}) of the metric-affine gravity takes
the superpotential form
\be
&& 0\ap - d_\la [2\si^{\la\m}c^\al\dl_{\al\m}\cL_G
+\\
&& \qquad (k_\m{}^\la{}_\g\dl^\m{}_\al{}^\g\cL_G -
 k_\m{}^\si{}_\al\dl^\m{}_\si{}^\la\cL_G -
k_\al{}^\si{}_\g\dl^\la{}_\si{}^\g\cL_G)c^\al +  \\
&& \qquad \dl^\la{}_\al{}^\m\cL_G c_\m^\al
-d_\m(\dl^\m{}_\al{}^\la\cL_G)C^\al +
 d_\m(\pi^{\m\la}{}_\al{}^\nu(c_\nu^\al
-k_\si{}^\al{}_\nu c^\si))],
\ee
where the energy-momentum current on the shell (\ref{998}) reduces
to the generalized Komar superpotential
\mar{K3}\beq
U_G{}^{\m\la}= \pi^{\m\la}{}_\al{}^\nu(c_\nu^\al
-k_\si{}^\al{}_\nu c^\si) \label{K3}
\eeq
\cite{giacqg,sard97}. We can rewrite this superpotential as
\be
U_G{}^{\m\la}= 2\frac{\dr\cL_G}{\dr \cR_{\m\la}{}^\al{}_\nu}(D_\nu
c^\al + T_\nu{}^\al{}_\si c^\si),
\ee
where $D_\nu$ is the covariant derivative relative to a connection
$k_\nu{}^\al{}_\si$ and
\be
T_\nu{}^\al{}_\si= k_\nu{}^\al{}_\si - k_\si{}^\al{}_\nu
\ee
is its torsion.

\begin{example}
Let us consider a Hilbert -- Einstein Lagrangian
\be
&& L_{\rm HE}=\frac{1}{2\kp}\cR\sqrt{-\si}\om,\\
&& \cR=\si^{\la\nu}\cR_{\al\la}{}^\al{}_\nu, \qquad
\si=\det(\si_{\al\bt}),
\ee
in a metric-affine gravitation model. Then the generalized Komar
superpotential (\ref{K3}) comes to the well-known Komar
superpotential if we substitute the Levi -- Civita connection
$k_\nu{}^\al{}_\si =\{_\nu{}^\al{}_\si\}$ (\ref{07103}).
\end{example}

\subsection{Chern -- Simons topological theory}

We consider gauge theory of principal connections on a principal
bundle $P\to X$ with a structure real Lie group $G$. In contrast
with the Yang -- Mills Lagrangian $L_{\rm YM}$ (\ref{5.1}), the
Lagrangian $L_{\rm CS}$ (\ref{csl}) of Chern -- Simons topological
field theory on an odd-dimensional manifold $X$ is independent of
a world metric on $X$. Therefore, its non-trivial gauge symmetries
are wider than those of the Yang -- Mills one. However, some of
them become trivial if $\di X=3$.

Note that one usually considers a local Chern -- Simons Lagrangian
which is the local Chern -- Simons form derived from the local
transgression formula for the Chern characteristic form. A global
Chern -- Simons Lagrangian is well defined, but depends on a
background gauge potential \cite{bor07,mpl,cs}.

Let $P\to X$ be a principal bundle with a structure Lie group $G$
and $C$ the bundle of principal connections (\ref{B1}) coordinated
by $(x^\la, c^r_\m)$ (see Section 6.1). One can show \cite{book00}
that the quotient bundle $J^1P\to C$ is a principal bundle with a
structure group $G$ which is canonically isomorphic to the
pull-back
\mar{b1.251}\beq
J^1P= P_C=C\op\times_X P\to C. \label{b1.251}
\eeq
This bundle admits the canonical principal connection
\mar{266}\beq
\cA =dx^\la\ot(\dr_\la +a_\la^p e_p) + da^r_\la\ot\dr^\la_r,
\label{266}
\eeq
with the strength
\mar{267}\beq
F_\cA  = (da_\m^r\w dx^\m + \frac{1}{2} c_{pq}^r a_\la^p a_\m^q
dx^\la\w dx^\m)\ot e_r. \label{267}
\eeq

Let
\mar{80f1}\beq
I_k(\chi)=b_{r_1\ldots r_k}\chi^{r_1}\cdots \chi^{r_k}
\label{80f1}
\eeq
be a $G$-invariant polynomial of degree $k>1$ on a Lie algebra
$\cG_r$ of $G$.  With the strength $F_\cA$ (\ref{267}) of the
canonical principal connection $\cA$ (\ref{266}), one can
associate to this polynomial $I_k$ a closed $2k$-form
\mar{0757}\beq
P_{2k}(F_\cA)=b_{r_1\ldots r_k}F_\cA^{r_1}\w\cdots\w F_\cA^{r_k},
\qquad 2k\leq n, \label{0757}
\eeq
on a bundle of principal connections $C$ which is invariant under
automorphisms of $C$ induced by vertical principal automorphisms
of $P$.  Given a section $A$ of $C\to X$, the pull-back
\mar{mos11}\beq
P_{2k}(F_A)=A^*P_{2k}(F_\cA) \label{mos11}
\eeq
of \index{$P_{2k}(F_A)$} the form $P_{2k}(F_\cA)$ (\ref{0757}) is
a closed $2k$-form on $X$ where
\mar{1136b}\ben
&& F_A
=\frac12 F^r_{\la\m} dx^\la\w dx^\m\ot e_r, \nonumber \\
&& F_{\la\m}^r = [\dr_\la +A^p_\la e_p, \dr_\m +A^q_\m
e_q]^r= \dr_\la A_\m^r - \dr_\m A_\la^r + c_{pq}^rA_\la^p A_\m^q,
\label{1136b}
\een
is a strength of a principal connection $A$. One calls the
$P_{2k}(F_A)$ (\ref{mos11}) the characteristic form because of its
following properties \cite{egu,book00}.

$\bullet$ Every characteristic form $P_{2k}(F_A)$ (\ref{mos11}) is
a closed form, i.e., $dP_{2k}(F_A)=0$;

$\bullet$ The difference $P_{2k}(F_A)-P_{2k}(F_{A'})$ of
characteristic forms is an exact form, whenever $A$ and $A'$ are
different principal connections on a principal bundle $P$.

It follows that characteristic forms $P_{2k}(F_A)$ possesses the
same de Rham cohomology class $[P_{2k}(F_A)]$ for all principal
connections $A$ on $P$. The association
\be
I_k(\chi)\to [P_{2k}(F_A)]\in H^*_{\rm DR}(X)
\ee
is the well-known Weil homomorphism.

Let $I_k$ (\ref{80f1}) be a $G$-invariant polynomial of degree
$k>1$ on the Lie algebra $\cG_r$ of $G$. Let $P_{2k}(F_\cA)$
(\ref{0757}) be the corresponding closed $2k$-form on $C$ and
$P_{2k}(F_A)$ (\ref{mos11}) its pullback onto $X$ by means of a
section $A$ of $C\to X$. Let the same symbol $P_{2k}(F_A)$ stand
for its pull-back onto $C$. Since $C\to X$ is an affine bundle
and, consequently, the de Rham cohomology of $C$ equals that of
$X$, the exterior forms $P_{2k}(F_\cA)$ and $P_{2k}(F_A)$ possess
the same de Rham cohomology class
\be
[P_{2k}(F_\cA)]=[P_{2k}(F_A)]
\ee
for any principal connection $A$. Consequently,  the exterior
forms $P_{2k}(F_\cA)$ and $P_{2k}(F_A)$ on $C$ differ from each
other in an exact form
\mar{r65}\beq
P_{2k}(F_\cA)-P_{2k}(F_A)=d\gS_{2k-1}(a,A).\label{r65}
\eeq
This relation is called the transgression formula on $C$
\cite{book09}. Its pull-back by means of a section $B$ of $C\to X$
gives the transgression formula on a base $X$:
\be
P_{2k}(F_B)-P_{2k}(F_A)=d \gS_{2k-1}(B,A).
\ee

For instance, let
\be
c(F_\cA)={\rm
det}\left(\bb+\frac{i}{2\pi}F_\cA\right)=1+c_1(F_\cA)+c_2(F_\cA)+\cdots
\ee
be the total Chern form on a bundle of principal connections $C$.
Its components $c_k(F_\cA)$ are Chern characteristic forms on $C$.
If
\be
P_{2k}(F_\cA)=c_k(F_\cA)
\ee
is the characteristic Chern $2k$-form, then $\gS_{2k-1}(a,A)$
(\ref{r65}) is the Chern -- Simons $(2k-1)$-form.

In particular, one can choose a local section $A=0$. In this case,
$\gS_{2k-1}(a,0)$ is called the local Chern -- Simons form. Let
$\gS_{2k-1}(A,0)$ be its pull-back onto $X$ by means of a section
$A$ of $C\to X$. Then the Chern -- Simons form $\gS_{2k-1}(a,A)$
(\ref{r65}) admits the decomposition
\mar{r75}\beq
\gS_{2k-1}(a,A)=\gS_{2k-1}(a,0) -\gS_{2k-1}(A,0) +dK_{2k-1}.
\label{r75}
\eeq
The transgression formula (\ref{r65}) also yields the
transgression formula
\mar{0742}\ben
&& h_0(P_{2k}(F_\cA)-P_{2k}(F_A))=d_H(h_0 \gS_{2k-1}(a,A)), \nonumber\\
&& h_0 \gS_{2k-1}(a,A)=k\op\int^1_0 \cP_{2k}(t,A)dt, \label{0742}\\
&& \cP_{2k}(t,A)=b_{r_1\ldots
r_k}(a^{r_1}_{\m_1}-A^{r_1}_{\m_1})dx^{\m_1}\w
\cF^{r_2}(t,A)\w\cdots \w \cF^{r_k}(t,A),\nonumber\\
&& \qquad \cF^{r_j}(t,A)= \frac12[ ta^{r_j}_{\la_j\m_j}
+  (1-t)\dr_{\la_j}A^{r_j}_{\m_j} - ta^{r_j}_{\m_j\la_j}
-\nonumber \\
&&\qquad (1-t)\dr_{\m_j}A^{r_j}_{\la_j}+ \frac12c^{r_j}_{pq} (ta^p_{\la_j}
+(1-t)A^p_{\la_j})(ta^q_{\m_j}
+\nonumber\\
&& \qquad(1-t)A^q_{\m_j}]dx^{\la_j}\w dx^{\m_j}\ot e_r,\nonumber
\een
on $J^1C$ (where $b_{r_1\ldots r_k}$ are coefficients of the
invariant polynomial (\ref{80f1})).

If $2k-1=\di X$, the density (\ref{0742})  is the global Chern --
Simons Lagrangian
\mar{csl}\beq
L_{\rm CS}(A)=h_0\gS_{2k-1}(a,A) \label{csl}
\eeq
of Chern -- Simons topological field theory. It depends on a
background gauge field $A$. The decomposition (\ref{r75}) induces
the decomposition
\mar{0747}\beq
L_{\rm CS}(A)=h_0\gS_{2k-1}(a,0) -h_0\gS_{2k-1}(A,0) +d_H
h_0K_{2k-1}, \label{0747}
\eeq
where
\mar{csl2}\beq
L_{\rm CS}=h_0\gS_{2k-1}(a,0) \label{csl2}
\eeq
is a local Chern -- Simons Lagrangian.

For instance, if $\di X=3$, the global Chern -- Simons Lagrangian
(\ref{csl}) reads
\mar{s20}\ben
&& L_{\rm CS}(A)= \left[\frac12 h_{mn} \ve^{\al\bt\g}a^m_\al(\cF^n_{\bt\g}
-\frac13 c^n_{pq}a^p_\bt a^q_\g)\right]\om - \label{s20} \\
&& \qquad \left[\frac12 h_{mn} \ve^{\al\bt\g}A^m_\al(F_A{}^n_{\bt\g}
-\frac13 c^n_{pq}A^p_\bt A^q_\g)\right]\om -\nonumber\\
&& \qquad d_\al(h_{mn} \ve^{\al\bt\g}a^m_\bt A^n_\g)\om, \nonumber
\een
where $\ve^{\al\bt\g}$ is the skew-symmetric Levi -- Civita
tensor.

Since a density
\be
-\gS_{2k-1}(A,0) +d_Hh_0K_{2k-1}
\ee
is variationally trivial, the global Chern -- Simons Lagrangian
(\ref{csl}) possesses the same Noether identities and gauge
symmetries as the local one (\ref{csl2}). They are the following.

In contrast with a Yang -- Mills Lagrangian, the Chern -- Simons
one $L_{CS}(B)$ is independent of a world metric on $X$.
Therefore, its gauge symmetries are all $G$-invariant vector
fields on a principal bundle $P$. They are identified with
sections
\mar{0745}\beq
\xi=\tau^\la\dr_\la +\xi^r e_r, \label{0745}
\eeq
of the vector bundle
\mar{fgj}\beq
T_GP=TP/G\to X, \label{fgj}
\eeq
and yield the vector fields
\mar{0653}\beq
\up=\tau^\la\dr_\la +(-c^r_{pq}\xi^pa^q_\la +\dr_\la \xi^r
-a^r_\m\dr_\la \tau^\m)\dr^\la_r \label{0653}
\eeq
on a bundle of principal connections $C$.  Sections $\xi$
(\ref{0745}) play a role of gauge parameters.

\begin{lemma}
Vector fields (\ref{0653}) are locally variational symmetries of
the global Chern -- Simons  Lagrangian $L_{\rm CS}(A)$
(\ref{csl}).
\end{lemma}

\begin{proof} Since $\di X=2k-1$, the transgression formula
(\ref{r65}) takes a form
\be
P_{2k}(F_\cA)=d \gS_{2k-1}(a,A).
\ee
The Lie derivative $\bL_{J^1\up}$ acting on its sides results in
the equality
\be
0=d(\up\rfloor d \gS_{2k-1}(a,A))=d(\bL_{J^1\up}\gS_{2k-1}(a,A)),
\ee
i.e., the Lie derivative $\bL_{J^1\up}\gS_{2k-1}(a,A)$ is locally
$d$-exact. Consequently, the horizontal form
$h_0\bL_{J^1\up}\gS_{2k-1}(a,A)$ is locally $d_H$-exact. A direct
computation shows that
\be
h_0\bL_{J^1\up}\gS_{2k-1}(a,A)= \bL_{J^1\up}(h_0\gS_{2k-1}(a,A))
+d_H S.
\ee
It follows that the Lie derivative $\bL_{J^1\up} L_{\rm CS}(A)$ of
the global Chern -- Simons Lagrangian along any vector field $\up$
(\ref{0653}) is locally $d_H$-exact, i.e., this vector field is
locally a variational symmetry of $L_{\rm CS}(A)$.
\end{proof}

By virtue of item (iii) of Lemma \ref{35l10}, a vertical part
\mar{0785}\beq
\up_V=(-c^r_{pq}\xi^pa^q_\la +\dr_\la \xi^r -a^r_\m\dr_\la \tau^\m
-\tau^\m a^r_{\m\la} )\dr^\la_r \label{0785}
\eeq
of the vector field $\up$ (\ref{0653}) also is locally a
variational symmetry of $L_{\rm CS}(A)$.

Given the fibre bundle $T_GP\to X$ (\ref{fgj}), let the same
symbol also stand for the pull-back of $T_GP$ onto $C$. Let us
consider the DBGA (\ref{41f5}):
\be
\cP^*_\infty[T_GP;C]=\cS^*_\infty[T_GP;C],
\ee
possessing the local generating basis $(a^r_\la, c^\la, c^r)$ of
even fields $a^r_\la$ and odd ghosts $c^\la$, $c^r$. Substituting
these ghosts for gauge parameters in the vector field $\up$
(\ref{0785}), we obtain the odd vertical graded derivation
\mar{0781}\beq
u=(-c^r_{pq}c^pa^q_\la + c^r_\la -c^\m_\la a^r_\m -c^\m
a_{\m\la}^r)\dr^\la_r \label{0781}
\eeq
of the DBGA $\cP^*_\infty[T_GP;C]$. This graded derivation as like
as vector fields $\up_V$ (\ref{0785}) is locally a variational
symmetry of the global Chern -- Simons Lagrangian $L_{\rm CS}(A)$
(\ref{csl}), i.e., the odd density $\bL_{J^1u}(L_{\rm CS}(A))$ is
locally $d_H$-exact. Hence, it is $\dl$-closed and, consequently,
$d_H$-exact in accordance with Corollary \ref{34c5}. Thus, the
graded derivation $u$ (\ref{0781}) is a variational symmetry and,
consequently, a gauge symmetry of the global Chern -- Simons
Lagrangian $L_{\rm CS}(A)$.

By virtue of the formulas (\ref{0656}) -- (\ref{0657}), the
corresponding Noether identities read
\mar{s15,'}\ben
&& \ol\dl\Delta_j= -c^r_{ji}a^i_\la\cE_r^\la -
 d_\la\cE_j^\la=0,\label{s15}\\
&& \ol\dl\Delta_\m=-
 a^r_{\m\la}\cE^\la_r +d_\la(a^r_\m\cE^\la_r)=0. \label{s15'}
\een
They are irreducible and non-trivial, unless $\di X=3$. Therefore,
the gauge operator (\ref{w108'}) is $\bu=u$. It admits the
nilpotent BRST extension (\ref{hhh}) which takes a form
\mar{ggt}\beq
\bbc=(-c^r_{ji}c^ja^i_\la + c^r_\la -c^\m_\la a^r_\m -c^\m
a_{\m\la}^r)\frac{\dr}{\dr a_\la^r} - \frac12
c^r_{ij}c^ic^j\frac{\dr}{\dr c^r} +c^\la_\m c^\m\frac{\dr}{\dr
c^\la}.\label{ggt}
\eeq

In order to include antifields $(\ol a^\la_r, \ol c_r, \ol c_\m)$,
let us enlarge the DBGA $\cP^*_\infty[T_GP;C]$ to the DBGA
\be
\cP^*_\infty\{0\}=\cS^*_\infty[\ol{VC}\op\oplus_C
T_GP;C\op\times_X \ol{T_GP}]
\ee
where $\ol{VC}$ is the density dual (\ref{vvc}) of the vertical
tangent bundle $VC$ of $C\to X$ and $\ol{T_GP}$ is the density
dual of $T_GP\to X$ (cf. (\ref{vvc1})). By virtue of Theorem
\ref{w130}, given the BRST operator $\bbc$ (\ref{ggt}), the global
Chern -- Simons Lagrangian $L_{\rm CS}(A)$ (\ref{csl}) is extended
to the proper solution (\ref{w133}) of the master equation which
reads
\be
 L_E=L_{\rm CS}(A)+ (-c^r_{pq}c^pa^q_\la + c^r_\la -c^\m_\la
a^r_\m -c^\m a_{\m\la}^r)\ol a^\la_r\om - \frac12
c^r_{ij}c^ic^j\ol c_r\om + c^\la_\m c^\m\ol c_\la\om.
\ee

If $\di X=3$, the global Chern -- Simons Lagrangian takes the form
(\ref{s20}). Its Euler -- Lagrange operator is
\be
\dl L_{\rm CS}(B)=\cE^\la_r \thh^r_\la\w \om, \qquad
\cE^\la_r=h_{rp} \ve^{\la\bt\g}\cF^p_{\bt\g}.
\ee
A glance at the Noether identities (\ref{s15}) -- (\ref{s15'})
shows that they are equivalent to the Noether identities
\mar{s16,a}\ben
&& \ol\dl\Delta_j=-c^r_{ji}a^i_\la\cE_r^\la -
d_\la\cE_j^\la=0,\label{s16}\\
&& \ol\dl\Delta'_\m= \ol\dl\Delta_\m
+a^r_\m\ol\dl\Delta_r=c^\m\cF^r_{\la\m}\cE^\la_r=0. \label{s16a}
\een
These Noether identities define the gauge symmetry $u$
(\ref{0781}) written in the form
\mar{s18}\beq
u=(-c^r_{pq}c'^pa^q_\la + c'^r_\la +c^\m\cF^r_{\la\m})\dr^\la_r
\label{s18}
\eeq
where $c'^r=c^r-a^r_\m c^\m$. It is readily observed that, if $\di
X=3$, the Noether identities $\ol\dl\Delta'_\m$ (\ref{s16a}) are
trivial. Then the corresponding part $c^\m\cF^r_{\la\m}\dr^\la_r$
of the gauge symmetry $u$ (\ref{s18}) also is trivial.
Consequently, the non-trivial gauge symmetry of the Chern --
Simons Lagrangian (\ref{s20}) is
\be
u=(-c^r_{pq}c'^pa^q_\la + c'^r_\la)\dr^\la_r.
\ee

\subsection{Topological BF theory}

We address the topological BF theory of two exterior forms $A$ and
$B$ of form degree $|A|+|B|=\di X-1$ on a smooth manifold $X$
\cite{birm,book09}. It is reducible degenerate Lagrangian theory
which satisfies the homology regularity condition (Condition
\ref{v155}) \cite{jmp05a}. Its dynamic variables $A$ and $B$ are
sections of a fibre bundle
\be
Y=\op\w^pT^*X\oplus \op\w^qT^*X, \qquad p+q=n-1>1,
\ee
coordinated by $(x^\la, A_{\m_1\ldots\m_p},B_{\nu_1\ldots\nu_q})$.
Without a loss of generality, let $q$ be even and $q\geq p$. The
corresponding differential graded algebra is $\cO^*_\infty Y$
(\ref{ppp}).

There are the canonical $p$- and $q$-forms
\be
&& A=A_{\m_1\ldots\m_p}dx^{\m_1}\w\cdots\w
dx^{\m_p},\\
&& B=B_{\nu_1\ldots\nu_q}dx^{\nu_1}\w\cdots\w
dx^{\nu_q}
\ee
on $Y$. A Lagrangian of topological BF theory reads
\mar{v182}\beq
L_{\rm BF}=A\w d_HB= \e^{\m_1\ldots\m_n}A_{\m_1\ldots\m_p}
d_{\m_{p+1}}B_{\mu_{p+2}\ldots\m_n}\om, \label{v182}
\eeq
where $\e$ is the Levi -- Civita symbol. It is a reduced first
order Lagrangian. Its first order Euler -- Lagrange operator
(\ref{305}) is
\mar{v183,a,b}\ben
&&\dl L= \cE_A^{\m_1\ldots\m_p}dA_{\m_1\ldots\m_p}\w\om + \cE_B^{\nu_{p+2}\ldots
\nu_n} dB_{\nu_{p+2}\ldots \nu_n}\w \om, \label{v183'}\\
&& \cE_A^{\m_1\ldots\m_p}=\e^{\m_1\ldots \m_n} d_{\m_{p+1}}B_{\mu_{p+2}\ldots\m_n},\label{v183a}\\
&&  \cE_B^{\m_{p+2}\ldots
\m_n} = - \e^{\m_1\ldots \m_n} d_{\m_{p+1}}
A_{\m_1\ldots\m_p}.\label{v183b}
\een
The corresponding Euler -- Lagrange equations can be written in a
form
\mar{wrt1}\beq
d_HB=0, \qquad d_HA=0.\label{wrt1}
\eeq
They obey the Noether identities
\mar{wrt3}\beq
d_Hd_HB=0, \qquad d_Hd_HA=0. \label{wrt3}
\eeq

One can regard the components $\cE_A^{\m_1\ldots\m_p}$
(\ref{v183a}) and $\cE_B^{\m_{p+2}\ldots \m_n}$ (\ref{v183b}) of
the Euler -- Lagrange operator (\ref{v183'}) as a $(\op\w^p
TX)\op\ot_X(\op\w^n T^*X)$-valued differential operator on the
fibre bundle $\op\w^q T^*X$ and a $(\op\w^q TX)\op\ot_X(\op\w^n
T^*X)$-valued differential operator on the fibre bundle $\op\w^p
T^*X$, respectively. They are of the same type as the
$\op\w^{n-1}TX$-valued differential operator (\ref{v183}) in
Example \ref{46e1} (cf. the equations (\ref{wrt1}) and
(\ref{iio1})). Therefore, the analysis of Noether identities of
the differential operators (\ref{v183a}) and (\ref{v183b}) is a
repetition of that of Noether identities of the operator
(\ref{v183}) (cf. the Noether identities (\ref{wrt3}) and
(\ref{iio2})).

Following Example \ref{46e1}, let us consider the family of vector
bundles
\be
&& E_k=\op\w^{p-k-1}T^*X\op\times_X \op\w^{q-k-1}T^*X, \qquad 0\leq
k< p-1, \\
&& E_k={\mathbb R} \op\times_X
\op\w^{q-p}T^*X, \qquad k=p-1, \\
&& E_k=\op\w^{q-k-1}T^*X, \quad p-1<k<q-1, \\
&& E_{q-1}=X\times \mathbb R.
\ee
Let us enlarge the differential graded algebra $\cO^*_\infty Y$ to
the BGDA $\cP_\infty^*\{q-1\}$ (\ref{w6}) which is
\mar{jnm}\beq
\cP_\infty^*\{q-1\}=\cP^*_\infty[\ol{VY}\op\oplus_Y
E_0\oplus\cdots \op\oplus_Y E_{q-1} \op\oplus_Y \ol
E_0\op\oplus_Y\cdots\op\oplus_Y \ol E_{q-1};Y]. \label{jnm}
\eeq
It possesses the local generating basis
\be
&& \{A_{\m_1\ldots\m_p}, B_{\nu_1\ldots\nu_q},
\ve_{\m_2\ldots\m_p},\ldots,\ve_{\m_p},\ve,\xi_{\nu_2\ldots\nu_q},
\ldots, \xi_{\nu_q},\xi,\\
&&\qquad \ol A^{\m_1\ldots\m_p}, \ol B^{\nu_1\ldots\nu_q},
\ol\ve^{\m_2\ldots\m_p}, \ldots,\ol\ve^{\m_p}, \ol \ve, \ol
\xi^{\nu_2\ldots\nu_q}, \ldots, \ol \xi^{\nu_q},\ol \xi\}
\ee
of Grassmann parity
\be
&& [\ve_{\m_k\ldots\m_p}]=[\xi_{\nu_k\ldots\nu_q}]=(k+1){\rm
mod}\,2, \qquad [\ve]=p\,{\rm mod}\,2, \qquad [\xi]=0,\\
&& [\ol\ve^{\m_k\ldots\m_p}]=[\ol\xi^{\nu_k\ldots\nu_q}]= k\,{\rm
mod}\,2, \qquad [\ol\ve]=(p+1){\rm mod}\,2, \qquad [\ol\xi]=1,
\ee
of ghost number
\be
{\rm gh}[\ve_{\m_k\ldots\m_p}]={\rm gh}[\xi_{\nu_k\ldots\nu_q}]=k,
\qquad {\rm gh}[\ve]=p+1, \qquad {\rm gh}[\xi]=q+1,
\ee
and of antifield number
\be
&& {\rm Ant}[\ol A^{\m_1\ldots\m_p}]={\rm Ant}[\ol
B^{\nu_{p+1}\ldots\nu_q}]=1, \\
&&  {\rm
Ant}[\ol\ve^{\m_k\ldots\m_p}]={\rm Ant}[\ol\xi^{\nu_k\ldots\nu_q}]=k+1,\\
&& {\rm Ant}[\ol\ve]=p, \qquad {\rm Ant}[\ol\ve]=q.
\ee

One can show that the homology regularity condition holds (see
Lemma \ref{v220}) and that the DBGA $\cP_\infty^*\{q-1\}$ is
endowed with the Koszul -- Tate operator
\mar{va202''}\ben
&& \dl_{\rm KT}= \frac{\rdr}{\dr \ol A^{\m_1\ldots \m_p}} \cE_A^{\m_1\ldots
\m_p} + \frac{\rdr}{\dr \ol B^{\nu_1\ldots \nu_q}}
\cE_B^{\nu_1\ldots \nu_q} + \label{va202'}\\
&& \qquad \op\sum_{2\leq k\leq p}
\frac{\rdr}{\dr \ol\ve^{\m_k\ldots \m_p}} \Delta_A^{\m_k\ldots
\m_p}+ \frac{\rdr}{\dr \ol\ve} d_{\m_p}\ol\ve^{\m_p}+\\
&& \qquad
\op\sum_{2\leq k\leq q} \frac{\rdr}{\dr\ol\xi^{\nu_k\ldots \nu_q}}
 \Delta_B^{\nu_k\ldots \nu_q}+ \frac{\rdr}{\dr \ol\xi} d_{\nu_q}\ol\xi^{\nu_q},\nonumber \\
&& \Delta_A^{\m_2\ldots \m_p}=d_{\m_1}
\ol A^{\m_1\ldots \m_p}, \qquad \Delta_A^{\m_{k+1}\ldots
\m_p}=d_{\m_k} \ol\ve^{\m_k\m_{k+1}\ldots \m_p},\qquad 2\leq k< p, \nonumber\\
&& \Delta_B^{\nu_2\ldots \nu_q}=d_{\nu_1}
\ol B^{\nu_1\ldots \nu_q}, \quad \Delta_B^{\nu_{k+1}\ldots
\nu_q}=d_{\nu_k} \ol\xi^{\nu_k\nu_{k+1}\ldots \nu_q},\quad 2\leq k
<q.\nonumber
\een
Its nilpotentness provides the complete Noether identities
(\ref{wrt1}):
\be
d_{\m_1}\cE_A^{\m_1\ldots \m_p}=0, \qquad
d_{\nu_1}\cE_B^{\nu_1\ldots \nu_q}=0,
\ee
and the $(k-1)$-stage ones
\be
&& d_{\m_k}\Delta_A^{\m_k\ldots \m_p}=0, \qquad
k=2,\ldots,p,\\
&& d_{\nu_k}\Delta_B^{\nu_k\ldots \nu_q}=0, \qquad
k=2,\ldots,q,
\ee
(cf. (\ref{v212})). It follows that the topological BF theory is
$(q-1)$-reducible.

Applying inverse second Noether Theorem \ref{w35}, one obtains the
gauge operator (\ref{w108'}) which reads
\mar{0a0}\ben
&& \bu= d_{\m_1}\ve_{\m_2\ldots\m_p}\frac{\dr}{\dr
A_{\m_1\m_2\ldots\m_p}} +
d_{\nu_1}\xi_{\nu_2\ldots\nu_q}\frac{\dr}{\dr
B_{\nu_1\nu_2\ldots\nu_q}}+ \label{0a0}\\
&& \qquad \left[d_{\m_2}\ve_{\m_3\ldots\m_p}\frac{\dr}{\dr
\ve_{\m_2\m_3\ldots\m_p}}+\cdots  +d_{\m_p}\ve\frac{\dr}{\dr
\ve_{\m_p}}\right]+ \nonumber\\
&&\qquad \left[d_{\nu_2}\xi_{\nu_3\ldots\nu_q} \frac{\dr}{\dr
\xi_{\nu_2\nu_3\ldots\nu_q}}+\cdots + d_{\nu_q}\xi\frac{\dr}{\dr
\xi_{\nu_q}}\right].\nonumber
\een
In particular,  the gauge symmetry of the Lagrangian $L_{\rm BF}$
(\ref{v182}) is
\be
u= d_{\m_1}\ve_{\m_2\ldots\m_p}\frac{\dr}{\dr
A_{\m_1\m_2\ldots\m_p}} +
d_{\nu_1}\xi_{\nu_2\ldots\nu_q}\frac{\dr}{\dr
B_{\nu_1\nu_2\ldots\nu_q}}.
\ee
This gauge symmetry is abelian. It also is readily observed that
higher-stage gauge symmetries are independent of original fields.
Consequently, topological BF theory is abelian, and its gauge
operator $\bu$ (\ref{0a0}) is nilpotent. Thus, it is the BRST
operator $\bbc=\bu$. As a result, the Lagrangian $L_{\rm BF}$ is
extended to the proper solution of the master equation $L_E=L_e$
(\ref{w8}) which reads
\be
&& L_e=L_{\rm BF} + \ve_{\m_2\ldots\m_p}d_{\m_1}\ol
A^{\m_1\ldots\m_p} + \op\sum_{1<k<p}\ve_{\m_{k+1}\ldots\m_p}
d_{\m_k}\ol \ve^{\m_k\ldots\m_p}+\ve d_{\m_p}\ol\ve^{\m_p}+\\
&& \qquad \xi_{\nu_2\ldots\nu_q}d_{\nu_1}\ol
B^{\nu_1\ldots\nu_q} + \op\sum_{1<k<q}\xi_{\nu_{k+1}\ldots\nu_q}
d_{\nu_k}\ol \xi^{\nu_k\ldots\m_q}+\xi d_{\nu_q}\ol\xi^{\nu_q}.
\ee


\begin{thebibliography}{ederf}


\bibitem{and} I. Anderson and T. Duchamp, On the existence of
global variational principles, \emph{Amer. J. Math.} \textbf{102}
(1980) 781.

\bibitem{ander} I. Anderson, Introduction to the variational
bicomplex, \emph{Contemp. Math.} \textbf{132} (1992) 51.

\bibitem{barn} G. Barnich, F. Brandt, M. Henneaux, Local
BRST cohomology in gauge theories, \emph{Phys. Rep.} \textbf{338}
(2000) 439.

\bibitem{bart} C. Bartocci, U. Bruzzo and D. Hern\'andez
Ruip\'erez, \emph{The Geometry of Supermanifolds} (Kluwer, 1991).

\bibitem{ijgmmp05} D. Bashkirov and G. Sardanashvily, On the BV quantization of gauge
gravitation theory, \emph{Int. J. Geom. Methods Mod. Phys.}
\textbf{2} (2005) 203; \emph{arXiv}: hep-th/0501254.

\bibitem{jmp05} D. Bashkirov, G. Giachetta, L. Mangiarotti and
G. Sardanashvily, Noether's second theorem for BRST symmetries,
\emph{J. Math. Phys.} \textbf{46} (2005) 053517; \emph{arXiv}:
math-ph/0412034.

\bibitem{jpa05} D. Bashkirov, G. Giachetta, L. Mangiarotti and
G. Sardanashvily, Noether's second theorem in a general setting.
Reducible gauge theorie, \emph{J. Phys. A} \textbf{38} (2005)
5329; \emph{arXiv}: math-ph/0411070.


\bibitem{jmp05a} D. Bashkirov, G. Giachetta, L. Mangiarotti and
G. Sardanashvily, The antifield Koszul -- Tate complex of
reducible Noether identities, \emph{J. Math. Phys.} \textbf{46}
(2005) 103513; \emph{arXiv}: math-ph/0506034.

\bibitem{lmp08} D. Bashkirov, G. Giachetta, L. Mangiarotti and
G. Sardanashvily, The KT-BRST complex of degenerate Lagrangian
systems, \emph{Lett. Math. Phys.} \textbf{83} (2008) 237;
\emph{arXiv}: math-ph/0702097.


\bibitem{batch1} M. Batchelor, The structure of supermanifolds, \emph{Trans.
Amer. Math. Soc.} \textbf{253} (1979) 329.

\bibitem{bau} M. Bauderon, Differential geometry and Lagrangian formalism in
the calculus of variations, In: \emph{Differential Geometry,
Calculus of Variations, and their Applications}, Lecture Notes in
Pure and Applied Mathematics. \textbf{100} (Dekker, New York,
1985) p. 67.

\bibitem{birm} D. Birmingham and M. Blau, Topological field theory,
\emph{Phys. Rep.} \textbf{209} (1991) 129.


\bibitem{bor07} A. Borowoiec, L. Fatibene, M. Ferraris and
S. Mercadante, Covariant Lagrangian formulation of Chern -- Simons
theories, \emph{Int. J. Geom. Methods Mod. Phys.} \textbf{4}
(2007) 277.

\bibitem{bran01} F. Brandt, Jet coordinates for local BRST
cohomology, \emph{Lett. Math. Phys.} \textbf{55} (2001) 149.

\bibitem{bred} G. Bredon, \emph{Sheaf theory} (McGraw-Hill, New York, 1967).

\bibitem{bruz} U. Bruzzo, The global Utiyama theorem in Einstein -- Cartan
theory, \emph{J. Math. Phys.} \textbf{\bf 28} (1987) 2074.

\bibitem{bry} R. Bryant, S. Chern, R. Gardner, H. Goldschmidt, P. Griffiths,
\emph{Exterior Differential Systems} (Springer, 1991).


\bibitem{cari03} J. Cari\~nena, H. Figueroa,  Singular Lagrangian
in supermechanics, \emph{Diff. Geom. Appl.} \textbf{18} (2003) 33.

\bibitem{cia95} R. Cianci, M. Francaviglia, I. Volovich,
Variational calculus and Poincar\'e -- Cartan formalism in
supermanifolds, \emph{J. Phys. A.} \textbf{28} (1995) 723.

\bibitem{egu} T. Eguchi, P. Gilkey and A. Hanson, Gravitation, gauge theories
and differential geometry, \emph{Phys. Rep.} \textbf{66} (1980)
213.

\bibitem{fat94} L. Fatibene, M. Ferraris and M. Francaviglia,
N\"other formalism for conserved quantities in classical gauge
field theories, \emph{J. Math. Phys.} \textbf{35} (1994) 1644.

\bibitem{fat} L. Fatibene, M. Ferraris, M. Francaviglia
and R. McLenaghan, Generalized symmetries in mechanics and field
theories, \emph{J. Math. Phys.} \textbf{43} (2002) 3147.

\bibitem{fisch} J. Fisch and M. Henneaux,  Homological
perturbation theory and algebraic structure of the
antifield-antibracket formalism for gauge theories, \emph{Commun.
Math. Phys.} \textbf{128} (1990) 627.

\bibitem{franc} D. Franco, C. Polito, Supersymmetric
field-theoretic models on a supermanifold, \emph{J. Math. Phys.}
\textbf{45} (2004) 1447.

\bibitem{fuks} D. Fuks, \emph{Cohomology of Infinite-Dimensional Lie
Algebras} (Consultants Bureau, New York, 1986).

\bibitem{fulp02} R. Fulp, T. Lada and J. Stasheff, Sh-Lie
algebras induced by gauge transformations, \emph{Commun. Math.
Phys.} \textbf{231} (2002) 25.

\bibitem{fulp} R. Fulp, T. Lada, and J. Stasheff, Noether
variational Theorem II and the BV formalism, \emph{Rend. Circ.
Mat. Palermo (2) Suppl.} No. 71 (2003)  115.


\bibitem{hern} D. Hern\'andez Ruip\'erez and J. Mu\~noz
Masqu\'e, Global variational calculus on graded manifolds,
\emph{J. Math. Pures Appl.} \textbf{63} (1984) 283.

\bibitem{giacqg} G. Giachetta and G. Sardanashvily,
Stress-energy-momentum of affine-metric gravity. Generalized Komar
superportential, \emph{Class. Quant. Grav.} \textbf{13} (1996)
L67; \emph{arXiv}: gr-qc/9511008.

\bibitem{book}  G. Giachetta, L. Mangiarotti and G. Sardanashvily, \emph{New
Lagrangian and Hamiltonian Methods in Field Theory} (World
Scientific, 1997).

\bibitem{jmp} G. Giachetta, L. Mangiarotti and
G. Sardanashvily,  Cohomology of the infinite-order jet space and
the inverse problem, \emph{J. Math. Phys.} \textbf{42} (2001)
4272; \emph{arXiv}: math/0006074.

\bibitem{mpl} G. Giachetta, L. Mangiarotti and
G. Sardanashvily, Noether conservation laws in higher-dimensional
Chern -- Simons theory, \emph{Mod. Phys. Lett. A} \textbf{18}
(2003) 2645.

\bibitem{book05} G. Giachetta, L. Mangiarotti, and G. Sardanashvily,
\emph{Geometric and Algebraic Topological Methods in Quantum
Mechanics} (World Scientific,  2005).

\bibitem{cmp04} G. Giachetta, L. Mangiarotti and G. Sardanashvily,
Lagrangian supersymmetries depending on derivatives. Global
analysis and cohomology, \emph{Commun. Math. Phys.} \textbf{259}
(2005) 103; \emph{arXiv}: math/0305303.

\bibitem{jmp09} G. Giachetta, L. Mangiarotti and G. Sardanashvily,
 On the notion of gauge symmetries of generic Lagrangian field
theory, \emph{J. Math. Phys} \textbf{50} (2009) 012903;
\emph{arXiv}: 0807.3003.

\bibitem{book09} G. Giachetta, L. Mangiarotti, and G. Sardanashvily,
\emph{Advanced Classical Field Theory} (World Scientific, 2009).

\bibitem{gom} J. Gomis, J. Par\'\i s and S. Samuel,
Antibracket, antifields and gauge theory quantization, \emph{Phys.
Rep.} \textbf{295} (1995) 1.

\bibitem{got92} M. Gotay and J. Marsden, Stress-energy-momentum
tensors and the Belinfante -- Rosenfeld formula, \emph{Contemp.
Math.} \textbf{132} (1992) 367.

\bibitem{gron} F. Gronwald, BRST antifield treatment of
metric-affine gravity, \emph{Phys. Rev. D} \textbf{57} (1998) 961.


\bibitem{hehl} F. Hehl, J. McCrea, E. Mielke and Y. Ne'eman, Metric-affine
gauge theory of gravity: field equations, Noether identities,
world spinors, and breaking of dilaton invariance, \emph{Phys.
Rep.} \textbf{258} (1995) 1.


\bibitem{hir} F. Hirzebruch, \emph{Topological Methods in Algebraic Geometry}
(Springer, 1966).

\bibitem{ibr} N. Ibragimov, \emph{Transformation Groups Applied
to Mathematical Physics} (Riedel, Boston, 1985).

\bibitem{iva} D. Ivanenko and G. Sardanashvily,
The gauge treatment of gravity, \emph{Phys. Rep.} \textbf{94}
(1983) 1.

\bibitem{jad} A. Jadczyk and K. Pilch, Superspaces and Supersymmetries,
\emph{Commun. Math. Phys.} \textbf{78} (1981) 391.

\bibitem{julia} B. Julia and S. Silva, Currents and superpotentials
in classical gauge inveriant theories. Local results with
applications to perfect fluids and General Relativity,
\emph{Class. Quant. Grav.} \textbf{15} (1998) 2173.

\bibitem{kol} I. Kol\'a\v{r}, P. Michor and J. Slov\'ak, \emph{Natural Operations
in Differential Geometry} (Springer, 1993).

\bibitem{KS} Y. Kosmann-Schwarzbach, \emph{The Noether Theorems.
Invariance and the Conservation Laws in the Twentieth Century}
(Springer, 2011).

\bibitem{kras} I. Krasil'shchik, V. Lychagin, A. Vinogradov, \emph{Geometry of
Jet Spaces and Nonlinear Partial Differential Equations} (Gordon
and Breach, Glasgow, 1985).


\bibitem{book00} L. Mangiarotti and G. Sardanashvily,  \emph{Connections in
Classical and Quantum Field Theory} (World Scientific, 2000).

\bibitem{mont92} J. Monterde, J. Munos Masque, Variational
problems on graded manifolds, \emph{Contemp. Math.} \textbf{132}
(1992) 551.

\bibitem{mont} J. Monterde, J. Munos Masque, J. Vallejo, The
Poincare -- Cartan form in superfield theory, \emph{Int. J. Geom.
Methods Mod. Phys.} \textbf{3} (2006) 775.

\bibitem{olv} P. Olver, \emph{Applications of Lie Groups to
Differential Equations} (Springer, 1986).

\bibitem{ren} A. Rennie, Smoothness and locality for nonunital
spectral triples, \emph{K-Theory} \textbf{28} (2003) 127.

\bibitem{sard97} G. Sardanashvily, Stress-energy-momentum conservation law in
gauge gravitation theory, \emph{Class. Quant. Grav.} \textbf{14}
(1997) 1371.

\bibitem{sard01} G. Sardanashvily, Remark on the Serre-Swan theorem for non-compact
manifolds; \emph{arXiv}: math-ph/0102016.

\bibitem{cs} G. Sardanashvily, Gauge conservation laws in higher-dimensional Chern-Simons
models, \emph{arXiv}: hep-th/0303059.

\bibitem{oper} G. Sardanashvily, Noether identities of a differential operator.
The Koszul -- Tate complex, \emph{Int. J. Geom. Methods Mod.
Phys.} \textbf{2} (2005) 873, \emph{arXiv}: math.DG/0506103.

\bibitem{sard06} G. Sardanashvily, Preface. Gauge gravitation
theory from geometric viewpoint, \emph{Int. J. Geom. Methods Mod.
Phys.} \textbf{3} (2006) No. 1; \emph{arXiv}: gr-qc/0512115.

\bibitem{ijgmmp07} G. Sardanashvily, Graded infinite order jet manifolds,
\emph{Int. J. Geom. Methods Mod. Phys.} \textbf{4}, (2007) 1335;
\emph{arXiv}: 0708.2434.

\bibitem{sard08} G. Sardanashvily, Classical field theory.
Advanced mathematical formulation, \emph{Int. J. Geom. Methods
Mod. Phys.} \textbf{5} (2008) 1163; \emph{arXiv}: 0811.0331.


\bibitem{gauge09} G. Sardanashvily, Gauge conservation laws in a general
setting. Superpotential, \emph{Int. J. Geom. Methods Mod. Phys.}
\textbf{6} (2009) 1047; \emph{arXiv}: 0906.1732.

\bibitem{sard09} G. Sardanashvily, Lectures on supergeometry,
\emph{arXiv}: 0910.0092.

\bibitem{sard11} G. Sardanashvily, Classical gauge gravitation
theory, \emph{Int. J. Geom. Methods Mod. Phys.} \textbf{8} (2011)
1869; \emph{arXiv}: 1110.1176.

\bibitem{book12} G. Sardanashvily, \emph{Lectures on Differential
Geometry of Modules and Rings. Application to Quantum Theory}
(Lambert Academic Publishing, Saarbrucken, 2012); \emph{arXiv}:
0910.1515.

\bibitem{book13} G. Sardanashvily, \emph{Advanced Differential Geometry for Theoreticians.
Fiber bundles, jet manifolds and Lagrangian theory} (Lambert
Academic Publishing, Saarbrucken, 2013); \emph{arXiv}: 0908.1886.

\bibitem{sard13} G. Sardanashvily, Graded Lagrangian formalism,
\emph{Int. J. Geom. Methods. Mod. Phys.} \textbf{10} (2013)
1350016; \emph{arXiv}: 1206.2508.

\bibitem{SS} G. Sardanashvily, Remark on the Serre -- Swan theorem
for graded manifolds, \emph{arXiv}: 1304.1371.

\bibitem{sard14} G. Sardanashvily, W. Wachowski, SUSY gauge theory
on graded manifolds, \emph{arXiv}: 1406.6318.

\bibitem{sau} D. Saunders, \emph{The Geometry of Jet Bundles}
(Cambridge Univ. Press, Cambridge, 1989).

\bibitem{stavr} T. Stavracou, Theory of connections on graded principal
bundles, \emph{Rev. Math. Phys.} \textbf{10} (1998) 47.


\bibitem{tak2} F. Takens, A global version of the inverse problem of
the calculus of variations, \emph{J. Diff. Geom.} \textbf{14}
(1979) 543.

\bibitem{ten} B. Tennison, \emph{Sheaf Theory} (Cambridge Univ.
Press, 1975).

\bibitem{terng} C. Terng, Natural vector bundles and natural
differential operators, \emph{American J. Math.} \textbf{100}
(1978) 775.


\end{thebibliography}
\end{document}